\documentclass[pra,preprint,notitlepage,nofootinbib]{revtex4-1}
\usepackage[utf8]{inputenc}
\usepackage[T1]{fontenc}
\usepackage{amsmath, amsthm, amssymb, mathtools, dsfont}  
\usepackage{fixmath}
\usepackage{times}
\usepackage{graphicx}  
\usepackage{dcolumn} 
\usepackage{bm,bbm}
\usepackage[normalem]{ulem}
\usepackage[usenames,dvipsnames]{xcolor}
\usepackage{esint}
\usepackage{paralist}
\usepackage{appendix}
\usepackage{geometry}
\definecolor{myurlcolor}{rgb}{0,0,0.7}
\definecolor{myrefcolor}{rgb}{0.8,0,0}
\usepackage[unicode=true,pdfusetitle,  bookmarks=true,bookmarksnumbered=false,bookmarksopen=false, breaklinks=false,pdfborder={0 0 0},backref=false,colorlinks=true, linkcolor=myrefcolor,citecolor=myurlcolor,urlcolor=myurlcolor]{hyperref}
\usepackage{etoolbox}
\makeatletter
\patchcmd{\frontmatter@abstract@produce}
  {\vskip200\p@\@plus1fil
   \penalty-200\relax
   \vskip-200\p@\@plus-1fil}
  {}
  {}
  {}
 \theoremstyle{plain}
 
 \theoremstyle{plain}
 \newtheorem{lem}{Lemma}
 \theoremstyle{plain}
 \newtheorem{thm}{Theorem}
 \theoremstyle{plain}
 
 \theoremstyle{plain}
 \newtheorem{corr}{Corollary}
 \theoremstyle{plain}
 
 \theoremstyle{remark}
 \newtheorem*{rem*}{Remark}
 \theoremstyle{plain}
  \newtheorem{rem}{Remark}

\theoremstyle{plain}
\newtheorem{schm}{Scheme}
\theoremstyle{plain}

\newtheorem{defn}{Definition}
\theoremstyle{plain}

\newcommand*{\textlabel}[2]{%
  \edef\@currentlabel{#1}
  \phantomsection
  #1\label{#2}
}

\newcommand{\supp}{\mathrm{supp}} 
\newcommand{\range}{\mathrm{range}} 
\newcommand{\e}{\mathrm{e}}

\newcommand{\unity}{\mathbbmss{1}}

\DeclareMathOperator{\tr}{tr}
\newcommand{\Tr}{\mathrm{Tr}}

\renewcommand{\C}{\mathbb{C}} 
\renewcommand{\U}{\mathrm{U}} 
\newcommand{\SL}{\mathrm{SL}(2,2^m)} 
\newcommand{\Pmb}{\overline{\mathcal{P}}_m} 
\newcommand{\Pm}{\mathcal{P}_m} 
\newcommand{\Pmbb}{\hat{\mathcal{P}}_m} 
\newcommand{\Cl}{\mathcal{C}_m}
\newcommand{\Ua}[1]{\mathcal{U}_{#1}}
\newcommand{\op}[1]{\left|\left|  \ #1  \ \right| \right|_{\mathrm{op}}}

\definecolor{dukeblue}{rgb}{0.0, 0.0, 0.61}
\definecolor{cadmiumgreen}{rgb}{0.0, 0.42, 0.24}
\newcommand{\tanmayc}[1]{{\color{dukeblue}1}}
\newcommand{\tanmay}[1]{{\color{cadmiumgreen} 1}}
\newcommand{\BCg}{\mathcal{B}\left( {\C^{N}}^{\otimes t} \right)}
\newcommand{\BCt}{\mathcal{B}\left( {\C^{N}}^{\otimes 2} \right)}
\newcommand{\BCr}{\mathcal{B}\left( {\C^{N}}^{\otimes 3} \right)}
\renewcommand{\H}{\mathcal{G}_{H}}

\newcommand{\EP}{\mathcal{G}_{\mathcal{P}}}

\newcommand{\ET}{\mathcal{G}_{\mathcal{T}}}

\newcommand{\Sp}{\mathrm{Sp}\left(2m, \mathbb{F}_2 \right)}
\newcommand{\id}{\mathrm{id}}
\newcommand{\tncone}{\mathrm{T}_{1}^{(\mathrm{nc})}}
\newcommand{\tnctwo}{\mathrm{T}_{2}^{(\mathrm{nc})}}
\newcommand{\tncthree}{\mathrm{T}_{3}^{(\mathrm{nc})}}
\newcommand{\tnconem}{\mathrm{T}_{1,-}^{(\mathrm{nc})}}
\newcommand{\tnconep}{\mathrm{T}_{1,+}^{(\mathrm{nc})}}
\newcommand{\tnconepm}{\mathrm{T}_{1,\pm}^{(\mathrm{nc})}}
\newcommand{\tnctwom}{\mathrm{T}_{2,-}^{(\mathrm{nc})}}
\newcommand{\tnctwop}{\mathrm{T}_{2,+}^{(\mathrm{nc})}}
\newcommand{\tnctwopm}{\mathrm{T}_{2,\pm}^{(\mathrm{nc})}}
\newcommand{\tncthreem}{\mathrm{T}_{3,-}^{(\mathrm{nc})}}
\newcommand{\tncthreep}{\mathrm{T}_{3,+}^{(\mathrm{nc})}}
\newcommand{\tncthreepm}{\mathrm{T}_{3,\pm}^{(\mathrm{nc})}}

\newcommand{\vzero}{\mathcal{V}_0}
\newcommand{\vdthree}{\mathcal{V}_{3}^{(\mathrm{d})}}
\newcommand{\vdtwo}{\mathcal{V}_{2}^{(\mathrm{d})}}
\newcommand{\vdone}{\mathcal{V}_{1}^{(\mathrm{d})}}
\newcommand{\vnc}{\mathcal{V}^{(\mathrm{nc})}}
\newcommand{\vc}{\mathcal{V}^{(\mathrm{c})}}
\newcommand{\vncones}{\mathcal{V}_{1,+}^{(\mathrm{(nc)}}}
\newcommand{\w}[2]{\ket{\hat{#1};{#2}}}
\newcommand{\q}[1]{\ket{\hat{#1}}}

\newcommand{\g}[2]{\ket{\Bar{#1};#2}}

\newcommand{\wtr}{\mathcal{W}_{\mathrm{CS}}}
\newcommand{\wpp}{\mathcal{W}_{+}}
\newcommand{\wm}{\mathcal{W}_{-}}
\newcommand{\wpm}{\mathcal{W}_{\pm}}

\newcommand{\intH}{{\int} \, d\mu(U) \,}

\newcommand{\ETT}{\hat{\ET}}

\newcommand{\ETTps}{{\ETT}_{+,0}}
\newcommand{\ETTpw}{{\ETT}_{+,1}}
\newcommand{\ETTpws}{{\ETT}_{+,2}}

\newcommand{\F}{\mathbb{F}_{2}} 
\newcommand{\FF}{\mathbb{F}_{2^m}} 
\newcommand{\Fs}{{\F}_*^{2m}}
\newcommand{\Fm}{\F^{2m}}
\newcommand{\Fl}{\mathbb{F}_{2,\mathrm{L}}^{2m}}
\newcommand{\Fr}{\mathbb{F}_{2,\mathrm{R}}^{2m}}

\newcommand{\pcs}{\mathcal{P}_{\mathrm{CS}}}
\newcommand{\pas}{\mathcal{P}_{\mathrm{AS}}}
\newcommand{\ps}{\mathcal{P}_{\mathrm{S}}}
\newcommand{\pw}{\mathcal{P}_{\omega}}
\newcommand{\pws}{\mathcal{P}_{\omega^*}}
\newcommand{\pj}[1]{\mathcal{P}_{#1}}
\newcommand{\pdp}{\mathcal{P}_{\mathrm{d},+}}
\newcommand{\pdm}{\mathcal{P}_{\mathrm{d},-}}
\newcommand{\lam}{\mathrm{\Lambda}}

\newcommand{\muv}{\mu}
\newcommand{\nuv}{\nu}

\newcommand{\Sr}{\mathrm{S}_3} 

\newcommand{\W}[1]{\mathrm{W}_{#1}}

\newcommand{\Wa}[1]{\mathcal{W}_{#1}} 

\newcommand{\sgn}[1]{\mathrm{sgn}\left(\pi\right)}

\newcommand{\s}[2]{\left\langle #1 , #2  \right\rangle}
\newcommand{\Om}{\mathrm{\Omega}}

\global\long\global\long\global\long\def\bra#1{\mbox{\ensuremath{\langle#1|}}}
\global\long\global\long\global\long\def\ket#1{\mbox{\ensuremath{|#1\rangle}}}
\global\long\global\long\global\long\def\bk#1#2{\mbox{\ensuremath{\ensuremath{\langle#1|#2\rangle}}}}
\global\long\global\long\global\long\def\kb#1#2{\mbox{\ensuremath{\ensuremath{\ensuremath{|#1\rangle\!\langle#2|}}}}}

\newcommand{\kt}[3]{\left| #1 \; \;\;\;\; #2 \;\;\;\;\; #3 \right\rangle}

\newcommand{\ptr}{\mathcal{P}_{\mathrm{tr}}}
\newcommand{\pzp}{\mathcal{P}_{0,+}}
\newcommand{\pzm}{\mathcal{P}_{0,-}}

\newcommand{\pop}{\mathcal{P}_{1,+}}
\newcommand{\pom}{\mathcal{P}_{1,-}}

\newcommand{\ptp}{\mathcal{P}_{2,+}}
\newcommand{\ptm}{\mathcal{P}_{2,-}}

\newcommand{\pjpm}{\mathcal{P}_{j,\pm}}

\newcommand{\vzp}{\mathcal{V}_{0,+}}
\newcommand{\vzm}{\mathcal{V}_{0,-}}

\newcommand{\vop}{\mathcal{V}_{1,+}}
\newcommand{\vom}{\mathcal{V}_{1,-}}

\newcommand{\vtp}{\mathcal{V}_{2,+}}
\newcommand{\vtm}{\mathcal{V}_{2,-}}

\newcommand{\vjpm}{\mathcal{V}_{k,\pm}}

\newcommand{\ETTCS}{\ETT_{\left[\mathrm{S}\right]}}

\newcommand{\ETTd}{\ETT_{\mathrm{D}}}
\newcommand{\ETTpjk}{{\ETT}_{\mathrm{D},k}}

\newcommand{\vcone}{\mathcal{V}_{1}^{(\mathrm{(c)}}}

\newcommand{\vdcone}{{\mathcal{V}_{1}^{(\mathrm{c})}}}
\newcommand{\vqcone}{{\mathcal{V}_{1,+}^{(\mathrm{c})}}}
\newcommand{\vrcone}{{\mathcal{V}_{1,-}^{(\mathrm{c})}}}

\newcommand{\pdcone}{{{\mathcal{P}_{1}^{(\mathrm{c})}}}}
\newcommand{\pqcone}{{{\mathcal{P}_{1,+}^{(\mathrm{c})}}}}
\newcommand{\prcone}{{{\mathcal{P}_{1,-}^{(\mathrm{c})}}}}

\newcommand{\pdctwo}{{{\mathcal{P}_{2}^{(\mathrm{c})}}}}

\newcommand{\pdcthree}{{{\mathcal{P}_{3}^{(\mathrm{c})}}}}

\newcommand{\tcone}{\mathrm{T}_{1}^{(\mathrm{c})}}
\newcommand{\tctwo}{\mathrm{T}_{2}^{(\mathrm{c})}}
\newcommand{\tcthree}{\mathrm{T}_{3}^{(\mathrm{c})}}
\newcommand{\tconem}{\mathrm{T}_{1,-}^{(\mathrm{c})}}
\newcommand{\tconep}{\mathrm{T}_{1,\mathrm{D}}^{(\mathrm{c})}}

\newcommand{\tctwom}{\mathrm{T}_{2,-}^{(\mathrm{c})}}
\newcommand{\tctwop}{\mathrm{T}_{2,\mathrm{D}}^{(\mathrm{c})}}

\newcommand{\tcthreem}{\mathrm{T}_{3,-}^{(\mathrm{c})}}
\newcommand{\tcthreep}{\mathrm{T}_{3,\mathrm{D}}^{(\mathrm{c})}}

\makeatother

\newcommand{\mh}[1]{\textcolor{red}{}}
\begin{document}    
\title{Approximate $3$-designs and partial decomposition of the Clifford group representation using transvections}
     
        \author{Tanmay Singal}
    \email{tanmaysingal@gmail.com}
    \affiliation{Institute of Physics, Faculty of Physics, Astronomy and Informatics,
Nicolaus Copernicus University, \\ Grudziadzka 5/7, 87-100 Toru\'n, Poland \\ Physics Division, National Center for Theoretical Sciences, Taipei 10617, Taiwan \\ Department of Physics, National Taiwan University, No.1 Sec. 4., Roosevelt Road, Taipei 106, Taiwan
 \\ Department of Analysis, Budapest University of Technology and Economics,
1111 Budapest, Egry József u. 1.}

    \author{Min-Hsiu Hsieh}
    \email{min-hsiu.hsieh@foxconn.com}
    \affiliation{Hon Hai Quantum Computing Research Center, Taipei, Taiwan}

\begin{abstract}
We study a scheme to implement an asymptotic unitary $3$-design. The scheme implements a random Pauli once followed by the implementation of a random transvection Clifford by using state twirling. Thus the scheme is implemented in the form of a quantum channel. We show that when this scheme is implemented $k$ times, then, in the $k \rightarrow \infty$ limit, the overall scheme implements a unitary $3$-design. This is proved by studying the eigendecomposition of the scheme: the $+1$ eigenspace of the scheme coincides with that of an exact unitary $3$-design, and the remaining eigenvalues are bounded by a constant. Using this we prove that the scheme has to be implemented approximately $\mathcal{O}(m \  + \  \log 1/\epsilon)$ times to obtain an $\epsilon$-approximate unitary $3$-design, where $m$ is the number of qubits, and $\epsilon$ is the diamond-norm distance of the exact unitary $3$-design.  Also, the scheme implements an asymptotic unitary $2$-design with the following convergence rate: it has to be sampled $\mathcal{O}(\log 1/\epsilon)$ times to be an $\epsilon$-approximate unitary $2$-design. Since transvection Cliffords are a conjugacy class of the Clifford group, the eigenspaces of the scheme's quantum channel coincide with the irreducible invariant subspaces of the adjoint representation of the Clifford group. Some of the subrepresentations we obtain are the same as were obtained in \cite{Helsen2016}, whereas the remaining are new invariant subspaces. Thus we obtain a partial decomposition of the adjoint representation for $3$ copies for the Clifford group. Thus, aside from providing a scheme for the implementation of unitary $3$-design, this work is of interest for studying representation theory of the Clifford group, and the potential applications of this topic. The paper ends with open questions regarding the scheme and representation theory of the Clifford group.
\end{abstract}
\maketitle  
\thispagestyle{plain}

\section{Introduction}
\label{sec:intro}
The act of randomisation by the Haar-measure of the $m$-qubit unitary group is an important operation in quantum computation and information. It has found applications in, 
for instance, randomized benchmarking \cite{RB1,RB2,RB3,RB4,RB5,RB6,RB7,RB8,RB9,RB10,RB11}, quantum control, quantum data hiding \cite{DiVincenzo2002}, and channel twirling \cite{Dankert2009}, to name a few. Besides practical applications, it plays a significant role in decoupling theorems, which is of immense significance in quantum information theory. It has also found applications in more general fields of physics like the black hole information paradox \cite{BH1,BH2,BH3} and thermalisation of isolated systems \cite{TH1,TH2,TH3}. For various purposes 
one is only interested in the $t$-th order moments of the Haar-measure, which can be achieved by randomisation over a finite ensemble of unitaries $\left\{\left(p_i,U_i \right)\right\}_{i=1}^n$ on the unitary group. Any finite ensemble of unitaries which reproduces these $t$-th order moments is called a unitary $t$-design. There are many equivalent formulations of unitary $t$-designs (see, e.g., \cite{Cleve16}, and the references therein). An approximate unitary $t$-design is an ensemble of unitaries which approximates the $t$-th order moments of the Haar-measure of the unitary group. One is often interested in asymptotic unitary $t$-designs, which are iterative schemes to obtain an exact design in the asymptotic limit. Given that such an iterative scheme converges, one typically asks how fast it converges to its limit. Various schemes to implement unitary $t$-designs are known in the literature, e.g. \cite{Cleve16, Zhu15,Webb15, RB11}.


For qubit systems, randomly sampling from the $m$-qubit Clifford group $\Cl$ gives an exact unitary $3$-design \cite{Zhu15,Webb15}. The order of the Clifford group is $|\Cl| = \mathcal{O} \left( 2^{m^2} \right)$. 
The Clifford group $\Cl$ is the normaliser of the $m$-qubit Pauli group $\Pmbb$ in the unitary group $\U(2^m)$. Arguably the simplest (non-Pauli) Clifford unitaries are known as transvections Clifford unitaries \cite{misc_transvections1, misc_transvections2, misc_transvections3}. While transvections have been studied earlier \cite{TR1,TR2}, more interest has been shown to them recently \cite{ Koenig14, misc_transvections1, misc_transvections2, misc_transvections3, Tan20, Pllaha21}. Such applications are based on two properties of transvections: (i) that they generate the Clifford group, and (ii) they are a conjugacy class of the Clifford group, and moreover are non-Pauli Cliffords with minimal Pauli support. This finds applications in designing fault-tolerant gadgets \cite{Pllaha21,misc_transvections1}. Also, transvections gates are particularly suited for ion-trap systems. Specifically, two-qubit transvections are suited to ion-trap architectures (see for instance, in Section B of \cite{ion_trap1}, and \cite{ion_trap2}). Also, multi-qubit gates can be performed on ion-traps efficiently \cite{ion_trap3,ion_trap4,ion_trap5,ion_trap6}. 

A recent paper \cite{Tan20}  gives a scheme for an asymptotic unitary $3$-design: it implements a random Pauli, followed by a unitary $2$-design, after which the scheme iteratively applies uniformly random transvections to interpolate the gap between the two-design and the three-design. The size of the set being sampled from reduces from {$\mathcal{O}\left(2^{m^2} \right)$}
(for the exact unitary $3$-design by using the Clifford group) to $\mathcal{O}\left( 2^m \right)$, although the need to sample exactly from the uniform ensemble of the $\SL$-subgroup requires sampling from another set of unitaries of size $\mathcal{O}\left(2^{m}\right)$. The scheme converges to a $3$-design and the authors provide the convergence rate $\mathcal{O}\left(5 \log m + \log 1/\epsilon \right)$. After a detailed study of the paper, we found that the convergence rate of their scheme is closer to $\mathcal{O}\left( \frac{3}{2} m + \log 1/\epsilon \right)$.


In this work we present a scheme to implement an asymptotic unitary $3$-design by directly randomising over transvections. The scheme employed will be as follows.
\begin{itemize}
    \item[(i)] Uniform sampling over the Pauli group.
    \item[(ii)] Uniform sampling over a set of transvection Cliffords. This step has to be performed iteratively.
\end{itemize} 
That the second step is iterative means that it has to be implemented a certain number of times (say, $k$) to converge to a unitary $3$-design. In the asymptotic $k \rightarrow \infty$ limit, the scheme gives us an exact unitary $3$-design. For any such asymptotic scheme, one is interested in the rate of convergence. 


The iterative scheme amounts to applying a Markov chain on the group algebra of $\Cl$. Furthermore, since transvections are generators of the Clifford group, this Markov chain is irreducible. By the Perron-Frobenius theorem one is guaranteed that the scheme will converge to a uniform Clifford design \cite{Diaconis2003}. As we mentioned above, a uniform Clifford design is an exact $3$-design. Hence, asymptotically, our scheme gives a unitary $3$-design.

\subsection{Main results: Overview}
We prove the following results.
\begin{itemize}
    \item[(i)] The rate of convergence for the unitary $3$-design is $k= \mathcal{O} \left(m \ {+} \ \log 1/\epsilon \right)$. 
    \item[(ii)] The rate of convergence for the asymptotic unitary $2$-design is $k= \mathcal{O} \left( \log \  1/\epsilon \right)$. 
    \item[(iii)] Partial decomposition of the three-copy adjoint representation of the Clifford group.
\end{itemize}
One can anticipate the optimal rates of convergence for any asymptotic unitary design as follows: since a minimal unitary $2$-design has $\mathcal{O}\left(2^{m}\right)$ elements \cite{Gross2006}, one anticipates that any convergent and iterative scheme (like ours and also the scheme in \cite{Tan20}), which randomly samples among $2^n$ elements during a single step, and such that there is negligible overlap between the cumulatively sampled elements after any number of steps, will converge to the minimal design in $\mathcal{O}\left( m/n \right)$ iterations. Similarly since the Clifford group has $\mathcal{O}\left(2^{m^2}\right)$ elements, one anticipates that the optimal convergence rate to be $\mathcal{O}\left( m^2/n \right)$. For $n=\mathcal{O}\left(m\right)$, Theorems \ref{thm:asymptotic_unitary_2_design} and \ref{thm:asymptotic_unitary_3_design} are consistent with this anticipation. See Remark \ref{rem:optimal_convergence_rate} for further elaboration of this point, and also see \cite{Pitt_Coste_unpublished}).


While efficient construction of quantum circuits for Clifford unitaries are known and well-studied \cite{Aaronson2008,Koenig14,Gidney2021,Bravyi_Maslov_21,Bravyi21}, one is still interested in studying new schemes adapted to different architectures. For instance, for most architectures, quantum circuits are compiled with single-qubit and two-qubit gates. The reason being that multiqubit gates are still quite noisy on these architectures. Transvection Cliffords are generally multiqubit Cliffords. Thus the implementation of our scheme will require a circuit compilation for transvections into single-qubit and two-qubit gates on such aforesaid architectures. That being said, in recent years multi-qubit gates are also being experimentally realised, particularly so for ion-trap quantum computers \cite{Hahn2019,Shapira2020,Rasmussen2020}. This makes it conducive to study various schemes adapted to multiqubit gates.   

Aside from the practical reasons mentioned above, this scheme offers a useful tool to study the representation theory of the Clifford group. In particular, we found that the eigenspaces of the $3$-copy adjoint representation of the Clifford group are also eigenspaces of the quantum channel which represents the scheme. In \cite{Helsen2016} Helsen et. al. gave a complete decomposition of the $2$-copy adjoint representation of the Clifford group into irreducible subspaces. We recover some of these subspaces in the $3$-copy case by constructing intertwiners between the subrepresentations found in \cite{Helsen2016} and subrepresentations in the $3$-copy case. Furthermore, we obtain some new invariant subspaces. About some of these invariant subspaces, we conjecture that they are irreducible. The $2$-copy adjoint representation \cite{Helsen2016} has found applications in randomized benchmarking \cite{RB11, RBCL1,RBCL2}. Other works and, also applications of the representation theory of the Clifford group have also been found \cite{Appleby2011, Zhu16,Gross17,Gross21,Haferkamp2020}. In our work we further develop the representation theory of the Clifford group. We believe that this would also be of significant interest to members of the quantum information and quantum computation community.


\section{Preliminaries and notation}
\label{sec:pre_not}

\noindent The quantum system we study is that of $m$-qubits. The Hilbert space is isomorphic to $\C^{N}$, where $N=2^m$.  We will be working with a system of $t$ copies of $m$-qubits, whose Hilbert space dimension is $2^{tm}$. The complex vector space of linear operators on such a system will be denoted by $\BCg$. \noindent By the support of a linear operator $A$ on an inner product vector space $V$, we mean the orthogonal complement in $V$ of the kernel of $A$, i.e., $\ker(A)$. 
\begin{equation}
    \label{eq:supp}
    \supp(A) \ = \ \left\{ \ v \in V, \ | \ \bk{w}{v} = 0, \ \forall \ w \in \ker(A) \right\}.
\end{equation}
 A quantum state $\rho$ will be a linear operator in $\BCg$, such that $\rho \ge 0$ and $\tr \rho =1$.  A quantum channel is a trace preserving and completely positive linear map on $\BCg$. The $\Diamond$-norm is a distance measure between two quantum channels $\mathrm{\Phi}$  and $\mathrm{\Psi}$, i.e.,
\begin{equation}
    \label{eq:def_diamond_norm}
    \left\|  \, \mathrm{\Phi} \, - \,  \mathrm{\Psi} \right\|_{\Diamond} \, = \, \underset{\rho}{\mathrm{sup}} \, \left\|   \, \mathrm{\Phi} \otimes  \id \ \left( \rho \right)  \ -  \mathrm{\Psi} \otimes \id \ \left( \rho \right) \right\|_{1},
\end{equation}
where $\|(\cdot)\|_1$ is the one-norm on $\BCg$: $\|A\|_1 = \tr \sqrt{A^\dag \  A }$, and $\id$ is the identity quantum channel.  In the RHS of equation \eqref{eq:def_diamond_norm}, the quantum channels $\Phi\otimes \mathrm{id}$ and $\Psi \otimes \mathrm{id}$ act on two copies of $t$ $m$-qubit systems. $\mathrm{\Phi}$ and $\mathrm{\Psi}$ act on the first of these two copies, whereas $\id$ acts trivially on the second copy.

Let $\left\{ U_i \right\}_{i}$ be a finite subset of $\U(N)$ of size $n$. Let $p_i>0$ be a non-zero probability distribution over $i$. Then an ensemble $\mathcal{E}$ is defined as
\begin{equation}
    \label{eq:def_unitary_ensemble}
    \mathcal{E} \ := \ \left\{ \ \left(p_i, U_i\right) \ \big| \ i \in \left\{1,2,\cdots n\right\}   \right\}.
\end{equation} 
  The $t$-fold twirl of $X$ in $\BCg$ by $\mathcal{E}$ is the quantum channel $\mathcal{G}_\mathcal{E}$ which implements the following transformation on an input $X$,
\begin{equation}
    \label{eq:t_twirl_general}
    \mathcal{G}_\mathcal{E} (X) \ := \ \sum_{i =1}^n \ p_i \ U_i^{\otimes t} \ X \ {U_i^\dag}^{\otimes t}.
\end{equation}

\begin{defn}
\label{def:Haar_measure} The unitary group $\U(N)$ admits a unique left and right translation invariant measure, $\mu$, known as the Haar-measure. A random matrix in $\U(N)$ which is distributed according to the Haar-measure is called a Haar-random unitary. 
\end{defn}

The $t$-fold twirl of a quantum state by the Haar measure will be denoted by $\H$, namely,
\begin{equation}
\label{eq:defn_t_design}
\H(\rho) \ := \ \intH \ U^{\otimes t} \ \rho \ {U^\dag}^{\otimes t} . 
\end{equation}

\begin{defn}
\label{def:unitary_t_design}
A unitary $t$-design is a finite ensemble of unitaries $\left\{p_i, U_i \right\}_{i=1}^n$ whose moments of order less than or equal to $t$ equal that of the Haar measure.
\end{defn}
The definition of a unitary $t$-design can be cast in many equivalent forms (see \cite{Cleve16} for many different forms and proofs of their equivalence). We will use the following one. An ensemble of unitaries $\mathcal{E}$ is a unitary $t$-design if and only if its $t$-fold twirl is the same as the $t$-fold twirl by the Haar-measure:
\begin{equation}
\label{eq:defn_t_design_2}
\intH \ U^{\otimes t} \ \rho \ {U^\dag}^{\otimes t} \, = \, \sum_{i=1}^n \, p_i \  U_i^{\otimes t} \ \rho \ {U_i^\dag}^{\otimes t}.
\end{equation}

\begin{defn}
\label{def:eps_approx_unitary_design}
An $\epsilon$-approximate unitary $t$-design is an ensemble of unitaries $\mathcal{E}$ whose $t$-fold twirl is within $\epsilon$ $\Diamond$-norm distance of a unitary $t$-design, i.e.,
\begin{equation}
    \label{eq:diamond_norm}
    \left\|  \, \mathcal{G}_{\mathcal{E}} \, - \,  \H \right\|_{\Diamond} \, \le \, \epsilon. 
\end{equation}
\end{defn}

\begin{defn}
\label{def:asymptotic_unitary_design}
An asymptotic unitary $t$-design is an ensemble of unitaries $\mathcal{E}$, whose $t$-fold twirl denoted by $\mathcal{G}$, when iterated $k$ times, implements a unitary $t$-design in the $k \rightarrow \infty$:
\begin{equation}
\label{def:asymptotic_unitary_design_2}
\mathrm{lim}_{k \rightarrow \infty} \, \mathcal{G}^k  \, = \H,
\end{equation}
where by $\mathcal{G}^k$ we mean $\underbrace{\mathcal{G} \circ \mathcal{G} \cdots  \circ\mathcal{G}}_{k \ \mathrm{times}}$.
\end{defn}

One typically asks is how fast the convergence is. Explicitly, given some $\epsilon > 0$, one wants to know some $k$, so that
\begin{equation}
    \label{eq:convergence_asymptotic_design}
       \left\|  \, \mathcal{G}^k \, - \,  \H \right\|_{\Diamond} \ \le \epsilon.
\end{equation}
Here $k$ is measured in terms of the order of magnitude in $m$ and $\epsilon$: one seeks a function $f(m,\epsilon)$ such that for $m\ge m_0$ and $\epsilon \le \epsilon_0$, $k \ge f(m_0,\epsilon_0)$. One expects $f$ to be of the following form: $f(m,\epsilon) = \mathrm{poly}(m) + \log \ 1/\epsilon$. When this is the case one says that $k  = \ \mathcal{O}\left(\mathrm{poly}(m) \ + \ \log 1/\epsilon\right)$. 

\subsection{Pauli group}
\label{subsec:Pauli}
The single qubit Pauli group is generated by three unitary operators $X$, $Z$ and $i \unity_2$, where  
\begin{equation}
    \label{eq:XYZ}
    X = \begin{pmatrix} 
     0 & 1 \\
     1 & 0 
     \end{pmatrix}, \;  \;
     Z = \begin{pmatrix} 
     1 & 0 \\
     0 & -1 
     \end{pmatrix},
\end{equation} and by $\unity_2$ we mean the $2 \times 2$ identity matrix. The $m$-qubit Pauli group will be denoted by  $\Pmbb$. It is the $m$-fold direct product of the single-qubit Pauli group, i.e., 
\begin{equation}
\label{eq:Pm_P1}
\hat{\mathcal{P}}_m \, = \, \underbrace{\hat{\mathcal{P}_1} \otimes \hat{\mathcal{P}_1} \otimes \cdots \otimes \hat{\mathcal{P}_1}}_{m \; \mathrm{times}}.
\end{equation}

\noindent The order of $\Pmbb$ is $4N^2$. $\Pmbb$ has a non-trivial centre: $\mathcal{Z} \left( \Pmbb \right) \ = \ \left\langle \  i \unity_{N} \right\rangle$, and  it has $4$ elements. $\mathcal{Z}\left( \Pmbb \right)$ is a normal (invariant) subgroup of $\Pmbb$, and the factor group is $\Pmb$:
\begin{align}
    \label{eq:pmbb_pmb_homomorphism}
    & \Pmb \ = \ \Pmbb \  / \ \mathcal{Z} \left( \Pmbb \right),  \\ 
    \label{eq:pmbb_pmb_homomorphism1}
    & E \in \Pmbb \longrightarrow  \ [ \ E \ ] = \left\{ \pm E, \pm i \ E\right\} \in \Pmb,
\end{align}
where 
the following composition law can be consistently defined on $\Pmb$:
\begin{equation}
    \label{eq:composition_law_pmb0}
    [ E].[F] =  [E \  F], \,  \forall \ E,F \in \Pmbb.
\end{equation} 
Hence $\Pmb$ is an abelian group. We earlier chose coset representatives for each element of $\Pmb$, such that they were self-adjoint. These also form a convenient choice for a basis for the operator space on the $m$-qubit system. Later, it will be seen that it is an orthogonal basis for this operator space (see Section \ref{eq:sp_linear_operators}). For convenience we define another related set $\Pm$ as 
\begin{equation}
    \label{eq:def_Pm}
    \Pm \ := \ \Pmb \setminus \left\{ \unity \right\}.
\end{equation} Next, $\Pmb \ \simeq \ \mathbb{F}_2^{2m}$, which is the $2m$-dimensional vector space over field $\mathbb{F}_2$, and the `group composition law' is addition.
\begin{thm}
\label{thm:pmb_fm_iso}
\begin{equation}
    \label{eq:pmb_fm_iso}
 \Pmb  \ \simeq \ \mathbb{F}_2^{2m}.
\end{equation}
\end{thm} \noindent For proof, see Appendix \ref{subsec:ph_thm_pmb_fm_iso}.  For some fixed choice of isomorphism \footnote{The isomorphism isn't unique - we construct a specific isomorphism in Appendix \ref{subsec:ph_thm_pmb_fm_iso}.} map $a \in \Fm$ to $E\in \Pmb$, then we label $E$ as $E(a)$. Note that $E(0)=\unity$. From the proof in Appendix \ref{subsec:ph_thm_pmb_fm_iso}, we see that the  commutation relation between two Paulis $E(a)$ and $E(b)$ is
\begin{equation}
    \label{eq:Pm_composition_law}
    E(a) \ E(b) \ = \ (-1)^{\s{a}{b}} \ E(b) \ E(a),
\end{equation}
{where $ \s{.}{.}:\ \Fm \times \Fm \ \longrightarrow \mathbb{F}_2  $ is the symplectic bilinear form defined on $\Fm$} 
\begin{equation}
    \label{eq:symp}
    \s{a}{b} := \ a^T \ \Om_{2m} \ b,
\end{equation} where $b$ in $\Fm$ is represented by a $2m$ bit column vector, and $a^T$ represents a $2m$ bit row vector, and where $\Om$ is the $2m \times 2m$ matrix which takes the form
\begin{equation}
    \label{eq:Omega}
    \Om \ = \ \begin{pmatrix}
    0_m & \unity_m \\
    \unity_m & 0_m 
    \end{pmatrix}.
\end{equation}
Note that $\s{a}{a} = 0$. Equation \eqref{eq:Pm_composition_law} tells us that when $\s{a}{b}=0$, then $E(a)$ and $E(b)$ commute whereas when $\s{a}{b}=1$, then $E(a)$ and $E(b)$ anti-commute. We will use the Pauli basis of $\BCg$, i.e., for $E(a_1),E(a_2),\cdots,E(a_t)\in\Pmb$, tensor products of the form $E(a_1) \otimes E(a_2) \cdots\otimes E(a_t)$ will be basis vectors. For all proofs in Section \ref{sec:pfs}, we will use elements in $\BCg$ as vectors in a vector space, and not as linear operators on $\C^{N^{\otimes t}}$. Hence it will prove convenient to simplify the notation as follows.
\begin{equation}
    \label{eq:vec1_t_copy}
    \text{\scalebox{0.75}{%
 $\frac{1}{\sqrt{N^t}}$}}  \ 
    E\left( a_1 \right) \otimes E\left( a_2 \right) \otimes \cdots \otimes E\left( a_t \right) \ \longrightarrow \ \ket{a_1,a_2,\cdots, a_t}, \; \forall \ a_1,a_2,\cdots,a_t \in \Fm.  
\end{equation} 
For any other $A \in \BCg$, we simply use the notation $A \longrightarrow \ket{A}$.

\subsection{The Clifford group}
\label{subsec:Cliff}

The Clifford group $\Cl$ is the normaliser of the Pauli group $\Pmb$ in the $m$-qubit unitary group $\U(N)$, i.e., it is the set of all those unitaries $U \in U(N)$ such that for any $E \in \Pmb$, we have 
\begin{equation}
    \label{eq:normalising_action_of_U}
    U \ E \ U^\dag \ =  \ \pm F, \; \mathrm{for} \; \mathrm{some} \; F \in \Pmb.
\end{equation} 

\subsection{Transvection Cliffords}
\label{subsec:Clm_transvections}
We define transvection Cliffords as follows. For each $E \in \Pmb$, choose any other $F \in \Pmbb$, and define the following linear operator:
\begin{equation}
    \label{eq:Cl_transvection}
    U_{E,F} \, := \, \dfrac{1}{\sqrt{2}}\left( \, F \,  +  \,  i  E \ F \, \right).
\end{equation} 
Note that $U_{E,F}$ defined in equation \eqref{eq:Cl_transvection} is a Clifford unitary because for $X\in \Pmb$,
\begin{equation}
    \label{eq:Cl_transvection_adjoint_action}
    U_{E,F} \ X \ U_{E,F}^\dag \ = \ \begin{cases}
   \, \, \,   \pm X, \, \mathrm{when} \ [E,X] \ = \ 0, \\  \, \, \,
     \pm  i \ E \ X, \, \mathrm{when} \ \left\{ E, X \right\} \ = \ 0. 
    \end{cases}
\end{equation}

\begin{rem}
\label{rem:transvections}
Later one we will also define transvections as a special case of symplectic matrices in $\Sp$. To distinguish between both kinds of transvections, we will refer to the Clifford unitaries as transvection Cliffords, and to the symplectic matrices simply as transvections. 
\end{rem}

\noindent 

\section{Main results}
\label{sec:main_results}

\noindent We prove the main results in the following two subsections. In Subsection \ref{subsec:main_results_asymptotic_unitary_designs} we give the main results of the asymptotic unitary $2$- and $3$-designs, and in Subsection \ref{sec:3_copy_adjoint_Cm} we present a partial decomposition of the three-copy adjoint representation of the Clifford group $\Cl$ on $m$-qubits.

\subsection{Asymptotic unitary $2$- and $3$-designs}
\label{subsec:main_results_asymptotic_unitary_designs}

\noindent We state and prove the main results using the definition of a unitary $t$-design by twirling. To do this, we imagine that one is implementing the Scheme \ref{schm}. For that we need two definitions.
\begin{defn}[Pauli ensemble]
\label{def:ensemble_Paulis}
We define the Pauli ensemble $\mathcal{P}$ to be the uniform ensemble of Paulis in $\Pmb$, namely,
\begin{equation}
    \label{eq:def_ensemble_Paulis}
    \mathcal{P} \ := \ \left\{ \ \left(\frac{1}{N^2}, \ E \right), \big  | \ \forall \ E \in \Pmb  \right\}.
\end{equation} The $t$-fold twirl over $\mathcal{P}$ will be denoted by $\EP$: $\forall \ X \in \BCg$, 
\begin{equation}
    \label{eq:def_EP}
    \EP \ (X) \ := \ \frac{1}{N^2} \sum_{E \in \Pmb} \ E^{\otimes t} \ X \ E^{\otimes t}.
\end{equation}
\end{defn}

\begin{defn}[Ensemble of transvection Cliffords] 
\label{def:ensemble_transvection}
We construct an ensemble of transvection Cliffords as follows: for each $E \in \Pmb$, choose some $F \in \Pmbb$, such that $\left\{ E, F \right\} = 0$ (see Remark \ref{rem:transvection_convention}), and collect all such $U_{E,F}$ in a set. The set has $N^2$ elements, one element for each $E$. Each $U_{E,F}$ has a probability $\frac{1}{N^2}$. 

\begin{equation}
    \label{eq:ensemble_over_transvections}
    \mathcal{T} \ := \ \left\{ \ \left( \ \frac{1}{N^2},  \ U_{E,F} \ \right) \ \big| \ \forall \ E \in \Pmb, \, F \ \mathrm{fixed} \ \mathrm{but} \ \mathrm{arbitrary}.    \right\}.
\end{equation} The $t$-fold twirl by $\mathcal{T}$ will be denoted by $\ET$: $\forall \ X \in \BCg$
\begin{equation}
    \label{eq:def_ET}
    \ET \ (X) \ := \ \frac{1}{N^2} \sum_{U_{E,F} \in \mathcal{T}} \ U_{E,F}^{\otimes t} \ \left( X \right) \ \left( U_{E,F}^\dag \right)^{\otimes t}  \ .
\end{equation}
\end{defn}

\begin{schm}[Formal statement of scheme.]
\label{schm}

\noindent This scheme is implemented in $t$ copies of $m$ qubit systems. Let $\mathcal{P}$ and $\mathcal{T}$ be some Pauli and transvection ensembles for $m$-qubit systems as in Definition \ref{def:ensemble_Paulis} and Definition \ref{def:ensemble_transvection}. Then scheme is as follows.
\begin{itemize}
    \item Sample from $\mathcal{P}$ once, and rotate the input state $\rho$ by the sampled Pauli.
    \item Perform the following step $k$ times: sample once from $\mathcal{T}$, and rotate  the input state $\rho$ by the sampled transvection Clifford.
\end{itemize}
\end{schm}


\noindent We first present two theorems on the optimal rate of convergence for an asymptotic unitary $t$-designs, where $t=2$ and $t=3$.




\begin{thm}[Technical version of Theorem 1]
\label{thm:asymptotic_unitary_2_design}
The Scheme \ref{schm} implements an asymptotic unitary $2$-design for $m \ge 2$. The convergence rate of this scheme for $m \ge 3$  is as follows.
\begin{equation}
\label{eq:cvg_unitary_2_design}
\mathrm{for } \ \mathrm{any}  \  \epsilon > 0, \mathrm{when} \ k \ge \  6 + 5/4 \log 1/\epsilon \Longrightarrow \ \left\|  \, \ET^{k} \circ \EP \, - \,  \H \right\|_{\Diamond} \, < \, \epsilon,
\end{equation}
where \begin{enumerate}
    \item $\EP$ denotes the quantum channel which implements the $2$-fold twirl by the uniformly sampling Paulis in $\Pmb$.
    \item $\ET$ denotes the quantum channel which implements the $2$-fold twirl by uniformly sampling $N^2$ transvections in $\mathcal{T}$.
\end{enumerate}
\end{thm}
\begin{thm}
\label{thm:asymptotic_unitary_3_design}
The Scheme in \ref{schm} implements an asymptotic unitary $3$-design for $m\ge 3$. The convergence rate for this scheme is as follows: 
\begin{equation}
\label{eq:cvg_unitary_3_design}
\mathrm{for} \  \epsilon > 0, \mathrm{when} \  k >  \dfrac{ 3m \ +  \log 1/\epsilon}{1 - \log \left( \  1  \ + \frac{4 }{N} \ + \  \frac{1}{2N(N-2)} \  \right)} \ \Longrightarrow \ \left\|  \, \ET^{k} \circ \EP \, - \,  \H \right\|_{\Diamond} \, < \, \epsilon,
\end{equation}
where \begin{enumerate}
    \item $\EP$ denotes the quantum channel which implements the $3$-fold twirl by the uniformly sampling Paulis in $\Pmb$.
    \item $\ET$ denotes the quantum channel which implements the $3$-fold twirl by uniformly sampling $N^2$ transvections in $\mathcal{T}$.
\end{enumerate}
\end{thm}
\begin{corr}
\label{corr:order_steps}
The Scheme $\ref{schm}$ will also work by reversing the order of the steps, or by twirling over $\mathcal{P}$ at some intermittent point in between two successive iterations of the $k$-time $\mathcal{T}$ twirl.
\end{corr}
\noindent For the proofs, see Subsection \ref{subsec:cgt_analysis}. We give the outline of the proofs here.

\begin{proof}[Outline of proofs for Theorems \ref{thm:asymptotic_unitary_2_design} and \ref{thm:asymptotic_unitary_3_design}:] 
We prove the following two things. \begin{itemize} \item[(i)] $\ET \circ \EP - \H$ is a diagonalizable linear operator acting on $\BCg$.
\item[(ii)] The largest eigenvalue of $\ET\circ \EP - \H$ is less than $1$. Also it scales as $\mathcal{O}(1)$ with $m$. \item[(iii)] For the $t=2$ case, the convergence (viz, $k = \mathcal{O}\left(\log 1/\epsilon\right)$) will be separately proved. For the $t=3$ case, the convergence (viz $k=\mathcal{O}\left(m + \log 1/\epsilon \right)$) is straightforward to evaluate.  \end{itemize}
\end{proof}

\begin{rem}
\label{rem:transvection_convention}
We defined transvection Cliffords as $U_{E,F} = \frac{1}{2}(\unity + i E)F$. While this is a Clifford unitary for either cases: $[E,F]=0$ and $\left\{ E, F \right\}=0$, it is self-adjoint only when $\left\{ E, F \right\}=0$. We will later on see that we require the $U_{E,F}$'s to be self-adjoint (or skew-adjoint) in order for $\ET$ to be self-adjoint as a super-operator (see Lemma \ref{lem:Phi_adjoint}). Note that $\ET$ has been defined with only $N^2$ transvection Cliffords. If one were to choose $[E,F]=0$ instead, then instead of having only $N^2$, we would need $2 N^2$ transvection Cliffords in order for $\ET$ to be a self-adjoint superoperator.
\end{rem}
\begin{rem}
\label{rem:asymptotic_Clifford_desing}
Since the Clifford group is a unitary $3$-design, Theorem \ref{thm:asymptotic_unitary_3_design} can also be viewed as the statement that Scheme \ref{schm} implements an asymptotic Clifford design. 
\end{rem}
\begin{rem}
\label{rem:thm_asymptotic_unitary_designs}
 Directly proving the convergence rates in Theorems \ref{thm:asymptotic_unitary_2_design} and \ref{thm:asymptotic_unitary_3_design} automatically proves that the Scheme gives us an asymptotic unitary $2$- and $3$-design.
\end{rem}
\begin{rem}
\label{rem:Cl_transvection1}
When $E= \unity$, $U_{E,F}  \ = \ \e^{\pm i \pi /4} \ F $. We choose to include the $U_{\unity,F}$ as transvection Cliffords as well. Not doing so doesn't affect the results.
\end{rem}
\begin{rem}
\label{rem:Cl_transvection2}
One can define the naive composition of ensembles $\mathcal{E}$ and $\mathcal{F}$ as follows: for $(p_i,U_i) \in \mathcal{E}$ and $(q_j,W_j) \in \mathcal{F}$, 
\begin{equation}
    \label{eq:def_composition_ensembles}
\mathcal{E} \circ \mathcal{F} := \left\{ \left(p_iq_j,U_i W_j \right) \ \big| \ \forall \ \mathrm{indices} \ i, j \right\}.
\end{equation}
Denote by $\mathcal{E}^k$ the $k$-time self-composed ensemble, i.e., $\underbrace{\mathcal{E} \circ \mathcal{E} \circ \cdots  \mathcal{E}}_{k \ \mathrm{times}}$. Then one sees that $\ET^k \circ \EP \ = \mathcal{G}_{\mathcal{T}^k\circ \mathcal{P}}$. Thus, $\ET^k\circ \mathcal{\EP}$ is also realizable as $t$-fold twirl over an ensemble - in this case $\mathcal{T}^k\circ \mathcal{P}$. And the scheme amounts to sampling once from this ensemble, and rotating the input state $\rho$ by the sampled Clifford. 
\end{rem}
\begin{rem}
\label{rem:optimal_convergence_rate}
The convergence rates in Theorems \ref{thm:asymptotic_unitary_2_design} and \ref{thm:asymptotic_unitary_3_design} may be anticipated from the following arguments. In the group algebra of $\C\Cl$ (for the concept of a group algebra, see \cite{wiki_group_ring}), one may represent the uniform ensemble of transvections $\mathcal{T}$ as follows,
\begin{equation}
    \label{eq:group_algebra_1}
    \mathcal{T}= \sum_{\tau \in \mathcal{T}} p(\tau) \tau,
\end{equation} where $\tau$ represents a transvection element in $\Cl$, and $p:\Cl \rightarrow \C$, is a probability distribution on $\Cl$, such that 
\begin{equation}
    \label{eq:rem_optimal_convergence_rate}
    p(g)  \ = \ \begin{cases} \frac{1}{N^2} \ \mathrm{when} \ $g$ \ \mathrm{is} \ \mathrm{a} \ \mathrm{transvection}, \\ 
    0, \ \mathrm{otherwise}. 
    \end{cases}
\end{equation}
The ensemble generated by composing $\mathcal{T}$ with itself, i.e., the ensemble corresponding to $\mathcal{T} \circ \mathcal{T}$ (see Remark \ref{rem:Cl_transvection2}) is then represented by
\begin{equation}
\label{eq:group_algebra_2} 
\mathcal{T}\circ \mathcal{T} \ = \ \sum_{g \in \Cl} \left( \  \left[p \star p \right] (g)  \right) g,
\end{equation} where $g$ represents a group element in the group algebra $\C\Cl$, $p \star p $ represents the group convolution of $p$ with itself, that is
\begin{equation}\label{eq:group_convolution}
\left( p \star p \right)(g) \ = \ \sum_{h \in G} p(h)p\left( h^{-1}g \right).
\end{equation}
While $p$ is a uniform probability on $N^2$ elements of $\Cl$, $p \star p$ is approximately a uniform distribution on $N^4$ elements of $\Cl$. Similarly, $p \star p \cdots \star p$ ($k$ times) is \emph{approximately} a uniform distribution on $N^{2k}$ elements of $\Cl$. This may be computed easily using the definition of transvections in Eq.~ \eqref{eq:Cl_transvection} (or even Eq. ~ \eqref{eq:Th_matrix_form}). Thus a finite number of iterations gives an approximately uniform distribution on $\SL$ which is a $2$-design, whereas $\mathcal{O}(m)$ iterations give an approximately uniform distribution on $\Cl$, which is a $3$-design. We refer the interested reader to \cite{Pitt_Coste_unpublished} for more along the same line of reasoning. \newline ~ \newline \textbf{On the optimality of this scheme:} If the cardinality of the support of the probability measure $p \star p \star \cdots \star p$ ($k$-times) would've grown as $o\left(N^{2k}\right)$, and the distribution on this support not been (approximately) uniform (in order of magnitude terms), the order of magnitude of iterations required for convergence would've been higher. In this sense, the scheme is optimal.
\end{rem}

\subsection{Subrepresentation of the three-copy adjoint representation of the Clifford group \texorpdfstring{$\Cl$}{}: partial decomposition}
\label{sec:3_copy_adjoint_Cm}
Here we give some of the subrepresenations carried by the $3$-copy adjoint representation of the Clifford group. For convenience we will use the Pauli basis to represent the subspaces which the subrepresentations act on (see Subsection \ref{subsec:Pauli}). We obtain subspaces which lie in the following subspace of $\BCr$.
\begin{equation}
\label{eq:Clifford_vspace}
    \mathrm{span} \left\{ \ \ket{a,b,c} \ | \ a,b,c \in \Fm, \text{such that } a+b+c = 0 \right\}.
\end{equation}
We begin with an important lemma. 
\begin{lem}
\label{lem:Clifford_irreps}
Of the subrepresentations in the single-copy and two-copy adjoint representations, the ones which act on the trivial representation of the Pauli group are the only ones which appear in the subspace in equation \eqref{eq:Clifford_vspace}. 
\end{lem} \noindent This is because the Pauli group also acts trivially on the subspace in equation \eqref{eq:Clifford_vspace}.

We now give the subrepresentations we found. The case where $a,b,c=0$ corresponds to the identity $\unity$ operator, and it is trivial to see that the Clifford group acts trivially on the subspace spanned by this. When $a=b$ and $c=0$ (or any permutation thereof), we get the following invariant subspace \begin{equation*} \mathrm{span} \ \left\{ \ket{a,a,0}  \ | \ a \in \Fs \ \right\}, \end{equation*} which is a subrepresentation of the Clifford group, and also easily seen to be isomorphic to the diagonal sector which was given in \cite{Helsen2016} (see Lemma 4 in \cite{Helsen2016}). The subspace which remains is such that $a,b,c \neq 0$. Some of the subrepresenations which appear in this subspace, have appeared before in the single copy adjoint representation of the Clifford group and the two-copy adjoint representation of the Clifford group \cite{Helsen2016}. The remaining subrepresentations which we obtain, haven't appeared so far, in works which we are aware of. Some of these subrepresentations are reducible, but we haven't been able to reduce them further, while we are not sure about the remaining ones. First, we need to define the action of the symmetry group $\Sr$ on $\BCr$. For any permutation $\sigma \in \Sr$, define its representation $W_\sigma$ to have the following action 
 \begin{equation}
     \label{eq:Wsigma}
     W_\sigma \ \ket{a_1,b_1,c_1} \ = \ \ket{a_{\sigma^{-1}(1)}, \ a_{\sigma^{-1}(2)}, \ a_{\sigma^{-1}(3)}},
 \end{equation}
where we have defined the action of $W_\sigma$ on arbitrary Pauli operators, for $a_1$, $a_2$ and $a_3$ in $\Fm$. 
\begin{thm}
 \label{thm:nc_inv_subsp}
 In the following, we give some invariant subspaces which lie in the following subspace: \begin{equation*} \mathrm{span}\left\{ \ \ket{a,b,a+b} \ |  \ a,b \in \Fs, \mathrm{such} \ \mathrm{that} \ \s{a}{b}=1 \  \right\}. \end{equation*} To define these invariant subspaces, we need to define the following. For any $a \in \Fs$, $b \in \Fm$ such that $\s{a}{b}=0$, we define
     \begin{equation}
         \label{eq:nc_a_b_0}
         \ket{\hat{a};b} \ :=  \ - \ket{\hat{a};a+b} \ := \ \sqrt{\dfrac{2}{N}} \ \sum_{h:\s{a}{b}=1} \ \left( -1 \right)^{\s{h}{b}} \ \ket{a,h,a+h} .
     \end{equation}
 \begin{enumerate}
     \item[(i)] Define $ V^{(\mathrm{nc})}_d$ as \begin{equation}
         \label{eq:nc_c_dd}
         V^{(\mathrm{nc})}_d \ := \  \mathrm{span} \left\{ \ket{\hat{a}} \ | \ a \in \Fs \right\}, \ \mathrm{where} \ \ket{\hat{a}} := \ket{\hat{a};0} \ \text{(See equation \eqref{eq:nc_a_b_0})}.
     \end{equation} \noindent This subspace carries a subrepresentation of $\Cl$, and it is equivalent to the subrepresentation carried by the diagonal sector subspace in \cite{Helsen2016} (see Lemma 4 in \cite{Helsen2016}), with the intertwiner\footnote{For an explanation of an intertwiner, please see equation \eqref{eq:def_intertwiner}.} given by (we restrict the action of the intertwiner to the diagonal sector)
     \begin{equation}
         \label{eq:nc_d_intertwiner}
         L^{\mathrm{d}\rightarrow (\mathrm{nc})} \ \ket{a,a} \ =  \ket{\hat{a}}, \ \forall \ a \in \Fs,
     \end{equation}
     where $\ket{a,a} \in \BCt$ and $\ket{\hat{a}}$ is given by equation \eqref{eq:nc_a_b_0} for the choice $b=0$. The subspace $V^{(\mathrm{nc})}_d$ carries three irreducible representations of the Clifford group.
     \item[(ii)] For $b=1,2$,  \begin{align} \label{eq:nc_V_b} 
     V ^{(\mathrm{nc})}_{b} \ :=  \  \mathrm{span} \left\{ \ket{\pm} \ = \ \sum_{x \in \Fs } \ \lambda_x \ \ket{A(x)} \ \Big| \ \sum_{\substack{y \in \Fs\\ \s{y}{x}=0}} \ \lambda_y \ = \ \pm \  \dfrac{N}{2} \ \lambda_x  , \sum_{x\in\Fs} \lambda_x=0.\right\} \end{align} where 
     \begin{align}
         \label{eq:Aab}
      &   \ket{A(a)} \ := \ \sum_{h: \substack{\s{h}{a}=0 \\ h \neq a,0}} \ \ket{\hat{h};a}.
     \end{align} 
     The above subspaces carry subrepresenations of $\Cl$ which are equivalent to the subrepresentations of $\Cl$ carried by the subspaces $V_1$ and $V_2$ defined in Lemma 4 of \cite{Helsen2016}. 
     \item[(iii)] Define the projector, $\ps$  as 
     \begin{equation}
         \label{eq:def:ps}
         \ps \ := \ \frac{1}{3} \left( \unity + W_{[(123)]} + W_{[132]} \right).
     \end{equation}
     Then the following is an invariant subspace of the adjoint representation of the Clifford group.
     \begin{equation}
         \label{eq:nc_c_null}
         V^{(\mathrm{nc})}_{\mathrm{null,S}} \ := \ \mathrm{span} \left\{ \ \ps \ \left( \ket{\hat{a};b} + \ket{\hat{b};a+b} + \ket{\hat{a+b};a} \right) \ | \ \forall \ a,b \in \Fs, \ a \neq b, \ \s{a}{b}=0  \right\}.
     \end{equation} \noindent  
     \item[(iv)] Define the projector $\mathcal{P}_1$ in as 
     \begin{equation}
         \label{eq:def:pg}
         \mathcal{P}_{1} \ := \ \frac{1}{3} \left(  2 \unity  - \  \left( W_{[(123)]} + W_{[132]} \right) \  \right).
     \end{equation}
     Then the following is an invariant subspace of the adjoint representation of the Clifford group.
     \begin{equation}
         \label{eq:nc_c_null_0}
         V^{(\mathrm{nc})}_{\mathrm{null,1}} \ := \ \mathrm{span} \left\{ \ \mathcal{P}_{1} \ \left( \ket{\hat{a};b} + \ket{\hat{b};a+b} + \ket{\hat{a+b};a} \right) \ | \ \forall \ a,b \in \Fs, \ a \neq b , \ \s{a}{b}=0   \right\}.
     \end{equation} \noindent 
     \item[(v)] The intersection of the orthogonal complements of all above subspaces within $\mathrm{span} \left\{ \ket{a,b,a+b} + \ket{a,a+b,b} \ | \ \s{a}{b}=1 \right\}$, i.e., 
     \begin{equation}
         \label{eq:orthogonal_complement}
         \mathrm{span} \left\{ \ket{a,b,a+b} + \ket{a,a+b,b} \ | \ \s{a}{b}=1 \right\} \ \cap \ \left(  V^{(\mathrm{nc})}_d  \ \oplus \ V^{(\mathrm{nc})}_{\mathrm{null}} \ \oplus  \ V^{(\mathrm{nc})}_{1} \  \oplus \ V^{(\mathrm{nc})}_{1}   \right)^\perp,
     \end{equation} where the superscript $\perp$ denotes the orthogonal complement of a subspace, and $V^{(\mathrm{nc})}_{\mathrm{null}} \coloneqq V^{(\mathrm{nc})}_{\mathrm{null},S} \oplus V^{(\mathrm{nc})}_{\mathrm{null},1}$
 \end{enumerate}
\end{thm} The proof is given in Appendix \ref{sec:appendix_nc_0}.
\begin{rem}
\label{rem:subreps_new}
The subrepresenations carried by the subspaces $V^{(\mathrm{nc})}_{\mathrm{null,S}}$ and $V^{(\mathrm{nc})}_{\mathrm{null,0}}$ haven't appeared among the subrepresenations in the single-copy and two-copy adjoint representations of the Clifford group. 
\end{rem}
Define $W$ as 
\begin{equation}
    \label{def:W}
    W  \ := W_{[(123)]} - W_{[(132)]}.
\end{equation}
 Then we have 
 \begin{thm}
 \label{thm:nc_1_decomposition} 
 Let $V^{(\mathrm{nc})}_{\mathrm{irrep}}$ be any irreducible subrepresentations contained in $V^{(\mathrm{nc})}\cap \mathrm{supp}W$, where $x=d,1,2,\mathrm{null}$\footnote{These irreducible subrepresentations are contained in the subrepresentations mentioned in Theorem \ref{thm:nc_inv_subsp}.} and also let $\Sr$ act non-trivially on $V^{(\mathrm{nc})}_{\mathrm{irrep}}$.  Then $W \ V^{(\mathrm{nc})}_{\mathrm{irrep}} $ carries an isomorphic irreducible representation.  
 \end{thm} The proof is given in Appendix \ref{sec:appendix:nc_W}.
 
\begin{thm}
 \label{thm:c_inv_subsp}
 In the following, we give some invariant subspaces which lie in the following subspace: \begin{equation*} \mathrm{span}\left\{ \ \ket{a,b,a+b} \ |  \ a,b \in \Fs, \mathrm{such} \ \mathrm{that} \ \s{a}{b}=0, \ \mathrm{and} \ a \neq b \  \right\}. \end{equation*} To define these invariant subspaces, we need to define the following. For any $a \in \Fs$, we define 
     \begin{equation}
         \label{eq:c_a}
         \ket{\Bar{a}} \ := \  \sqrt{\dfrac{2}{N}} \ \sum_{\substack{h:\s{h}{a}=0 \\ h \neq 0,a}} \  \ket{a,h,a+h}.
     \end{equation} \noindent Using the above, we define the following subspace.
 \begin{equation}
         \label{eq:c_c_dd}
         V^{(\mathrm{c})}_d \ := \  \mathrm{span} \left\{ \ket{\Bar{a}} \ | \ a \in \Fs \right\}.
     \end{equation} \noindent This subrepresentation is isomorphic to the diagonal sector in \cite{Helsen2016} (see Lemma 4 in \cite{Helsen2016}), with the intertwiner given by
     \begin{equation}
         \label{eq:c_d_intertwiner}
         L^{\mathrm{d}\rightarrow (\mathrm{c})} \ \ket{a,a} \ =  \ket{\hat{a}}, \ \forall \ a \in \Fs,
     \end{equation}
     where $\ket{a,a} \in \BCt$. The subspace $V^{(\mathrm{c})}_d$ carries three irreducible representations of the Clifford group.
\end{thm}
The proof is the identical to the proof of the subspace $V^{(\mathrm{nc})}_d$ in Theorem \ref{thm:nc_inv_subsp}.

%


\section{Proofs}
\label{sec:pfs}

While the proofs of Theorems \ref{thm:asymptotic_unitary_2_design} and \ref{thm:asymptotic_unitary_3_design} are covered in subsection \ref{subsec:cgt_analysis}, prerequisites of the proofs are given in the next few subsections.  

\subsection{Basics: representations of finite groups on finite dimensional complex vector spaces}
\label{subsec:representation_groups}
\noindent Let $G$ be a finite group. A representation of $G$ on $\C^d$ is a map $\lam:G\longrightarrow \mathcal{B}(\C^d)$ which satisfies the following.
\begin{equation}
    \label{eq:representation_def}
    \lam(g) \ \lam(h) \ = \ \lam(g.h), \ \forall \ g,h \in G.
\end{equation}
If there exists a subspace $W$ in $\C^d$ such that for all $w \in W$, $\lam(g) \ w \in W$, then $\lam$ is said to be a reducible representation and $W$ is said to be an invariant subspace of $\C^d$ under the action of $\lam$, and $\lam$ restricted to $W$, i.e., $\lam_{W}:=\lam\big|_{W}$ is a sub-representation of $\Lambda$. If no such $W$ exists then $\lam$ is said to be irreducible. For two representations $\lam_1$ on $\C^{d_1}$ and $\lam_2$ on $\C^{d_2}$, if there exists a non-singular linear operator $X:\C^{d_1} \longrightarrow \C^{d_2}$, such that 
\begin{equation}
    \label{eq:def_intertwiner}
    \lam_2 (g) \ X \ \ket{v_1} \ = \ X \ \lam_1 (g) \ \ket{v_1}, \ \forall \ \ket{v_1}\in \C^{d_1}, \ \mathrm{and} \ \forall \ g \in G,
\end{equation}
then $\lam_1$ and $\lam_2$ are said to be equivalent representations, and $X$ is called the intertwiner between them. If $\lam_1$ and $\lam_2$ are such that the only solutions for $X$ are singular, then $\lam_1$ and $\lam_2$ are said to inequivalent representations. If $d_1 \neq d_2$, $\Lambda_1$ and $\Lambda_2$ are inequivalent. Given the Euclidean inner product on $\C^d$, and some representation $\lam$ on $\C^d$, there exists an equivalent unitary representation of $\lam$ on $\C^d$. We henceforth assume that all representations of all groups are unitary, i.e., $\lam: G \longrightarrow \U(d)$.  For finite dimensional groups, one can decompose a vector space $\C^d$ as a direct sum of orthogonal subspaces, $\C^d = \oplus_k \ W_k$, such that $W_k$ are invariant under $\lam$, and $\lam$ restricted to $W_k$, i.e., $\lam \big|_{W_k}$ is irreducible.

\subsection{Important lemmas in the Pauli group}
\label{subsec:Pauli_1}
\noindent  

Let $\Fs$ denote $\Fm \setminus \left\{ 0 \right\}$. We will need the following lemma repeatedly for this work. Their proofs are given in Appendix \ref{subsec:appendix_lem_1_pf}.
\begin{lem}
\label{lem:1}
Let $a \in {\F}_*^{2m}$. The number of $b \in \Fm$ which satisfy the equation $\s{a}{b}=1$ equals $\frac{N^2}{2}$, and the number of $b \in \Fm$ which satisfy the equation $\s{a}{b}=0$ equals $\frac{N^2}{2}$. This set includes $a,0$ since $\s{a}{0}=\s{a}{a}=0$. 
More generally, for $l \le 2m$, let $a_1$,$a_2$, $\cdots$, $a_l$ be a set of vectors in $\F^{2m}$, let $\dim \mathrm{span} \left\{a_j\right\}_{j=1}^l=d$. For some fixed choice of $b \in \F^{l}$, the set $ \ \left\{ \ h \in \Fm, \ \mathrm{s.t.}  \ \s{a_1}{h} = b_1, \ \s{a_2}{h}=b_2, \ \cdots \s{a_l}{h} = b_l \ \right\}$ has (i) no solutions of the equations $\s{a_j}{h}=b_j$ are inconsistent, (ii) $\dfrac{N^2}{2^d}$ solutions, if the equations are consistent. 
\end{lem}

\subsection{Adjoint representation and inner product on \texorpdfstring{$\BCg$}{}}
\label{eq:sp_linear_operators} 
We will use the well-known Frobenius inner product between two linear operators $A,B \in \BCg$: for $A,B \in \BCg$ represented as $\ket{A}$ and $\ket{B}$ respectively, it is defined as
\begin{equation}
    \label{eq:inner_product_bcg}
    \bk{A}{B} \ = \  \tr \  A^\dag \ B.
\end{equation} 
It is well-known that any quantum channel $\mathrm{\Phi}$ admits a Kraus representation, that is, some set of operators $F_k \in \BCg$, such that for any $X \in \BCg$, $\mathrm{\Phi}(X) \ = \ \sum_{k} \ F_k \ X \ F_k^\dag$, where $\sum_{k}F_k^\dag F_k = \unity$. The inner product on $\BCg$ gives us the adjoint of $\mathrm{\Phi}$: $\mathrm{\Phi}^\dag$ such that for all $A,B$, 
\begin{equation}
    \label{eq:super_operator_adjoint}
 \tr \  A^\dag \ \mathrm{\Phi} \left( B  \   \right) \ = \ \tr  \ \left(\mathrm{\Phi}^\dag ( A )\right)^\dag B,
\end{equation} hence we have the following lemma.
\begin{lem}
\label{lem:Phi_adjoint}
For any quantum channel $\mathrm{\Phi}$ with Kraus operators $F_k$, $\mathrm{\Phi}$ and $\mathrm{\Phi^\dag}$ are related in the following way.
\begin{equation}
    \label{eq:Phi_adjoint}
    \mathrm{\Phi}  (X) = \sum_{k} F_k \ X F_k^\dag \ \Longrightarrow \mathrm{\Phi}^\dag  (X) \ = \ \sum_k \ F_k^\dag \ X \ F_k.
\end{equation}
\end{lem} \noindent The proof is in Subsection \ref{subsec:pf_lem_Phi_adjoint} of the Appendix.
\begin{corr}
\label{corr:selfadjoint_kraus}
A sufficient condition for a quantum channel to be self-adjoint is that the adjoint of the Kraus operators $F_k$ have the followin property: $F_k^\dag = \pm F_k$. Then the quantum channel admits a spectral decomposition into orthogonal eigenspaces.
\end{corr}
\noindent For any $U \in \U(N)$, one can define the adjoint representation of $\U(N)$ on $\BCg$ as 
\begin{equation}
    \label{eq:adjoint_UN_BCG}
    \mathcal{U}_U \ket{X} \ =  \ U^{\otimes t} \ X {U^\dag}^{\otimes t}. 
\end{equation}
Note that for any $A, B \in \BCg$,
 \begin{equation}
     \label{eq:unitarity_adjoint}
     \left( \Ua{U} \ket{A}\right)^\dag \left(\Ua{U} \ket{B} \right) \ =  \ \tr \left( \left( U^{\otimes t} A {U^\dag}^{\otimes t} \right)^\dag \ \left(  U^{\otimes t} B {U^\dag}^{\otimes t}\right)  \right) \ = \ \tr \ A^\dag B \ = \bk{A}{B},
 \end{equation}
 which proves that the adjoint representation of $\U(N)$ defined by the action \eqref{eq:adjoint_UN_BCG} is unitary with respect to the inner product which was defined in \eqref{eq:inner_product_bcg}.
 
\subsection{Clifford group and transvections revisited}
\label{subsec:Clifford}
\noindent It is well known that the quotient group of the Clifford group by the Pauli group, i.e., $\Cl / \Pmbb$ is the symplectic group $\Sp$:
 \begin{equation}
     \label{eq:factor_group}
     \Sp \ \simeq \ \Cl / \hat{\mathcal{P}}_m,
 \end{equation}
 where $\Sp$ is the group of $2m \times 2m$ matrices, which satisfy the symplectic condition, i.e.,
 \begin{equation}
     \label{eq:Sp_condition}
     S \in \Sp \ \Longleftrightarrow \ S^T \ \Om \  S = \Om,
 \end{equation} \noindent and the homomorphism between $\Cl$ and $\Sp$ is clearly visible in the following well-known transformation law (see \cite{Gross2005} for more details).
 \begin{equation}
     \label{eq:Clifford_adjoint_action}
     U \ E(a) \ U^\dag \ = \ \pm \ E\left(S_U \ a \right), \ \forall \ a \in \Fm,
 \end{equation} where we have denoted $S_U$ to be the symplectic transformation corresponding to the Clifford unitary $U$. Thus the adjoint representation of the Clifford group on $\BCg$ takes the following form.
 \begin{equation}
     \label{eq:adjoint_Clifford_BCg}
     \Ua{g} \ \ket{a_1,a_2,\cdots, a_t} \ = \ \pm \  \ket{S_g \ a_1, \ S_g \ a_2, \ \cdots, \ S_g \ a_t}, \ \forall \ g \in \Cl.
 \end{equation} The adjoint transformation under some Pauli $E(h)$ is
 \begin{equation}
     \label{eq:adjoint_transformation_Pauli}
     \Ua{E(h)} \ \ket{a_1,a_2,\cdots, a_t} \ = \  \ (-1)^{\s{h}{s}} \   \ket{a_1,a_2,\cdots,a_t}, \ \mathrm{where} \  s = \sum_{j=1}^{t} \ a_j. \end{equation} This gives us the action of the channel $\EP$:
     \begin{align}
         \label{eq:adjoint_action_EP}
         \EP \ \ket{a_1,\cdots,a_t} \ = \ \begin{cases}
           \ \ket{a_1,\cdots,a_t}, \ \mathrm{when} \ \sum_{j=1}^t \ a_j = 0, \\
           \ 0, \, \, \,\, \, \, \, \, \, \,\, \, \, \, \, \, \,  \, \, \, \, \,  \, \, \, \, \,     \mathrm{otherwise},
         \end{cases}
     \end{align} where we used Lemma \ref{lem:1} to prove the equation at the bottom.  Recall from Subsection \ref{subsec:Clm_transvections} the definition of transvection Clifford unitaries. From equation \eqref{eq:Cl_transvection_adjoint_action}, we see that the symplectic matrices for $U_{E(h),F}$ for all $F \in \Pmbb$ are 
\begin{equation}
    \label{eq:Th_matrix_form}
    T_h \ = \ \unity_{2m}  \ + \ h. h^T \Om, \ \mathrm{and} \ T_h \ a \ = \ a + h \s{a}{h}, \ \forall \ a \in \Fm.
\end{equation} 
It is easily verified that $T_h^2 \ = \ \unity_{2m}$. For $h = 0$, one gets $T_0 \ = \ \unity_{2m}$. We also consider $T_0$ to be a transvection. Thus there are $N^2$ transvections.
\begin{lem}
\label{lem:transvections_conjugacy_class}
The set of all transvections $T_h$ for $h \in \Fs$ forms a conjugacy class.
\end{lem} \noindent The proof is in Subsection \ref{subsec:appendix_transvection_conjugacy_class}. The identity transvection will form a conjugacy class on its own. We will make use of the fact that the set of all transvections (identity included) is closed under conjugation by any $S \in \Sp$.  Another important fact about transvections is that they form a generating set for $\Sp$ (a simple way to see this: \footnote{Construct the symplectic matrices for the Hadamard, phase and CNOT gates. Since the Hadamard, phase and CNOT gates generate $\Cl$, and given the isomorphism in \eqref{eq:factor_group}, their symplectic matrix versions generate $\Sp$. Each Hadamard and phase gates are seen to correspond to some single-qubit transvection, whereas the CNOT can be decomposed into a product of at most two transvections, one a single-qubit and the other a two-qubit transvection.}). The action of $\ET$ is
\begin{equation}
    \label{eq:ET_adjoint_action}
    \ET \ \ket{a_1,a_2,\cdots,a_t} \ = \ \frac{1}{N^2} \  \sum_{h\in \Fm} \ \pm \  \ket{a_1+\s{a_1}{h} \ h, a_2+\s{a_2}{h} \ h, \cdots, a_t+\s{a_t}{h} \ h}.
\end{equation}
\begin{corr}
\label{corr:ET_EP_H_properties}
$\ET$, $\EP$ and $\H$ are self-adjoint with respect the definition of adjoint in Lemma  \eqref{lem:Phi_adjoint}. $\EP$ and $\H$ are projectors. In particular
\begin{align}
    \label{eq:ET_EP_commutativity}
    & \ \ET\circ\EP  \ = \ \EP \circ \ET, \ \\ 
    \label{eq:ET_EP_H}
    &  \ \H \circ \ET \ = \ \EP \circ \H \ = \H,
\end{align} and 
\begin{equation} 
\label{eq:commutation_EP_Clifford}
 \Ua{g}\circ \ET\circ\EP = \ET \circ\EP\circ \Ua{g}, \ \forall \ g \in \Cl.
 \end{equation}
 For $t=2$ and $t=3$, 
 \begin{equation}
 \label{eq:Haar_Clifford_relation}
 \H  \ = \  \mathcal{G}_{\Cl}. 
 \end{equation}
\end{corr} \noindent The proof is in Subsection \ref{subsec:pf_corr_ET_EP_H_properties}. Next we prove an important lemma.
     \begin{lem}
     \label{lem:Pauli_basis_choice}
     There exists a Paul-product basis $\ket{a_1,a_2,\cdots,a_t}$, for $\mathcal{B}\left( \ {\C^N}^{\otimes t} \ \right)$, so that whenever $\sum_{j=1}^t a_j = 0$, then
     \begin{equation}
         \label{eq:Clifford_action_special_basis}
         \Ua{g} \ket{a_1,a_2,\cdots,a_t} \ = \ \ket{S_g a_1,S_g a_2,\cdots,S_g a_t}, \ \forall \ g \in \Cl.
     \end{equation}
     \end{lem} This is proved in Subsection \ref{subsec:appendix_Clifford_action_special_basis}. Using equation \eqref{eq:Clifford_action_special_basis}, we get the following whenever $\sum_{j=1}^t a_j = 0$ for the cases $t=2$ and $t=3$.
     \begin{equation}
   \label{eq:ET_EP_adjoint_action}
    \ET \circ \EP \ \ket{a_1,a_2,\cdots,a_t} \ = \ \begin{cases}
      \frac{1}{N^2} \ \sum_{h \in \F^{2m}} \ \ket{a_1 + \s{a_1}{h}h, \cdots, a_t + \s{a_t}{h}h}, \    \mathrm{when}   \\  \ \sum_{j=1}^t \ a_j = 0, \\ ~ \\
    0, \ \mathrm{when} \ \sum_{j=1}^t a_j \neq 0,
    \end{cases}
\end{equation}
 \begin{equation}
     \label{eq:action_H}
 \H \ket{a_1,a_2,\cdots, a_t} \ = \dfrac{1}{\left| \Sp \right|}  \ \sum_{S \in \Sp} \ \ket{S a_1, S a_2, \cdots, S a_t}.
 \end{equation}
 In the following two lemmas we give the eigenspaces of $\H$ for $t=2$ and $t=3$.
 \begin{lem}
 \label{lem:GH_2_eigenspace}
 For $t=2$, $\H$ is as follows.
 \begin{equation}
     \label{eq:H_for_two}
     \H \ = \ \kb{0,0}{0,0} \ + \ \kb{w}{w},
 \end{equation}
 where $\ket{0,0}$ represents $E(0) \otimes E(0)$ which is $\mathrm{Id}$, and $\ket{w} \ := \ \dfrac{1}{\sqrt{N^2-1}} \  \sum_{a \in \Fs} \  \ket{a,a}$.
 \end{lem}
 The proof is given in Subsection \ref{subsec:appendix:pf_lem_GH_eigenspaces}.
 For $t=3$, that  $\H$ is as follows
 \begin{lem} 
 \label{lem:GH_3_eigenspace} For $t=3$, $\H$ takes the following form.
 \begin{align}
     \label{eq:H_for_three}
     \H \ = &  \ \kb{0,0,0}{0,0,0} \notag \\  
     &  +  \ \kb{w,1}{w,1} \ + \ \kb{w,2}{w,2} \ + \ \kb{w,3}{w,3} \notag \\ 
     & + \kb{c}{c} \ + \ \kb{nc}{nc}, 
 \end{align} where $\ket{0,0,0}$ is $\mathrm{Id}$, $\ket{w,1} \ = \ \frac{1}{\sqrt{N^2-1}} \ \sum_{a \in \Fs} \ \ket{0,a,a}$,  $\ket{w,2} \ = \ \ \sum_{a \in \Fs} \  \frac{1}{\sqrt{N^2-1}} \ \ket{a,0,a}$,  $\ket{w,3} \ = \ \frac{1}{\sqrt{N^2-1}} \  \sum_{a \in \Fs} \  \ket{a,a,0}$,  $\ket{\mathrm{C}} \ := \ \dfrac{\sqrt{2}}{\sqrt{(N^2-1)(N^2-4)}} \sum_{\substack{a,b \in \Fs \\ \s{a}{b}=0\\ a,b,a+b \neq 0}} \ \ket{a,b,a+b}$ and  $\ket{\mathrm{NC}} \ := \dfrac{\sqrt{2}}{N\sqrt{N^2-1}} \ \sum_{\substack{a,b \in \Fs \\ \s{a}{b}=1}} \ \ket{a,b,a+b}$. \end{lem}
 
 The proof is given in Subsection \ref{subsec:appendix:pf_lem_GH_eigenspaces}.
 
\begin{lem}
\label{lem:eigensystem_ET_EP_t_2}
For $t=2$, the eigendecomposition of $\ET \circ \EP$ is given as follows.
\begin{align}
    \label{eq:w}
    \ET \circ \EP \ = \ \H \ + \ \frac{1}{2}\left(1+\frac{1}{N}\right) \  \mathcal{P}_{d,+} \ +  \frac{1}{2}\left(1-\frac{1}{N}\right) \  \mathcal{P}_{d,-},
\end{align}
where $\ket{0,0}$ is $\unity$, $\ket{w}$ is given in Lemma \ref{lem:GH_2_eigenspace}, $\mathcal{P}_{d,\pm}$ are projectors on the subspace spanned by vectors $\ket{v_\pm}$ with the following properties. 
\begin{equation}
    \label{eq:v_pm_t_2}
    \ket{v_\pm} \ = \ \sum_{a \in \Fs} \ \lambda_a  \ \ket{a,a},
\end{equation}
where $\sum_{a \in \Fs} \ \lambda_a \ = 0$ and for all $a \in \Fs$, we have  $\sum_{\substack{b \in \Fs \\ \s{a}{b}=1}} \ \lambda_b \ = \ \pm \dfrac{N}{2} \ \lambda_a$. $\mathcal{P}_{d,\pm}$ has the rank $\frac{N(N \mp 1)}{2}-1$.
\end{lem} The proof is given in Subsection \ref{lem:pf_eigensystem_ET_EP_t_2} of the Appendix. Note that we obtained the full decomposition of $\ET\circ \EP$ because it will be needed in the proof of Theorem \ref{thm:asymptotic_unitary_2_design}.
 
\begin{lem}
\label{lem:unitary_3_design_parent_lemma}
When $t=3$, the largest eigenvalue of $\ET \circ \ET - \H$ is less than $\frac{1}{2}\left( 1  \ + \frac{4 }{N} \ + \  \frac{1}{2N(N-2)}\right)$.
\end{lem} The proof is given in Section \ref{sec:appendix:pf_lem:unitary_3_design_parent_lemma} of the Appendix.
\noindent
 
\subsection{Convergence Analysis}
\label{subsec:cgt_analysis}
    
\begin{proof}[Proof of Theorem \ref{thm:asymptotic_unitary_2_design}]
Note that 
\begin{align}
    & \ET^k \ \circ \EP - \H  \notag \\
    =  &  \  \left(  \ET^k \ \circ \EP - \H  \right) \ \circ \left( \id - \H + \H \right) \  \notag \\
    = & \  \left(  \ET^k \ \circ \EP - \H  \right) \circ \H +  \left(  \ET^k \ \circ \EP - \H  \right) \circ \left( \id - \H \right) \  \notag \\
    \label{eq:pf_thm_unitary_2_design_exp1}
     = & \  \ET^k \ \circ \EP  \circ \left( \id - \H \right), \\
     \label{eq:pf_thm_unitary_2_design_exp1_5}
     = & \  \frac{1}{2^k} \  \left( \ \left( 1+ \frac{1}{N}\right)^k \ \pdp  +   \left( 1 - \frac{1}{N}\right)^k \ \pdm \ \right) \circ \left( \id - \H \right)  \\
     \label{eq:pf_thm_unitary_2_design_exp2}
     = & \  \frac{1}{2^k} \left( \ \nuv \left( \pdp + \pdm \right) + \muv \  \frac{\pdp - \pdm}{N} \   \right) \circ \left( \id - \H \right)  \\
     \label{eq:pf_thm_unitary_2_design_exp3}
     = & \  \frac{1}{2^k} \left( \begin{array}{c} 
     (\nuv-\muv) \left(\H + \pdp + \pdm \right)    \\ 
     +  2\muv  \left( \H + \frac{1}{2} \left(\pdp + \pdm\right) + \dfrac{1}{2N} \left( \pdp - \pdm\right)  \right) \end{array}    \right) \circ \left( \id - \H \right)  \\ 
     \label{eq:pf_thm_unitary_2_design_exp4}
     = & \  \frac{1}{2^k} \left(
     (\nuv-\muv) \EP    
     +  2\muv  \ET\circ \EP\circ \left(\id - \H \right)  \right),
     \end{align}
     where we used the following:
     \begin{itemize}
         \item[Step \eqref{eq:pf_thm_unitary_2_design_exp1}:]  $ \ET\circ \EP \circ \H = \H $ as proved in Corollary \ref{corr:ET_EP_H_properties},
         \item[Step \eqref{eq:pf_thm_unitary_2_design_exp1_5}:] the eigendecomposition of $\ET \circ \EP $ as given in Lemma \ref{lem:eigensystem_ET_EP_t_2}. Here we also use the fact that $\left( \unity - \H \right) \ket{0,0} = \left( \unity - \H \right) \ket{w} = 0$ (see Lemmas \ref{lem:GH_2_eigenspace} and \ref{lem:eigensystem_ET_EP_t_2}).  
         \item[Step \eqref{eq:pf_thm_unitary_2_design_exp2}:] \begin{align}
             \label{def:nuv}
         &    \nuv \  := \ \frac{1}{2} \left(  \left( 1 + \frac{1}{N} \right)^k + \left( 1- \frac{1}{N} \right)^k \right) \\ 
             \label{def:muv}
         &    \muv \  := \ \sum_{l=0}^{k-1} \left( 1+ \frac{1}{N}\right)^l\left(1-\frac{1}{N} \right)^{k-1-l},
         \end{align}
         \item[Step \ref{eq:pf_thm_unitary_2_design_exp3}:] we add and subtract $\H$ because this $\H$ will cancel with $\id - \H$ on the right, and since $\H^2 = \H$ (see Lemma \ref{lem:eigensystem_ET_EP_t_2}).
         \item[Step \ref{eq:pf_thm_unitary_2_design_exp4}:] we used Lemma \ref{lem:eigensystem_ET_EP_t_2}.
     \end{itemize}
     Then we get 
     \begin{align}
   &    \left| \left| \ET^k \ \circ \EP - \H   \right|\right|_{\Diamond}   \notag \\ 
      = \ & \left| \left|  \frac{1}{2^k} \left(
     (\nuv-\muv) \EP    
     +  2\muv  \ET\circ \EP\circ \left(\id - \H \right)  \right)\right|\right|_{\Diamond}  \ \mathrm{(using} \ \mathrm{step} \ \mathrm{\eqref{eq:pf_thm_unitary_2_design_exp4})} \notag\\
     \label{eq:pf_thm_unitary_2_design_diamond_norm_step1}
     \le \ & \ \frac{|\muv - \nuv|}{2^k} \ \left| \left| \EP \right| \right|_{\Diamond} + \frac{ \muv}{2^{k-1}} \left| \left| \H - \ET\circ \EP \right|\right|_{\Diamond}  \\ 
     \label{eq:pf_thm_unitary_2_design_diamond_norm_step2}
     \le & \ \frac{|\muv - \nuv|}{2^k} + \frac{\mu}{2^{k-2}} \\
     \le & \  \frac{5\muv+\nuv}{2^k}, \notag \\ 
     \label{eq:pf_thm_unitary_2_design_diamond_norm_step3}
     < & \ \left(5k + 1\right)\left( \frac{1+\frac{1}{N}}{2}\right)^k,
     \end{align}
where we use the following. \begin{itemize}
    \item[Step \ref{eq:pf_thm_unitary_2_design_diamond_norm_step1}:] 
 proved using norm property of diamond norm on quantum channels. 
 \item[Step \ref{eq:pf_thm_unitary_2_design_diamond_norm_step2}:] For any quantum channel, $\mathrm{\Phi}$, $\left| \left| \mathrm{\Phi} \right| \right|_{\Diamond}=1$, and for any two quantum channels $\mathrm{\Phi}$ and $\mathrm{\Psi}$, $\left| \left| \mathrm{\Phi}-\mathrm{\Psi}\right|\right|_{\Diamond}\le 2$.
 \item[Step \ref{eq:pf_thm_unitary_2_design_diamond_norm_step3}:] Using
 \begin{align}
     \label{eq:pf_thm_unitary_2_design_diamond_norm_step3_1}
     \muv \ < k \ \left(1+\frac{1}{N} \right)^k, \; \; \nuv \ < \left(1+\frac{1}{N}\right)^k.
 \end{align}
 \end{itemize}
 Next, for given $\epsilon$, where $\epsilon > 0$, if 
 \begin{equation}
     \label{eq:pf_thm_unitary_2_design_diamond_norm_epsilon}
     \left( 5k+1\right) \ \left( \dfrac{1+\frac{1}{N} }{2}\right)^k \ \le \epsilon, 
 \end{equation}
 then 
 \begin{align}
     \label{eq:pf_thm_unitary_2_design_diamond_norm_epsilon_step1}
   \log_2 1 /\epsilon \ \le \ \left( 1- \log_2\left(1+\frac{1}{N} \right)\right) k \ - \ \log_2\left(5k+1\right).
 \end{align}
 For $m=3$ one sees that the RHS in the inequality \eqref{eq:pf_thm_unitary_2_design_diamond_norm_epsilon_step1} is positive for $k=6$. Graphically, one sees that 
 \begin{equation}
 \label{eq:pf_thm_unitary_2_design_diamond_norm_graph}
 \left( 1- \log_2\left(1+\frac{1}{N} \right)\right) k \ - \ \log_2\left(5k+1\right) \ > 0.8 \ \left( \ k - 6 \right), \, \mathrm{k} \ \ge \ 6,
 \end{equation}
 which suggests that  if 
 \begin{equation} 
 \label{eq:pf_thm_unitary_2_design_diamond_norm_linear_approx}
k \ \ge \  6 \ + \  5/4 \ \log_2 1/\epsilon,
 \end{equation}
 then that automatically ensures that the inequality \eqref{eq:pf_thm_unitary_2_design_diamond_norm_epsilon_step1} is satisfied and we have that $\left| \left| \ET^k \circ \EP - \H \right| \right|_{\Diamond} < \epsilon$ when $m\ge 3$. Hence proved.
\end{proof}

\begin{proof}[Proof of Theorem \ref{thm:asymptotic_unitary_3_design}]
Let $\lambda$ be the second largest eigenvalue of $\ET \circ \EP - \H$. Then we get
\begin{align}
    \left| \left| \  \ET^k \circ \EP - \H \  \right| \right|_{\Diamond} \ = & \ \underset{\ \rho \ }{ \mathrm{sup} } \ \ \left| \left| \ \ \  \left[ \left( \ \ET^k \circ \EP - \H  \right) \otimes \id \right] \ \left( \rho \right) \ \ \  \right| \right|_1 \notag \\ 
    \label{eq:pf_unitary_3_design_diamond_norm_step1}
    \le & \ N^3 \  \left| \left| \ \ \  \left[ \left( \ \ET \circ \EP - \H  \right)^k \otimes \id \right] \ \left( \rho \right) \ \ \  \right| \right|_2     \\ 
        \label{eq:pf_unitary_3_design_diamond_norm_step2}
    \le & \ N^3 \  \op{ \ET \circ \EP - \H  }^k  \\ 
    \label{eq:pf_unitary_3_design_diamond_norm_step3}
    = & \ N^3 \lambda^k,
\end{align}
where we used the following.
\begin{enumerate}
    \item[Step \ref{eq:pf_unitary_3_design_diamond_norm_step1}:] Proved using a well-known inequality between the $||.||_1$ and $||.||_2$ norms on linear operators on finite dimensional spaces, and also the fact that $\ET$, $\EP$ and $\H$ commute with each other, and that $\EP$ and $\H$ are projectors (see Corollary \ref{corr:ET_EP_H_properties}. 
    \item[Step \ref{eq:pf_unitary_3_design_diamond_norm_step2}:] For $\mathrm{\Delta} = \left( \ \ET \circ \EP - \H  \right) \otimes \id$ we use the fact that 
    \begin{equation}
        \label{eq:op_norm_inequality}
         \left| \left| \   \mathrm{\Delta}^k \ \left( \rho \right) \  \right| \right|_2 \ \le \ \op{\mathrm{\Delta}^k} \ . ||\rho||_2  \ \le \ \op{\mathrm{\Delta}}^k \ . ||\rho||_2 ,
    \end{equation}
    and the fact that $||\rho||_2=1$.
    \item[Step \ref{eq:pf_unitary_3_design_diamond_norm_step3}:] That $\lambda = \op{ \left( \ET \circ \EP - \H \right)\otimes \mathrm{id} } = \op{\ET \circ \EP - \H} $.
\end{enumerate} \noindent From inequality \eqref{eq:pf_unitary_3_design_diamond_norm_step3} we see that when, for any $\epsilon >0$, we have that 
\begin{equation}
    \label{eq:pf_unitary_3_design_diamond_norm_epsilon}
    N^3 \ \lambda^k \ \le \epsilon, 
\end{equation} then 
\begin{align}
    & 3m + \log 1/\epsilon \  \le \ k \log_2 \lambda^{-1}, \notag \\ 
    \mathrm{i.e.,} \ & 3m + \log 1/\epsilon \ \le  \ \left( 1 - \log \left(\text{ $1$ $+$ {\scriptsize $\frac{4}{N}$} + {\scriptsize $\frac{1}{2 N (N-2)}$} \normalsize }\right) \right) k ,
\end{align}
where we used Lemma \ref{lem:unitary_3_design_parent_lemma} for $m \ge 3$, and we see that the highest eigenvalue of $\ET\circ\EP-\H$ is less than $\frac{1}{2}\left(1+ \frac{4}{N} + \frac{1}{2N(N-2)}\right)$. 
\end{proof}

\begin{proof}[Proof of Corollary \ref{corr:order_steps}]
Since $\EP$ and $\ET$ commute (as seen from Corollary \ref{corr:ET_EP_H_properties}), reversing the order between $\EP$ and $\ET^k$ or implementing $\EP$ at some intermittent point between successive iterations of the $k$ time $\ET$ iteration, doesn't change the proofs for Theorem \ref{thm:asymptotic_unitary_2_design} and Theorem \ref{thm:asymptotic_unitary_3_design}.
\end{proof}

\section{Discussion and Conclusion}
\label{sec:conclusion}


We have proposed a scheme to implement an asymptotic unitary $2$-design and a unitary $3$-design. In Remark \ref{rem:optimal_convergence_rate}, we argue that this scheme is optimal, in terms of the number of times it has to be implemented for both, approximating a unitary $2$- and a unitary $3$-design within $\epsilon$-distance in diamond norm from the corresponding exact designs. The proof relies on the Frobenius-Perron theorem for irreducible Markov chains (which the transvections channel is), and the fact that the second largest eigenvalue of the scheme's channel is below a constant which is independent of $m$. We anticipate that this scheme would be of interest for certain architectures,  like those of ion traps, which can support multiple qubit interactions. In the process we also obtain some of the eigenspaces of the channel $\ET \circ \EP$, and discovered that these eigenspaces are, in fact, invariant subspaces of the action of the three-copy of the adjoint representation of the Clifford group. Thus, as a corollary of this work, we obtain some of the invariant subspaces of the three-copy adjoint representation of the Clifford group. Unsurprisingly, some of the irreps which are also found in the two-copy adjoint representation of the Clifford group in \cite{Helsen2016} occurred multiple times in different subspaces of the three-copy case. Finally we also obtained some new invariant subspaces of the Clifford group adjoint action in the three-copy case. While some are clearly reducible, we are still not sure about the others. Going forward, we identify the following questions, which we would like to seek answers to. The first three of these questions stem from our inability to give a complete decomposition into irreducible subspaces of the three-copy adjoint action of the Clifford group, when it acts on the $\mathrm{span} \left\{ \ket{a,b,c} \ |  \ a,b,c \in \Fm, \ \mathrm{such} \ \mathrm{that} \ : \ a+b+c =0 \right\}$.
\begin{enumerate}
    \item[(i)] The irreducibility of invariant subspace which were obtained: How do we know if the invariant subspaces which we obtained are irreducible or not? For comparison, the proof of the irreducibility of subspaces in \cite{Helsen2016} relied on a result in \cite{Zhu15} (see equations (8) to (11) and also, equation (12)). Computing the frame potentials for the three copy adjoint representation, one would need to compute $\Phi_6$. This is quite a difficult task. Thus, one would like to know if there is a simpler alternative to proving the irreducibility or not, of the corresponding invariant subspaces.
    \item[(ii)] The invariant eigenspaces in Theorem \ref{thm:nc_inv_subsp} and \ref{thm:nc_1_decomposition} lie in the subspace $\mathrm{span}\left\{ \ket{a,b,a+b} \right\}$ where $\s{a}{b}=1$. The case of $\mathrm{span}\left\{ \ket{a,b,a+b} \right\}$ where $\s{a}{b}=0$ with $a,b,a+b \neq 0$ remains to be done. 
    
    \item[(iii)] Since $\EP$ annihilates any Pauli basis elements of the form $\ket{a,b,c}$ where $a,b,c$ are not related by $a+b+c=0$, the decomposition of a large subspace of operators remains to be explored. We suspect that the structure for the $a+b+c=0$, though, is significantly richer (i.e., there are more subspaces in the $a+b+c=0$ case, for the size of its dimension). 
    \item[(iv)] The irreducible subrepresentations of the Clifford group for $t=4$ and higher $t$ is also something one would like to explore. 
    \item[(v)] Applications of the representations of the Clifford group in various carrier spaces.
\end{enumerate}
\noindent 




  \appendix
  
\section{Proofs in Section \ref{sec:pfs}}
  
  \subsection{Proof of Theorem \texorpdfstring{\ref{thm:pmb_fm_iso}}{5}}
  \label{subsec:ph_thm_pmb_fm_iso}

\begin{proof}[Proof of Theorem \ref{thm:pmb_fm_iso}] We prove the result by constructing the isomorphism. This construction maps a \emph{symplectic basis} for $\Fm$ to some choice of generators of $\Pmb$. A convenient choice is as follows: $e_j$ denotes a vector in $\Fm$ whose only non-zero component is the $j$-th component. Thus the $e_j$'s give a standard basis for $\Fm$, and any $a\in \Fm$ can be expanded in terms of this basis. Let $X_j$ and $Z_j$ be the single-qubit $X$ and $Z$ operator for all $j=1$ to $m$, i.e., 

\begin{equation}
    \label{eq:Xj}
    X_j  \ = \ \underbrace{\unity_2 \otimes \unity_2 \cdots \otimes \unity_2}_{j-1} \otimes X \otimes \underbrace{\unity_2 \otimes \unity_2 \cdots \otimes \unity_2}_{m-j},
\end{equation}
and similarly for $Z_j$. Note that $X_j$'s and $Z_j$'s all belong to $2m$ distinct cosets of $\Pmb$. We choose them to be the representatives of their respective cosets in $\Pmb$. The $X_j$'s themselves are a set of $m$ independent generators for an abelian subgroup of $\Pmb$. Let $\Fl = \ \mathrm{span} \left\{ e_j \right\}_{j=1}^m$. Then for any $a \in \Fl$, for the expansion $a=\sum_{j=1}^{2m} a_j e_j$, we find that $a_{l}=0$, for $l=m+1,m+2,\cdots,2m$. $\Fl$ is $m$ dimensional. We map any $a \in \Fl$ to a Pauli generated only by the $X_j$'s in the following way:
\begin{equation}
    \label{eq:Xa}
    X(a) \ := \ X_1^{a_1} X_2^{a_2} \cdots X_m^{a_m},
\end{equation} \noindent Similarly, define $\Fr := \mathrm{span} \left\{ e_j \right\}_{j=m+1}^{2m}$. So when $b \in \Fr$, for the expansion $b = \sum_{j=1}^{2m} b_j e_j$, $b_j=0$ for $j=1,2,\cdots, m$, and $\dim \Fr=m$. For any $b \in \Fr$, define a Pauli generated only by the $Z_j$'s as follows.
\begin{equation}
    \label{eq:Zb}
    Z(b) \ := Z_1^{b_{m+1}} Z_2^{b_{m+2}} \cdots Z_m^{b_{2m}},
\end{equation} Next, note that $\Fm = \Fl \oplus \Fr$, hence for any $a \in \Fm$, there exists a unique $a_X \in \Fl$ and $a_Z \in \Fr$, such that $a=a_X + a_Z$. Then for any $a$, with the decomposition $a= a_X+a_Z$ as above, define the following:
\begin{equation}
\label{eq:Pauli_F2_isomorphism}
E(a) \ := \ i^{a_X.a_Z} \ X(a_X) \ Z(a_Z),
\end{equation}
where $a_X.a_Z = \sum_{j=1}^{m} \  \left(a_X\right)_j \left(a_Z\right)_{m+j}$, but the sum is taken in $\mathbb{Z}$ and not in $\F$. $E(a)$ is always self-adjoint, since $a_X.a_Z$ is even if and only if $\s{a_X}{a_Z}=0$, i.e., if and only if $X(a_X)$ and $Z(a_Z)$ commute. As linear operators in $\mathcal{B}\left(\C^N\right)$, we find that $E(a)$ and $E(b)$ in $\Pmb$ compose in the following way.
\begin{align}
    \label{eq:composition_law_pmbb}
    E(a) \ E(b) \ & = \  i^{a_X.a_Z+b_X.b_Z} \  X(a_X) Z(a_Z) \ X(b_X) Z(b_Z)  \notag \\ 
     & = \  i^{a_X.a_Z+b_X.b_Z + 2 a_Z.b_X} \  X(a_X+b_X) Z(a_Z+b_Z)  \notag \\ 
    & = i^{a_Z.b_X - a_X.b_Z} \ E(a +b ).
\end{align} \noindent This proves that the map in equation \eqref{eq:Pauli_F2_isomorphism} is an isomorphism between $\Fm$ and $\Pmb$.
 \end{proof}

\subsection{Proofs of Lemma \ref{lem:1}}
\label{subsec:appendix_lem_1_pf}
\noindent We revise the terminology before we do the proof. For any linear operator $A:\Fm \longrightarrow \F^l$, for some positive integer $l$:  $\mathrm{rank}A := \dim \range{A}$, and $\mathrm{nullity}A:=\dim \mathrm{kernel}A$. The rank-nullity theorem tells us that $2m= \mathrm{rank}A + \mathrm{kernel}A$.

\begin{proof}
Consider the following linear map: $a^*: \ \Fm \ \longrightarrow \F$,
\begin{equation}
    \label{eq:1form}
    a^*( x) = \s{a}{x}.
\end{equation}
Since $a\neq 0$, it is seen that $a^*$ is a rank-one linear operator (it is a one-form) on $\Fm$. The rank-nullity theorem implies that $\mathrm{nullity} \ a^* \ = \ 2m-1$ dimensional. So $\mathrm{kernel} \ a^*$ has $N^2/2$ distinct $h$ in $\Fm$ such that $\s{a}{h}=0$, and all remaining $h$ in $\Fm$ are such that $\s{a}{h}=1$. 

\noindent For the more general case: define the following linear map $\left( a_1^*,a_2^*,\cdots,a_l^* \right): \Fm \longrightarrow \F^l$ with the action
 \begin{equation}
    \label{eq:2form}
    \left( a_1^*,a_2^*,\cdots,a_l^* \right) (x) \ = \ \left(a_1^*(x), a_2^*(x), \cdots, a_l^*(x) \right).
 \end{equation}
\noindent Since $\dim \mathrm{span}\left\{ a_1,a_2,\cdots, a_l \right\} =d$, $\mathrm{rank}\left( a_1^*,a_2^*,\cdots,a_l^* \right)=d$ and $\mathrm{nullity}\left(a_1^*,a_2^*,\cdots,a_l^* \right)=2m-d$. When $b\notin \range{\left(a_1^*,\cdots,a_l^*\right)}$, then there do not exist solutions $h\in\Fm$ such that $a_j^*(h)=b_j$ for $j=1,\cdots,l$. Else if $b \in \range{\left(a_1^*,\cdots,a_l^*\right)}$, then the number of solutions for each such $h$ in $\Fm$, which satisfies $a_j^*(h)=b_j$ is equal to the number of elements in the  kernel of $\left(a_1^*,\cdots,a_l^*\right)$, which is $N^2/2^d$. Hence proved.
\end{proof}

\subsection{Proof of Lemma \ref{lem:Phi_adjoint}}
\label{subsec:pf_lem_Phi_adjoint}
\begin{proof}
Let $Y,X \in \BCg$, then we have 
\begin{align}
    \label{eq:Phi_adjoint_pf}
 \bk{Y}{\mathrm{\Phi}(X)} \ =  \  &  \     \tr \left( Y^\dag \ \mathrm{\Phi}(X) \right) \notag \\
    =  \  &  \  \tr \ \left(   Y^\dag \  \  \sum_{k} F_k \ X F_k^\dag  \ \right) \notag \\ \ = \ & \tr \ \left( \ \left( \sum_k \ F_k^\dag \  Y^\dag \  F_k \ \right) \  X \ \right) \notag \\ 
    = \ & \tr \ \left( \ \left( \  \sum_k \ F_k^\dag \ Y \ F_k \right)^\dag \ X \  \right) \notag \\
    = \ & \tr \ \left( \ \left( \mathrm{\Phi}^\dag (Y) \right)^\dag \ X \  \right) \notag \\ 
    = \ & \ \bk{\mathrm{\Phi}^\dag (Y)}{X}.
\end{align}
\end{proof}

\subsection{Proof of Lemma \ref{lem:transvections_conjugacy_class}}
\label{subsec:appendix_transvection_conjugacy_class}
\begin{proof}
\noindent From equation \eqref{eq:Sp_condition} one can easily compute the inverse of any $S \in \Sp$,
 \begin{equation}
     \label{eq:Sinv}
     S^{-1} \ = \ \Om S^T \Om .
 \end{equation} Using equation \eqref{eq:Sinv}, it is seen that for any $S \in \Sp$, 
\begin{equation}
\label{eq:transvection_conjugation}
S \ T_h \ S^{-1} \, = \, T_{Sh}.
\end{equation}
Choose some $l \in \Fs$ such that $\s{l}{h}=1$. Then $T_{l+h}(h)=l$, choose $S=T_{l+h}$ on the LHS of equation \eqref{eq:transvection_conjugation}, so that the output is $T_{l}$. Thus for all $l \in \Fs$ such that $\s{h}{l}=1$, $T_l$ and $T_h$ lie in the same conjugacy class. Let $l' \in \Fs \setminus \left\{ h \right\}$ be such that $\s{h}{l'}=0$. The set of $l \in \Fs$ such that $\s{h}{l}=1$, $\s{l'}{l}=1$ is $N^2/4$ (from Lemma \ref{lem:1}). Thus we see that $T_{l'}$ too lies in the same conjugacy class. Hence proved.
\end{proof}

\subsection{Proof of Corollary \ref{corr:ET_EP_H_properties}}
\label{subsec:pf_corr_ET_EP_H_properties}
That $\H$ is self-adjoint follows by the invariance of the Haar measure with respect to inverse, i.e., $, \forall X \in \BCg$
\begin{align}
\label{eq:Haar_invariance_inverse}
\H^\dag \ket{X}  = & \ \int_{\U(N)} \ d\mu(U) \ \Ua{U^\dag} \ket{X} \notag  \\
= & \ \int_{\U(N)} \ d\mu\left( U^\dag \right) \ \Ua{U^\dag} \ \ket{X}  \notag \\
= & \  \int_{\U(N)} \ d\mu(U) \ \Ua{U} \ \ket{X} \notag \\ 
= &  \ \H \ \ket{X}.
\end{align}
Since the Haar-measure is also invariant under translation we have for all $W \in \U(N)$ and $X \in \BCg$:
\begin{equation}
    \label{eq:Haar_measure_translation}
    \Ua{W} \ \H  \ \ket{X} \ = \ \int_{\U(N)} d\mu(U) \Ua{W U} \ \ket{X} \ = \ \int_{\U(N)} d\mu(WU) \ \Ua{WU} \ \ket{X} \ = \ \H \ket{X}.
\end{equation} One can similarly prove $\H \Ua{W} = \H$ using the right translational invariance. This immediately proves $\mathcal{G}_\mathcal{E} \circ \H=\H$ for \emph{any} ensemble of unitary operators $\mathcal{E}$. In particular, when $\mathcal{E}$ is the Haar ensemble, then we get $\H^2 = \H$. Hence $\H$ is a self-adjoint projector. Equation \eqref{eq:ET_EP_H} is also proved by choosing $\mathcal{E}$ to be $\mathcal{P}$ and $\mathcal{T}$. Equation \eqref{eq:adjoint_action_EP} tells us that $\EP$ is a projector, and since it's Kraus operators are self-adjoint, Corollary \ref{corr:selfadjoint_kraus} tells us that $\EP$ is a self-adjoint projector. For the adjointness of $\ET$, it is sufficient that the Kraus operators are either self-adjoint or skew-adjoint as per Corollary \ref{corr:selfadjoint_kraus}, and we chose the transvection Cliffords to be self-adjoint in Definition \ref{def:ensemble_transvection}. To prove Eq. ~ \eqref{eq:commutation_EP_Clifford} we examine the action of $\ET \circ \EP$ on two orthogonal subspaces: $\mathrm{span} \ \left\{ \ket{a_1,a_2,\cdots,a_t} \ \vert \ \sum_{j=1}^t a_j =0 \right\}$, and $\mathrm{span} \ \left\{ \ket{a_1,a_2,\cdots,a_t} \ \vert \ \sum_{j=1}^t a_j =0 \right\}$. First for any non-zero $h \in \Fm$, and any $\ket{a_1,a_2,\cdots,a_t}$ in the former subspace, we see that  $\Ua{h} \ket{a_1,a_2,\cdots, a_t}$ belongs to the same former subspace. This is because if $\sum_{j=1}^t a_j=0$, then $T_h\left(\sum_{j=1}^t a_j\right)$. Also note that $T_h^2=\unity$, so if $T_h \left( \sum_{j=1}^t a_j \right)=0$, then that must imply that $\sum_{j=1}^t a_j=0$ as well. Similarly if $\ket{a_1,a_2,\cdots,a_t}$ belongs to the second subspace, then $\Ua{h} \ket{a_1,a_2,\cdots,a_t}$ also belongs to the second subspace. This implies that the both subspaces remain invariant under the action of $\ET$. From Eq. ~ \eqref{eq:adjoint_action_EP} we see that $\EP$ projects onto the first space (and annihilates the second). Thus for any $\ket{a_1,a_2,\cdots,a_t}$ in the first space,we see that $\EP \circ \ET \ket{a_1,\cdots,a_t}  \ \ET \ket{a_1,\cdots,a_t} = \ \ET \circ \EP \ket{a_1,\cdots, a_t} $. Also for $\sum_{j=1}^t a_j \neq 0$, we see that $\EP \circ \ET \ket{a_1,\cdots,a_t}  = 0 = \ \ET \circ \EP \ket{a_1,\cdots, a_t}$. Thus $\EP \circ \ET = \ET \circ \EP$. Using the homomorphism from $\Cl$ to $\Sp$ and the adjoint action of the Clifford group on the Paulis, which is shown in equation \eqref{eq:Clifford_adjoint_action}, we see that for any $g \in \Cl$, equation \eqref{eq:transvection_conjugation} tells us that $\Ua{g}\Ua{h}\Ua{g^{-1}} \ = \  \Ua{S_g h}$, where $S_g \in \Sp$ is the symplectic matrix corresponding to $g \in \Cl$. This is true because if $g U_h g^\dag = e^{i \phi} U_{S_g h}$, then \[ g \left(  U_h X U_h^\dag \right) g^\dag \ = \  \left (g U_h g^\dag\right) g \ X \ g^\dag  \left(g U_h^\dag g^\dag\right) \ = \ U_{S_g h} \  \left(  g X g^\dag \right) \  U_{S_g h}^\dag. \] Since $\mathrm{rank}S_g=2m$, $S_g$ is a bijection from $\Fm \rightarrow \Fm$, we see that
\begin{equation}
    \label{eq:ET_Ug_conjugation}
    \ET \circ \Ua{g} \ = \ \Ua{g} \circ \ET, 
\end{equation}
since $\ET$ is the uniform sum of all transvections. Using the fact that $\ET \circ \EP = \EP \circ \ET$, we prove equation \eqref{eq:commutation_EP_Clifford}. Finally, equation \eqref{eq:Haar_Clifford_relation} is proved in previous works, for instance \cite{Zhu15, Webb15}. For the case of $t=2$, the proof was given earlier in \cite{Dankert2009}.

\subsection{Proof of Lemma \ref{lem:Pauli_basis_choice}}
\label{subsec:appendix_Clifford_action_special_basis}
\noindent Let $a_1,a_2,\cdots,a_t$ be such that $a_t = \sum_{j=1}^{t-1} a_j$. This tells us that $E(a_t) = i^x \ E(a_1) E(a_2)\cdots E(a_{t-1})$ for some $x = 0,1,2,3$. Instead of choosing our basis element representing $\ket{a_1,a_2,\cdots,a_t}$ to be $E(a_1)\otimes E(a_2) \otimes \cdots \otimes E(a_t)$, we choose it to be 
\begin{equation}
    \label{eq:Pauli_special_basis_cases}
    \ket{a_1,a_2,\cdots,a_t} \ = \
       i^{y} \ 
     \left(   \bigotimes_{j=1}^{t-1}  E(a_j) \right)  \otimes \left( i^{-x} E(a_t) \right).
\end{equation}
Assume that we adopt this convention for all choices of $a_1, a_2, \cdots, a_{t-1}$. $a_t = \sum_{j=1}^{t-1} a_j$, so it's already decided once $a_1$, $a_2$, $\cdots$, $a_{t_1}$ are decided. We use the phase factor $i^{y}$ to ensure that $\ket{a_1,a_2,\cdots,a_t}$ corresponds to a hermitian operator: so $y=0$ when $ \left(   \bigotimes_{j=1}^{t-1}  E(a_j) \right)  \otimes \left( i^{-x} E(a_t) \right)$ is hermitian and $y=1$ when $ \left(   \bigotimes_{j=1}^{t-1}  E(a_j) \right)  \otimes \left( i^{-x} E(a_t) \right)$ is anti-hermitian. Let $g \in \Cl$ and let $g E(a_j)g^\dag = (-1)^{z_j} E\left( S_g \ a_j \right)$ for $j=1,2,\cdots, t-1$. Then we see that 
\begin{align}
    \label{eq:no_phase}
 & \Ua{g} \ \ket{a_1,a_2, \cdots, a_t} \notag \\  = & \ i^{y} \ \left( \bigotimes_{j=1}^{t-1} g \ E\left(a_j\right) g^\dag \right) \cdots  \ \left(  i^{-x} g E\left(a_t\right) g^\dag \right) \notag \\
 = & \ (i)^{y} \ \left( (-1)^{\sum_{j=1}^{t-1} z_j} \bigotimes_{j=1}^{t-1}  \ E\left(S_g a_j\right)  \right) \cdots  \ \left( (-1)^{\sum_{j=1}^{t-1} z_j} E(S_g a_1) \cdots E(S_g a_{t-1}) \right) \notag \\
     \ = & \ i^{y} \ \left( \bigotimes_{j=1}^{t-1} \ E\left(S_g a_j\right) \right)  \otimes  \left(  E \left(S_g a_1\right)  \cdots E\left(S_g a_{t-1}\right)\right) \\
     \ = & \ket{S_g a_1, S_g a_2, \cdots, S_g a_t },
\end{align}
where we note that \[  \ket{S_g a_1, S_g a_2, \cdots, S_g a_t } =  i^{y} \ \left( \bigotimes_{j=1}^{t-1} \ E\left(S_g a_j\right) \right)  \otimes  \left(  E \left(S_g a_1\right)  \cdots E\left(S_g a_{t-1}\right)\right).\]
This is because of our choice of convention adopted in Eq. ~  \eqref{eq:Pauli_special_basis_cases}, and also since hermiticity (and skew-hermiticity) is preserved under conjugation, which implies that \[ i^{y} \ \left( \bigotimes_{j=1}^{t-1} \ E\left(S_g a_j\right) \right)  \otimes  \left(  E \left(S_g a_1\right)  \cdots E\left(S_g a_{t-1}\right)\right)\] will always be hermitian since $i^{y} \ \left( \bigotimes_{j=1}^{t-1} g \ E\left(a_j\right) g^\dag \right) \cdots  \ \left(  i^{-x} g E\left(a_t\right) g^\dag \right)$ was hermitian to begin with.
Thus finally,
\begin{equation}
    \label{eq:no_phase_2}
    \Ua{g} \ \ket{a_1,a_2,\cdots,a_t} \ = \ \ket{S_g a_1, \ S_g a_2, \ \cdots, S_g a_t }. 
\end{equation}
Hence proved.

\subsection{Proofs of Lemmas \ref{lem:GH_2_eigenspace} and \ref{lem:GH_3_eigenspace}}
\label{subsec:appendix:pf_lem_GH_eigenspaces}
Proof of Lemma \ref{lem:GH_2_eigenspace}
\begin{proof}
 It is straightforward to verify that $\H \left(\unity\right) = \unity$. 
 \noindent Next, note that $\bk{0,0}{w}=0$. To prove that $\H \ket{w} = \ket{w}$, we use equation \eqref{eq:action_H}. Note that for any $a \in \Fs$, the set $S^{(a)}:=\left\{ S \in \Sp \ | \ S a = a \right\}$ is a subgroup of $\Sp$. Note that this subgroup contains the identity, hence it is not empty. Furthermore, for any $b \in \Fs$, there is a distinct left coset of $S^{(a)}$ in $\Sp$, such that if any $S$ in the left coset, $Sa=b$. Hence there are $N^2-1$ distinct left cosets of $S^{(a)}$ in $\Sp$, and symplectic matrices in each distinct left coset take $a$ to some distinct $b \in \Fs$.  Thus for $t=2$, equation \eqref{eq:action_H} becomes 
 \begin{equation}
     \label{eq:action_H_t_2}
     \H \ket{a,a} \ = \ \dfrac{1}{N^2-1} \sum_{b \in \Fs} \ket{b,b},
 \end{equation}
 and summing over all $a$'s in $\Fs$ (there are $N^2-1)$ of them, and multiplying by the normalising factor $\frac{1}{\sqrt{N^2-1}}$ gives us that $\H \ket{w} = \ket{w}$.
 \end{proof}
 Proof of Lemma \ref{lem:GH_3_eigenspace}:
\begin{proof}
 Firstly, note that the vectors $\ket{0,0,0}$, $\ket{w_j}$ for $j=1,2,3$, $\ket{\mathrm{C}}$ and $\ket{\mathrm{NC}}$ are mutually orthogonal, and each is normalized. In view of the proof of Lemma \ref{lem:GH_2_eigenspace}, we have only to prove that $\H \ket{\mathrm{C}} = \ket{\mathrm{C}}$ and $\H \ket{\mathrm{NC}} = \ket{\mathrm{NC}}$. We follow the same logic as in proof of Lemma \ref{lem:GH_2_eigenspace}. For any $a,b \in \Fs$ such that $a\neq b$, define a set $S^{(a,b)} := \left\{ S \in Sp \, | \, Sa=a \ \mathrm{and} \ Sb=b \right\}$. This is a subgroup of $\Sp$. Note that this subgroup contains the identity, hence it is not empty. For any $c,d \in \Fs$ such that $c \neq d$, and $\s{c}{d}=\s{a}{b}$, there is a distinct left coset of $S^{(a,b)} $ in $\Sp$ such that $Sa=c$ and $Sb=d$ for any $S$ in this left coset. For $\s{a}{b}=1$ there are $\frac{N^2(N^2-1)}{2}$ pairs of $c,d \in \Fs$ such that $c \neq d$ and $\s{c}{d}=1$. For $\s{a}{b}=0$ there are $\frac{(N^2-4)(N^2-1)}{2}$ pairs of $c,d \in \Fs$ such that $c \neq d$ and $\s{c}{d}=0$. Hence, from equation \eqref{eq:action_H} we get 
 \begin{equation}
     \label{eq:action_H_3}
     \H \ket{a,b,a+b} \ = \begin{cases}
      \frac{2}{N^2(N^2-1)} \sum_{\substack{c,d \in \Fs \\ c\neq d \\ \s{c}{d}=1}} \ket{c,d,c+d}, \ \mathrm{when} \ \s{a}{b}=1,   \notag \\ 
      \frac{2}{(N^2-4)(N^2-1)} \sum_{\substack{c,d \in \Fs \\ c\neq d \\ \s{c}{d}=0}} \ket{c,d,c+d}, \ \mathrm{when} \ \s{a}{b}=0.  
      \end{cases}
      \end{equation}
      from which we can easily obtain $\H \ket{\mathrm{C}}$ and $\H \ket{\mathrm{NC}}$. Hence proved.
 \end{proof}
 \subsection{Proof of Lemma \ref{lem:eigensystem_ET_EP_t_2}}
 \label{lem:pf_eigensystem_ET_EP_t_2}
\begin{proof}
Note that $\mathrm{span} \ \left( \left\{\ket{w}\right\} \cup \mathrm{supp} \pdm \ \cup \ \mathrm{supp} \ \pdp \right) \ = \ \mathrm{span} \left\{ \ket{a,a} \ | \ \mathrm{for} \ \mathrm{all} \ a \in \Fs \right\}$. Hence we have
\begin{equation}
    \label{eq:EP_t_2}
    \EP \ = \ \H \ + \ \pdp \ + \pdm.
\end{equation}
Also it is easy to see that $\pdp \H = \pdm \H =0$. Next we need to prove that
\begin{equation}
    \label{eq:ET_vpm_t_2}
    \ET \ \ket{v_{\pm}} \ = \ \frac{1}{2} \left( 1 \pm \frac{1}{N} \right) \ \ket{v_{\pm}}.
\end{equation} The proof of this is given in  $t=3$ case (see equation \eqref{eq:ET_action_vpm0}.
\end{proof}
\section{Proof of Lemma \ref{lem:unitary_3_design_parent_lemma}}
\label{sec:appendix:pf_lem:unitary_3_design_parent_lemma}

It will be important for us to introduce a little bit about the symmetry group $\mathrm{S}_3$. 

\subsection{Symmetry group \texorpdfstring{$\mathrm{S}_3$}{on three objects}}
\label{subsec:Sr}

$\mathrm{S}_3$ denotes the symmetry group on three objects. We list its properties below. 

\begin{itemize}
    \item $\left| \mathrm{S}_3 \right| = 6$.
    \item $\mathcal{A}_3 = \left\langle [(123)] \right\rangle = \left\langle [(132)] \right\rangle$ generates the only invariant subgroup of $\mathrm{S}_3$. It contains $3$ elements, all of which are even permutations. Since it is a cyclic group, it's abelian. 
    \item $\mathrm{S}_3 \setminus \mathcal{A}_3$ comprises of the set of transpositions, $\left\{[(12)],[(23)],[(13)]\right\}$. These are the odd permutations.
    \item There are three irreducible representations of $\mathrm{S}_3:$ 
    \begin{enumerate}
        \item[(i)] the completely symmetric representation of dimension one,
        \item[(ii)] the completely anti-symmetric representation of dimension one, and
        \item[(iii)]  the standard representation of dimension two.
    \end{enumerate} 
\end{itemize}

Let $\W{{S}_3}$ be a complex linear representation of $\Sr$ on the Hilbert space ${\C^N}^{\otimes 3}$. Its action is then defined as follows.
\begin{equation}
    \label{eq:S3_action}
    \W{\pi} \ \ket{\psi_1} \otimes \ket{\psi_2} \otimes \ket{\psi_3} \  = \ \ket{\psi_{\pi^{-1}(1)}} \otimes \ket{\psi_{\pi^{-1}(2)}} \otimes \ket{\psi_{\pi^{-1}(3)}}.
\end{equation}

The adjoint action of $\Sr$ on $\BCr$ is then given as \begin{align}
    \label{eq:adj_S3_action}
    ~ & \ \Wa{\pi} \left( X_1 \otimes X_2 \otimes X_3 \right) \ \notag \\
 =\  & \W{\pi} \ X_1 \otimes X_2 \otimes X_3 \ \W{\pi}^{-1} \notag \\ 
    =\ & \ X_{\pi^{-1}(1)} \otimes X_{\pi^{-1}(2)} \otimes X_{\pi^{-1}(3)}.
\end{align}
This is a linear representation of $\Sr$ on $\BCr$, and is easily seen to be unitary with respect to the inner product defined in \eqref{eq:inner_product_bcg}. 
\noindent Next, note the following for all $g \in \Cl$.
\begin{equation}
    \label{eq:schur_lemma}
    \Wa{\pi}  \circ \Ua{g} \ \ \left(  X \right) \ = \ \Ua{g} \ \circ \Wa{\pi}  \ \ \left(  X \right),
\end{equation}
for arbitrary $X \in \BCr$. By Schur's lemma we see that for each $\pi \in \Sr$, $\Wa{\pi}$ is linear combination of projectors on the irreducible representations of $\Cl$ on $\BCr$, and for each $g \in \Cl$, $\Ua{g}$ is a linear combination of projectors onto the irreducible representations of $\Sr$, whose action is defined by equation \eqref{eq:adj_S3_action}. We define the following projectors on this space as follows.
\begin{itemize}
    \item[(i)] 
    \begin{align} 
        \label{eq:pcs}
         \pcs \ := \ &  \ \dfrac{1}{\left| \Sr \right|} \  \sum_{\pi \in \mathrm{S}_3}  \Wa{\pi} \notag \\ 
        ~ = \  &  \ \dfrac{1}{6} \ \left( \id + \Wa{[(12)]} + \Wa{[(23)]} + \Wa{[(13)]} + \Wa{[(123)]} + \Wa{[(132)]}  \right). \notag \\ 
        ~ = \  &  \ \dfrac{1}{6} \ \left( \ \sum_{\sigma \in \mathcal{A}_3} \ \Wa{\sigma} + \sum_{\tau: {\small \mathrm{transposition}}} \ \Wa{\tau} \right).
    \end{align}
    \item[(ii)] 
    \begin{align} 
        \label{eq:pas}
         \pas \ := \ &  \  \dfrac{1}{\left| \Sr \right|} \  \sum_{\pi \in \mathrm{S}_3} \sgn{\pi}   \Wa{\pi} \notag \\ 
      ~ = \  &  \ \dfrac{1}{6} \ \left( \id - \Wa{[(12)]} - \Wa{[(23)]} - \Wa{[(13)]} + \Wa{[(123)]|} + \Wa{[(132)]}  \right) \  \notag \\ 
        ~ = \  &  \ \dfrac{1}{6} \ \left( \ \sum_{\sigma \in \mathcal{A}_3} \ \Wa{\sigma} \  -  \  \sum_{\tau: {\small \mathrm{transposition}}} \ \Wa{\tau} \right). 
    \end{align}
    \item [(iii)]
    \begin{align} 
        \label{eq:pw}
         \pw \  := &  \ \dfrac{1}{\left| \mathcal{A}_3 \right|} \  \left( \ \id \ + \ \omega \ \Wa{[(123)} \ + \ \omega^2 \Wa{[(132)]} \  \right),
    \end{align}
    \item[(iv)] 
    \begin{align} 
        \label{eq:pws}
         \pws \ := \ &  \ \dfrac{1}{\left| \mathcal{A}_3 \right|} \  \left( \ \id \ + \ \omega \ \Wa{[(132)} \ + \ \omega^2 \Wa{[(123)]} \  \right), 
    \end{align}
    where
    \begin{equation}
    \label{eq:def_omega}
    \omega \ := \ \mathrm{exp} \ \left( \frac{2 \pi i \ }{3} \right),
    \end{equation}
    is a third-root of unity (we could have chosen $ \mathrm{exp} \ \left( 4 \pi i / 3 \right) $ as well).
\end{itemize} 
\noindent That $\pcs$, $\pas$, $\pw$ and $\pws$ are indeed projectors is easily verified. We list some important properties of these projectors below. 
\begin{itemize}
    \item[Orthogonal decomposition:] The projectors defined in equations \eqref{eq:pcs}, \eqref{eq:pas}, \eqref{eq:pw} and \eqref{eq:pws} are all seen to be self-adjoint. Hence their spectral decomposition is into orthogonal eigenspaces. 
    \item[Irreducible representations of $\Sr$:]
    \begin{itemize}
    \item[Completely symmetric:] By the action defined in \eqref{eq:S3_action}, $\Sr$ acts trivially on $\range{\pcs}$. Thus, these vectors transform according to the trivial representation of $\Sr$. 
    \item[Completely antisymmetric space:] The representation of $\Sr$ on $\range{\pas}$ is the one-dimensional sign-representation. 
    \item[Standard representation:] The action of $\Sr$ on $\range{\pw}$ and on $\range{\pws}$ is the two-dimensional standard representation of $\Sr$. 
    \end{itemize}
\end{itemize}

\noindent Additionally define 
\begin{equation}
    \label{eq:ps}
    \ps \ := \ \pcs + \pas \ = \ \dfrac{1}{\left| \mathcal{A}_3 \right|} \ \sum_{\sigma \in \mathcal{A}_3} \ \W{\sigma},
\end{equation}
\noindent we see that
\begin{align}
    \label{eq:pcs_ps}
    \pcs \ = \ \dfrac{1}{2}\left( \id + \W{\tau} \right) \ \ps \ =  \ \dfrac{1}{2} \  \ps \ \left( \id + \W{\tau} \right),  \\
    \label{eq:pas_ps}
    \pas \ = \ \dfrac{1}{2}\left( \id - \W{\tau} \right) \ \ps \ =  \ \dfrac{1}{2} \  \ps \ \left( \id - \W{\tau} \right),
\end{align}
 where $\tau$ is a transposition in $\Sr$, and where the commutativity between $\ps$ and $\id \ \pm \Wa{\tau}$ is because $\ps$ is a uniform sum over a normal subgroup, which implies that its left cosets with respect to $\tau$ are the same as its right cosets. And because all transposition of $\Sr$ will within the same coset of $\mathcal{A}_3$, we also have 
 \begin{equation}
     \label{eq:ps_tauone_tautwo}
     \left( \id - \Wa{\tau} \right) \ \ps \ \left( \id + \Wa{\tau'} \right) \ = \
     \left( \id - \Wa{\tau} \right) \ \ps \ \left( \id + \Wa{\tau'} \right) \ = \ 0,
 \end{equation}
 for any transpositions $\tau, \tau' \in \Sr$.  \newline 
\noindent  The following completeness relation will also be significant.
\begin{equation}
    \label{eq:S3_id_decomposition}
    \id  \ = \ \ps + \pw + \pws, 
\end{equation}
and also
\begin{align}
    \label{eq:p_orthogonality}
     \ps  \  \pw = \ps  \  \pws \ = \pw \ \pws = 0,
\end{align}
 Equations \eqref{eq:pcs_ps} and \eqref{eq:p_orthogonality} also tell us that \begin{equation}
 \label{eq:p_orthogonality2}
 \pcs \ \pw = \pcs \ \pws = \pas \ \pw = \pas \ \pws \  = 0.
 \end{equation}
 
 For convenience, define $\pj{k}$ as follows.
\begin{equation}
    \label{eq:def_pj}
    \pj{k} \ : = \ \dfrac{1}{3} \left( \ \id  \ + \ \omega^{k} \Wa{[(123)]} \ + \ \omega^{2k} \Wa{[(132)]} \ \right).
\end{equation}
We note that $\ps$, $\pw$, and $\pws$ equal $\pj{0}, \pj{1}, \pj{2}$ respectively. We use the above notation in Lemmas \ref{lem:ETTp_eigenvalues} and \ref{vc_lem:ETTp_eigenvalues}.

\noindent This implies that we study the eigenspaces of $\ET$ after the action of $\EP$.

\noindent Define the following subspaces in $\supp{\EP}$.
\begin{description}
\item[\textlabel{identity sector}{vsp:unity}] ~\\
\begin{equation}
    \label{eq:vzero_vector}
    \unity \longrightarrow \ket{0,0,0},
\end{equation}
thus
\begin{equation}
    \label{eq:vec_vzero}
    \vzero \ = \ \mathrm{span} \ \left\{ \ \ket{0,0,0} \right\}.
\end{equation}
\item[\textlabel{first diagonal sector}{vsp:aan}] ~\\ 
\begin{equation}
    \label{eq:vec_aan}
    E(a)\otimes E(a) \otimes \unity_N \longrightarrow \ket{a,a,0},
\end{equation}
thus 
\begin{equation}
    \label{eq:vec_vdthree}
    \vdthree \ = \ \mathrm{span} \ \left\{ \ket{a,a,0} \ | \ a\in \Fs \right\}.
\end{equation}
\item[\textlabel{second diagonal sector}{vsp:ana}] ~ \\
\begin{equation}
\label{eq:vec_ana}
E(a) \otimes \unity_N \otimes E(a) \longrightarrow \ket{a,0,a},
\end{equation}
thus,
\begin{equation}
    \label{eq:vec_vdtwo}
    \vdtwo \ = \ \mathrm{span} \ \left\{ \ \ket{a,0,a} \ | \ a \in \Fs \right\}.
\end{equation}
\item[\textlabel{third diagonal sector}{vsp:naa}] ~ \\
\begin{equation}
\label{eq:vec_naa}
\unity_N \otimes E(a) \otimes E(a) \longrightarrow \ket{0,a,a},
\end{equation}
thus,
\begin{equation}
    \label{eq:vec_vdone}
    \vdone \ = \ \mathrm{span} \ \left\{ \ \ket{0,a,a} \ | \ a \in \Fs \right\}.
\end{equation}
\item[\textlabel{anti-commuting sector}{vsp:C}]~ \\
\begin{equation}
\label{eq:vec_nc}
E(a)\otimes E(b) \otimes E(a+b) \ \longrightarrow \  \ket{a,b,a+b}, \ \mathrm{where} \ \s{a}{b}=1.
\end{equation}
Thus,
\begin{equation}
    \label{eq:vec_vnc}
    \vnc \ = \ \mathrm{span} \ \left\{ \ \ket{a,b,a+b} \ | \ a,b,a+b \in \Fs, \ \mathrm{s.t.} \ \s{a}{b}=1 \right\}.
\end{equation} 
\item[\textlabel{commuting sector}{vsp:D}]  ~ \\
\begin{equation}
\label{eq:vec_c}
E(a)\otimes E(b) \otimes E(a+b) \ \longrightarrow \  \ket{a,b,a+b}, \ \mathrm{where} \ \s{a}{b}=0, \  \mathrm{for} \ a,b,a+b \in \Fs.
\end{equation}
Thus,
\begin{equation}
    \label{eq:vec_vc}
    \vc \ = \ \mathrm{span} \ \left\{ \ \ket{a,b,a+b} \ | \ a,b, a+b \in \Fs, \ \mathrm{s.t.} \ \s{a}{b}=0. \right\}
\end{equation} 
\end{description}
Since the above vectors are stabilized by $\EP$, we can directly apply $\ET$ to them. Let $\ket{a,b,c}\in \range{\EP}$. Then $\ket{a,b,c}$ is one of the following: $\ket{0,0,0}$, $\ket{a,a,0}$, $\ket{a,0,a}$, $\ket{0,a,a}$, $\ket{a,b,a+b}$ for $\s{a}{b}=1$ or $\ket{a,b,a+b}$ for $\s{a}{b}$ such that $a,b,a+b\neq 0$. For any $h \in \Fm$, we get then that 
\begin{equation}
\label{eq:Uh_action_vectorized}
\Ua{h} \ \ket{a,b,c} \ = \ \kt{a+\s{a}{h}h}{b+\s{b}{h}h}{c+\s{c}{h}h}
\end{equation}
Then 
\begin{align}
    \label{eq:ET_action}
    \ET \ \ket{a,b,c}  \ =  & \ \dfrac{1}{N^2} \ \sum_{h \in \Fm} \ \Ua{h} \ket{a, \; \; \; b, \; \; \;c} \notag \\ 
    \ = & \ \dfrac{1}{N^2} \ \sum_{h \in \Fm} \kt{a+\s{a}{h}h}{b+\s{b}{h}h}{c+\s{c}{h}h} 
\end{align}

\subsubsection{Action of \texorpdfstring{$\ET$}{randomizing transvections} on \texorpdfstring{$\vzero$}{identity}}
\label{subsubsec:ET_000}
This subspace is of dimension $1$. Since it is invariant under the action of the Clifford group one straightforwardly gets the following result.
\begin{equation}
    \label{eq:ET_000}
    \ET \ \ket{0,0,0} \ = \ \ket{0,0,0}.
\end{equation}
Equation \eqref{eq:ET_000} tells us that $\ET$ has eigenvalue $+1$, and that $\vzero$ is a subspace within this $+1$-eigenspace. This is already anticipated due to the fact that $\ET$ is a unital channel. 
\subsubsection{Action of \texorpdfstring{$\ET$}{randomizing transvections} on \texorpdfstring{$\vdthree$}{identity}}
\label{subsubsec:ET_aan}
When $\ket{a,b,c} = \ket{a,a,0}$ for $a \in \Fs$, the space is of dimension $N^2-1$.  The transformation in equation \eqref{eq:ET_action} takes the form
\begin{align}
\label{eq:ET_action_aan_1}
\ET \ \ket{a,a,0} \  = & \ \ \   \dfrac{1}{N^2} \ \sum_{h: \s{a}{h}=0} \ \ket{a,a,0} \notag \\ 
 \ &  \ + \ \dfrac{1}{N^2} \ \sum_{h: \s{a}{h}=1} \ \kt{a+\s{a}{h}}{a+\s{a}{h}}{0}.
\end{align}
The term $ \dfrac{1}{N^2} \ \sum_{h: \s{a}{h}=0} \ \ket{a,a,0}$ equals $\frac{1}{2} \ \ket{a,a,0}$ using Lemma \ref{lem:1}. For the second term note that 
\begin{equation}
    \label{eq:a_aph}
    \left\{h \in \Fm, | \s{a}{h}=1\right\} =  \left\{ a+h \in \Fm   \mathrm{s.t.}  \  \s{a+h}{h}=1, \ a: \ \mathrm{fixed}\right\}.
\end{equation}
Thus we can re-label the dummy variable $h$ as $a+h$, and by doing so we obtain the expression
\begin{equation}
    \label{eq:ET_action_aan_2}
    \left( \ET \ - \dfrac{1}{2} \id \right)  \ \ket{a,a,0} \  = \ \dfrac{1}{N^2} \ \sum_{h: \s{a}{h}=1} \ \ket{h,h,0}.
\end{equation}

\noindent Note that the third copy of the $m$ qubit system doesn't play an active role in equation \eqref{eq:ET_action_aan_2}. On this space $\ET$ can be written as a linear combination of projectors on some of the irreducible representations occuring in the two-copy adjoint representation of the the Clifford group, we will need Lemma 4 in \cite{Helsen2016} \footnote{Helsen et al \cite{Helsen2016} refer to this subspace as the diagonal sector}. 

\begin{lem}
\label{eq:lem_aan}
$\vdthree$ supports three irreducible representations of the Clifford group. These are listed below. 
\begin{itemize}
    \item The one-dimensional space which is spanned by the following vector. 
    \begin{equation}
        \label{eq:diagonal_trivial}
        v_\mathrm{tr} \ = \ \sum_{a=\Fs} \ \ket{a,a,0}.
    \end{equation}
    This one dimensional space supports the trivial irreducible representation of the Clifford group. 
    \begin{lem}
    \label{lem:diagaonl_invariant}
 The space spanned by $v_{\mathrm{tr}}$ in equation \eqref{eq:diagonal_trivial} lies in the $+1$ eigenspace of $\ET$.
    \end{lem}
    \item The $\frac{N(N+1)}{2}-1$ dimensional space spanned by 
    \begin{equation}
        \label{eq:diagonal_minus}
        v_- \ = \ \sum_{a \in \Fs} \ \lambda_a \ \ket{a,a,0}, 
    \end{equation}
    where $\lambda_a$ satisfy the following conditions.
    \begin{subequations}
    \begin{equation}
        \label{eq:condition_lambda_minus_a_1}
        \sum_{a \in \Fs} \ \lambda_a  \ = \ 0,
        \end{equation}
        which guarantees orthogonality to $v_\mathrm{tr}$ in equation 
        \eqref{eq:diagonal_trivial}. The second condition is true for any $\lambda_a$ in equation \eqref{eq:diagonal_minus}.
        \begin{equation}
            \label{eq:condition_lambda_minus_a_2}
            \lambda_a \ = \  -  \ \dfrac{2}{N} \ \sum_{b:\s{a}{b}=1} \ \lambda_b. 
        \end{equation}
        \end{subequations}
        Vectors of the form $v_-$ in equation \eqref{eq:diagonal_minus} support an $\frac{N(N+1)}{2}-1$ dimensional irreducible adjoint representation.
 \item The $\frac{N(N-1)}{2}-1$ dimensional space spanned by 
    \begin{equation}
        \label{eq:diagonal_plus}
        v_+ \ = \ \sum_{a \in \Fs} \ \lambda_a \ \ket{a,a,0},
    \end{equation}
    where $\lambda_a$ satisfy the following conditions.
    \begin{subequations}
    \begin{equation}
        \label{eq:condition_lambda_a_plus_1}
        \sum_{a \in \Fs} \ \lambda_a  \ = \ 0,
        \end{equation}
        which guarantees orthogonality to $v_\mathrm{tr}$ in equation 
        \eqref{eq:diagonal_trivial}. The second condition is true for any $\lambda_a$ in equation \eqref{eq:diagonal_plus}.
        \begin{equation}
            \label{eq:condition_lambda_a_plus_2}
            \lambda_a \ = \ \ \  \dfrac{2}{N} \ \sum_{b:\s{a}{b}=1} \ \lambda_b.. 
        \end{equation}
        \end{subequations}
        Vectors of the form $v_-$ in equation \eqref{eq:diagonal_minus} support an $\frac{N(N-1)}{2}-1$ dimensional irreducible adjoint representation.
        Equations \eqref{eq:condition_lambda_minus_a_2} and \eqref{eq:condition_lambda_a_plus_2} ensure that $v_+$ and $v_-$ from equations \eqref{eq:diagonal_plus} and \eqref{eq:diagonal_minus} are orthogonal.
\end{itemize}
\end{lem}
\begin{proof}
The proofs for each of the above is given in Lemma 4 in \cite{Helsen2016}. 
\end{proof}

\begin{lem}
    \label{lem:ET_aan_eigenvalues} ~ \newline
    \begin{itemize}
        \item The space spanned by $v_{\mathrm{tr}}$ in equation \eqref{eq:diagonal_trivial} lies in the $+1$ eigenspace of $\ET$.
        \item The space spanned by $v_-$ lies within an eigenspace of $\ET$ with the eigenvalue $\frac{1}{2}\left(1-\frac{1}{N}\right)$.
        \item The space spanned by $v_+$ lies within an eigenspace of $\ET$ with the eigenvalue $\frac{1}{2}\left(1+\frac{1}{N}\right)$.
    \end{itemize}
 \end{lem}
 \begin{proof}
From Lemma \ref{corr:ET_EP_H_properties}, we establish that $\ET$ is a linear combination of the projectors of the irreducible spaces of the Clifford group action. It remains to compute the corresponding eigenvalues.  Instead of computing $\ET \ v_{\mathrm{tr}}$ we compute $\left( \ET - \frac{1}{2} \id \right) \ v_{\mathrm{tr}}$ using equation \eqref{eq:ET_action_aan_2}.
 \begin{itemize}
     \item \begin{align}
         \label{eq:eigenvalues_ET_aan_trivial1}
         \left( \ET-\frac{1}{2}\id \right) \ v_{\mathrm{tr}} \ & = \  \dfrac{1}{N^2} \ \sum_{a \in \Fs}  \ \left( \left( \ET - \frac{1}{2} \ \id \right)  \ \ket{a,a,0} \right)  \\
         \label{eq:eigenvalues_ET_aan_trivial2}
         & = \ \dfrac{1}{N^2} \  \sum_{\substack{a, h \in \Fs \\ \ \mathrm{s.t.} \  \s{a}{h}= 1}} \ \ket{h,h,0}   \\ 
         \label{eq:eigenvalues_ET_aan_trivial3}
         & = \ \dfrac{1}{N^2} \ \sum_{h \in \Fs} \ \left( \sum_{a:\s{a}{h}=1} \ 1  \ \right) \ \ket{h,h,0}  \\
         \label{eq:eigenvalues_ET_aan_trivial4}
         & = \dfrac{1}{2} \ v_{\mathrm{tr}},
     \end{align}
     where \begin{enumerate}
         \item[] we applied equation \eqref{eq:ET_action_aan_2} in \eqref{eq:eigenvalues_ET_aan_trivial1} to get to \eqref{eq:eigenvalues_ET_aan_trivial2} and
         \item[] we employed Lemma \ref{lem:1} to compute $\left| \ \left\{ a \in \Fs \ : \s{h}{a}=1 \right\} \right|$, which is equal to $\frac{N^2}{2}$, and is independent of $h$ for $h \in \Fs$.
     \end{enumerate}
     \item We prove for $v_-$ and $v_+$ together. 
     \begin{align}
         \label{eq:ET_action_vpm0}
         \left( \  \ET - \frac{1}{2} \ \mathrm{id} \right)  \ v_{\pm} \ & = \ \sum_{a \in \Fs} \ \lambda_a \left( \  \ET - \frac{1}{2} \ \mathrm{id} \right) \ \ket{a,a,0}  \\ 
         \label{eq:ET_action_vpm1}
         & = \ \sum_{a \in \Fs} \ \lambda_a \ \left(  \  \dfrac{1}{N^2} \  \sum_{h:\s{a}{h}=1} \  \ket{h,h,0} \right) \\ 
         \label{eq:ET_action_vpm2}
         & = \dfrac{1}{N^2} \ \sum_{h \in \Fs} \ \left( \ \sum_{a:\s{a}{h}=1} \  \lambda_a \right) \  \ket{h,h,0}  \\ 
         \label{eq:ET_action_vpm3}
         & = \ \dfrac{1}{N^2} \ \sum_{h \in \Fs}  \ \left( \ \pm \ \frac{N}{2} \ \lambda_h  \ \ket{h,h,0} \right)  \\ 
         & = \ \pm \dfrac{1}{2N} \ v_{\pm},
     \end{align}
     where we used equation \eqref{eq:condition_lambda_minus_a_2} or \eqref{eq:condition_lambda_a_plus_2} in equation \eqref{eq:ET_action_vpm2} to get to equation \eqref{eq:ET_action_vpm3}.
 \end{itemize}
 \end{proof}
\subsubsection{Action of \texorpdfstring{$\ET$}{randomizing transvections} on \texorpdfstring{$\vdtwo$}{identity} and \texorpdfstring{$\vdone$}{identity}}
\label{subsubsec:ET_ana_naa}
\noindent We employ the Schur's lemma (see equation \eqref{eq:schur_lemma}) for the eigendecomposition of $\ET$ over the $\vdtwo$ and $\vdone$.

\noindent Let $v$ be any of the eigenvectors $v_\mathrm{tr}$, $v_-$ or $v_+$ of $\ET$, which were obtained in Subsubsection \ref{subsubsec:ET_aan}, and $\lambda$ the corresponding eigenvalue.   \begin{equation} \label{eq:ET_schur_lemma_ana_naa_1} \ET  \ \left( \ \Wa{\pi} \ v \right)  \   = \ \Wa{{\pi}} \ \left(  \ \ET \   v \ \right) \ =   \ \lambda \ \left( \Wa{{\pi}}  \    v \right). \end{equation} Note that  \begin{equation}  \label{eq:S3_aan_ana} \ket{a,0,a}  = \Wa{[(23)]} \ \ket{a,a,0}, \; \; \; \;  \; \; \; \; \;     \ket{0,a,a}  = \Wa{[(13)]} \ \ket{a,a,0}. \end{equation} \noindent Thus $\Wa{[(23)]} v_{\mathrm{tr}}$ and $\Wa{[(13)]} v_{\mathrm{tr}}$ lie in the $+1$ eigenspace of $\ET$, $\Wa{[(23)]} v_{-}$ and $\Wa{[(13)]} v_{-}$ lie in the $\frac{1}{2}\left( 1 - \frac{1}{N} \right)$ eigenspace of $\ET$, and $\Wa{[(23)]} v_{+}$ and $\Wa{[(13)]} v_{+}$ lie in the $\frac{1}{2}\left( 1 + \frac{1}{N} \right)$ eigenspace of $\ET$.

\subsubsection{Action of \texorpdfstring{$\ET$}{random transvections} on \texorpdfstring{$\vnc$}{Non-commuting sector} } 
\label{subsec:eigenvalues_ET_vnc}
\begin{lem}
\label{lem:eig_ET_vnc}
On $\vnc$, $\ET$'s highest eigenvalue is $+1$ with a one-dimensional eigenspace. It is given by
\begin{equation}
 \label{eq:def_wtr}
     \ \wtr \ := \ \mathrm{span} \  \left\{ \  \ket{w_\mathrm{tr}} \ :=  \sum_{\substack{a,b \\ 
    \s{a}{b}=1}} \ \ket{a,b, a+b}  \right\}. \end{equation}  
The second highest eigenvalue of $\ET$ on $\vnc$ is upper bounded by $\frac{1}{2} \left( 1 + \frac{7}{4N} \right)$.
\end{lem}
\begin{proof}
Structure of proof: we first prove that $\ET$ can be written as a sum of the following three linear operators.
\begin{equation}
    \label{eq:ET_decomposition}
    \ET \ = \ \frac{1}{4} \ \id + \ \ETT_{+} + \ETT_-,
\end{equation}
with the properties
\begin{enumerate}
\item $\ETT_+$ is a self-adjoint linear operator.  $\wtr$ is an eigenspace of $\ETT_+$ with eigenvalue $\frac{3}{4}$. The non-zero eigenvalues of $\ETT_+$ are $\dfrac{1}{4} \left( 1 \pm \frac{2}{N} \right)$, and $\dfrac{1}{4} \left( 1 \pm \frac{1}{N} \right)$.
\item The largest singular value of $\ETT_-$ is lower or equal to $\frac{3}{2N}$, i.e., 
\begin{equation}
    \label{eq:op_bound_ETTm}
    \op{\ETT_- } \ \le \dfrac{3}{2N}.
\end{equation}
$\supp{\ETT_-}$ lies in the orthogonal complement of $\wtr$.
\end{enumerate}
We construct the target decomposition of equation \eqref{eq:ET_decomposition} in Lemma \ref{lem:ET_decomposition}. We obtain the eigenvalues of $\ETT_+$ in Lemma \ref{lem:ETTp_eigenvalues} along with their eigenspaces. In Lemma \ref{lem:ETTm_eigenvalues} we prove the inequality  \eqref{eq:op_bound_ETTm}.

\end{proof}

\begin{lem}
\label{lem:ET_decomposition}
In this lemma, we construct the decomposition given in equation \eqref{eq:ET_decomposition}.
\end{lem}

Using equation \eqref{eq:ET_action}, one sees that the action of $\ET$ on $\ket{a,b,a+b}$ when $a,b$ are such that $\s{a}{b}=1$ is given by
\begin{align}
    \label{eq:ET_action_non_commuting_sector}
    \ET \  \ket{a,b,a+b}  \  =  \ & \ \ \ \ \ \ \ \dfrac{1}{N^2} \ \sum_{\substack{h: \\ 
    \s{a}{h}=0 \\ \s{b}{h}=0 }} \ \ket{a,b,a+b}  \\ 
    \label{eq:T}
    & \  + \ \dfrac{1}{N^2} \  \ \sum_{\substack{h: \\
    \s{a}{h}=0 \\ \s{b}{h}=1 }} \ \kt{a}{b+h}{a+b+h}  \\
    \label{eq:S}
    &\  +  \ \dfrac{1}{N^2} \ \sum_{\substack{h: \\ 
    \s{a}{h}=1 \\ \s{b}{h}=0 }} \ \kt{a+h}{b}{a+b+h} \\ 
    \label{eq:R}
    & \  +  \ \dfrac{1}{N^2} \ \sum_{\substack{h: \\ 
    \s{a}{h}=1 \\ \s{b}{h}=1 }} \ \kt{a+h}{b+h}{a+b}.
\end{align}
\noindent In Section \ref{sec:appendix:ET_vnc}, we prove from equation \eqref{eq:ET_action_non_commuting_sector} (and the expressions \eqref{eq:T}, \eqref{eq:S} and \eqref{eq:R}) that the action of $\ET$, when restricted to $\vnc$ can be re-written in the following form 
\begin{equation} 
\label{eq:ET_vnc_restriction}
\ET \  = \ \dfrac{1}{4} \id \ + \ \tncone \ + \ \tnctwo \ + \ \tncthree,
\end{equation} where by $\id$ we mean the identity restricted to $\vnc$, and we define the linear operators $\tncone$, $\tnctwo$ and $\tncthree$ on $\vnc$ by providing their action on the basis vectors $\ket{a,b,a+b}$.
\begin{align}
        \label{eq:def_tncone}
        \tncone \ \ket{a,b,a+b} \ := \ \dfrac{1}{N^2} \ \sum_{\substack{h: \\
    \s{a}{h}=1 \\ \s{b}{h}=1 }} \ \ket{a,h,a+h},
\\
        \label{eq:def_tnctwo}
                \tnctwo \ \ket{b,a,a+b} \ := \ \dfrac{1}{N^2} \ \sum_{\substack{h: \\
    \s{a}{h}=1 \\ \s{b}{h}=1 }} \ \ket{h,a,a+h}, 
    \\
        \label{eq:def_tncthree}
                \tncthree \ \ket{a+b,b,a} \ := \ \dfrac{1}{N^2} \ \sum_{\substack{h: \\
    \s{a}{h}=1 \\ \s{b}{h}=1 }} \ \ket{a+h,h,a}. 
    \end{align}
    
We define $\ETT$ as 
\begin{equation}
    \label{eq:defn_ETT}
    \ETT \ := \ \ET - \dfrac{1}{4} \ \id. 
\end{equation}
Then from equation \eqref{eq:ET_vnc_restriction} we see that 
\begin{equation}
    \label{eq:ETT}
    \ETT \ = \ \tncone \ + \ \tnctwo \ + \ \tncthree.
\end{equation}
Define 
\begin{align}
    \label{eq:def_tnconem}
    \tnconem \ := \ \frac{1}{2}\left( \id - \Wa{[(23)]} \right) \ \tncone \ \frac{1}{2}\left( \id - \Wa{[(23)]} \right). \\
\label{eq:def_tnconep}
    \tnconep \ := \ \frac{1}{2}\left( \id + \Wa{[(23)]} \right) \ \tncone \ \frac{1}{2}\left( \id + \Wa{[(23)]} \right),  \\ 
       \label{eq:def_tnctwom}
     \tnctwom \ := \  \ \frac{1}{2}\left( \id - \Wa{[(13)]} \right) \ \tnctwo \ \frac{1}{2}\left( \id - \Wa{[(13)]} \right)  \\ 
\label{eq:def_tnctwop}
    \tnctwop \ := \  \ \frac{1}{2}\left( \id + \Wa{[(13)]} \right) \ \tnctwo \ \frac{1}{2}\left( \id + \Wa{[(13)]} \right) \\
    \label{eq:def_tncthreem}
     \tncthreem \  := \  \ \frac{1}{2}\left( \id - \Wa{[(12)]} \right) \ \tncthree \ \frac{1}{2}\left( \id - \Wa{[(12)]} \right) \\ 
\label{eq:def_tncthreep}
     \tncthreep \ := \  \ \frac{1}{2}\left( \id + \Wa{[(12)]} \right) \ \tncthree \ \frac{1}{2}\left( \id + \Wa{[(12)]} \right). 
\end{align} In Lemma \ref{lem:tnc_rela}, we prove that $\tncone$, $\tnctwo$ and $\tncthree$ commute with $\Wa{[(23)]}$, $\Wa{[(13)]}$ and $\Wa{[(12)]}$ respectively. Hence the following equations give an orthogonal decomposition of $\tncone$, $\tnctwo$ and $\tncthree$.
\begin{align}
    \label{eq:tncone_decomposition}
    \tncone \ = \ \tnconep \ + \ \tnconem. \\
    \label{eq:tnctwo_decomposition}
    \tnctwo \ = \ \tnctwop \ + \ \tnctwom. \\
    \label{eq:tncthree_decomposition}
    \tncthree \ = \ \tncthreep \ + \ \tncthreem. 
\end{align} 
Define
 \begin{align}
     \label{eq:ETTp}
     \ETT_+ \ := \ \tnconep \  + \  \tnctwop \  + \  \tncthreep, \\
     \label{eq:ETTm}
     \ETT_- \ := \ \tnconem \  + \  \tnctwom \  + \  \tncthreem,
 \end{align} from which we see that
 \begin{align}
\label{eq:ETT_decomposition}
 \ETT \ 
 = \  \ \ETT_+ \ + \ \ETT_-,
 \end{align}
which is the decomposition proposed in equation \eqref{eq:ET_decomposition}. 
\begin{lem}
    \label{lem:tnc_rela}
   $\tncone$, $\tnctwo$ and $\tncthree$ have the following properties.
    \begin{itemize}
    \item $\tncone$, $\tnctwo$ and $\tncthree$ commute with $\Wa{[(23)]}$, $\Wa{[(13)]}$, and $\Wa{[(12)]}$ respectively. \begin{align}
    \label{eq:tncone_axis_of_symmetry}
    \Wa{[(23)]} \ \tncone \ \Wa{[(23)]} \ = \ \tncone,\\
    \label{eq:tnctwo_axis_of_symmetry}
    \Wa{[(13)]} \ \tnctwo \ \Wa{[(13)]} \ = \ \tnctwo, \\
    \label{eq:tncthree_axis_of_symmetry}
    \Wa{[(12)]} \ \tncthree \ \Wa{[(12)]} \ = \ \tncthree. 
\end{align}
\item $\tncone$, $\tnctwo$ and $\tncthree$ are permutations of each other:  \begin{align}
\label{eq:tnc_rela_3_1}
\tncthree \ = & \    \Wa{[(13)]} \  \tncone \ \Wa{[(13)]} \\ 
\label{eq:tnc_rela_3_2}
= & \  \  \Wa{[(132)]} \ \tncone \Wa{[(123)]}  \\
= & \  \Wa{[(123)]} \ \tnctwo \Wa{[(132)]} \notag \\
= & \   \Wa{[(23)]} \  \tnctwo \ \Wa{[(23)]}, \notag  \\
\label{eq:tnc_rela_2}
\ \tnctwo \ = \ & \ \Wa{[(123)]} \ \tncone \Wa{[(132)]} \\ 
= & \  \Wa{[(12)]} \  \tncone \ \Wa{[(12)]}  \notag \\ 
=  & \  \Wa{[(132)]} \ \tncthree \Wa{[(123)]} \notag \\
= & \ \Wa{[(23)]} \  \tncthree \ \Wa{[(23)]} \notag \\ 
\label{eq:tnc_rela_1}
\tncone \ = & \ \Wa{[(123)]} \ \tncthree \Wa{[(132)]} \\
= & \   \Wa{[(13)]} \  \tncthree \ \Wa{[(13)]} \notag  \\
=  &  \  \Wa{[(132)]} \ \tnctwo \Wa{[(123)]} \notag \\
= \ &   \Wa{[(12)]} \  \tnctwo \ \Wa{[(12)]}, \notag 
\end{align}
\end{itemize}
\end{lem}
\begin{proof}
Structure of proof:first we prove equation \eqref{eq:tncone_axis_of_symmetry}. Then we prove the first (top-most) relation in equation \eqref{eq:tnc_rela_3_1}. Equation \eqref{eq:tnc_rela_3_2} is proved by first conjugating both sides of equation \eqref{eq:tncone_axis_of_symmetry} by $\Wa{[(13)]}$, and noting that $\Wa{[(123)} = \Wa{[(13)]} \Wa{[(23)]}$. The rest of the relations can be proved using identical steps.    
\noindent Proof of equation \eqref{eq:tncone_axis_of_symmetry}: in equation \eqref{eq:def_tncone} we see that 
\begin{align}
 \ \tncone \ \ket{a,a+b,b} & \ = \   \dfrac{1}{N^2} \ \, \, \sum_{\substack{k: \\
    \s{a}{k}=1 \\ \s{a+b}{k}=1 }} \ \ket{a,k,a+k} \notag \\ 
    & \ = \   \dfrac{1}{N^2} \ \, \, \sum_{\substack{k: \\
    \s{a}{k}=1 \\ \s{b}{k}=0 }} \ \ket{a,k,a+k}  \notag  \\ 
    \label{eq:tncone_2}
    & \ = \ \dfrac{1}{N^2} \ \sum_{\substack{h: \\ \s{a}{h}=1 \\ \s{b}{h}=1}} \ \ket{a,a+h,h},
\end{align}
where we go from the first line to the second using
\begin{equation}
    \label{eq:tncone_2_1_2}
    \left\{ k \in \Fs \  \big| \ 
    \s{a}{k}= \s{a+b}{k}=1 \right\} \ = \ \left\{ k \in \Fs \  \big| \ \s{a}{k}=1, \  \s{b}{k}=0 \right\}.
\end{equation} To go from the second line to the third, we change variables from $k$ to $h=a+k$, by noting that 
\begin{equation}
    \label{eq:tncone_2_2_3}
     \left\{ \ h \ \big| \  \s{a}{h}= \s{b}{h}=1 \right\} \ = \ a \ + \  \left\{ k \in \Fs \  \big| \ \s{a}{k}=1, \  \s{b}{k}=0 \right\}, \, \, \, \left(\mathrm{using} \ \s{a}{b}=1\right).
\end{equation} Noting that the LHS of equation \eqref{eq:tncone_2} can be re-written as  $ \tncone  \Wa{[(23)]} \  \ket{a,b,a+b}$, and the RHS in \eqref{eq:tncone_2} can be re-written as the RHS of \eqref{eq:tncone_3}, we get the following.
\begin{align}
    \label{eq:tncone_3}
 \ \tncone \ \  \Bigg( \ \Wa{[(23)]} \  \ket{a,b,a+b} \Bigg) & \ = \  \ \Wa{[(23)]}  \ \left(  \ \dfrac{1}{N^2} \ \, \, \sum_{\substack{h: \\
    \s{a}{h}=1 \\ \s{b}{h}=1 }} \ \ket{a,h,a+h} \right).
\end{align} Comparing equation \eqref{eq:tncone_3} with equation \eqref{eq:def_tncone} we see that the action of $\tncone$ on $\ket{a,b,a+b}$ commutes with the action of $\Wa{[(23)]}$ on $\ket{a,b,a+b}$. Since this is true for all basis vectors $\ket{a,b,a+b}$ in $\vnc$, we obtain equation \eqref{eq:tncone_axis_of_symmetry}. 

\noindent Proof of equation \eqref{eq:tnc_rela_3_1}: apply $\Wa{[(13)]}$ to both sides of equation \eqref{eq:def_tncone}
\begin{align}
    \label{eq:tncone_tncthree_step1}
  \  \Wa{[(13)]} \ \tncone \ \Wa{[(13)]} \  \ket{a+b ,b,a}  \ = &  \   \dfrac{1}{N^2} \ \, \, \sum_{\substack{h: \\
    \s{a}{h}=1 \\ \s{b}{h}=1 }} \ \ket{a+h,h,a} \notag \\
    = &  \ \tncthree \ \ket{a+b,b,a}, 
\end{align}
and equation \eqref{eq:tnc_rela_3_1} is proved by noting that equation \eqref{eq:tncone_tncthree_step1} holds for all basis vectors $\ket{a+b,b,a}$.
\end{proof}

\begin{lem}
\label{lem:ETTp_eigenvalues} 
$\ETT_+$ (defined in equation \eqref{eq:ETTp}) is a self-adjoint linear operator. It's largest eigenvalue is $\frac{3}{4}$, and $\wtr$ (defined in equation \eqref{eq:def_wtr}) is the corresponding eigenspace. All remaining eigenvalues are less than or equal to $\frac{1}{4} \left( 1 \pm \frac{2}{N} \right)$.
\end{lem}
\begin{proof}
Structure of proof: we use the following orthogonal decomposition of $\ETT_+$, which is proved in Lemma \ref{lem:ETT_ortho_decomposition}.
 \begin{align}
     \label{eq:ETTp_decomposition}
     {\ETT}_{+} \  = & \ \ETTps \  +  \ETTpw \ + \ \ETTpws,
 \end{align}
where 
 \begin{align}
   \label{eq:def_ETTps}
    & \ \ETTps \ := \ \ps \ \ETT_{+} \ \ps \\
     \label{eq:def_ETTpw}
    & \  \ETTpw \ := \ \pw \ \ETT_{+} \ \pw \\
     \label{eq:def_ETTpws}
     & \  \ETTpws \ := \ \pws \ \ETT_{+} \ \pws.
 \end{align}
 $\ps$, $\pw$ and $\pws$ were defined in equations \eqref{eq:ps}, \eqref{eq:pw} and \eqref{eq:pws} in Subsection \ref{subsec:Sr}. Hence the eigendecomposition of $\ETT_+$ is given by the eigendecomposition of $\ETTps$, $\ETTpw$ and $\ETTpws$. To obtain the eigendecomposition of $\ETTps$, $\ETTpw$, $\ETTpws$, note the following. \begin{align}
    \label{eq:ETCCpcs_ps_tconep_rela}
    & \ETTps \ = \ 3 \ \ps \tnconep \ps,  \\
    & \ETTpw \ = \ 3 \ \pw \tnconep \pw,  \\
    & \ETTpws \ = \ 3 \ \pws \tnconep \pws,
\end{align} 
which is proved by using the following in equations \eqref{eq:def_ETTps}, \eqref{eq:def_ETTpw}, and \eqref{eq:def_ETTpws}: $\tnctwo = \Wa{[(123)]} \tncone \Wa{[(132)]}$ and $\tncthree = \Wa{[(132)]} \tncone \Wa{[(123)]}$ (see Section \ref{sec:appendix_tncpm_rela} for the proofs). Thus one needs to obtain the eigendecomposition of $\ps \tnconep \ps$, $\pw \tnconep \pw$ and $\pws \tnconep \pws$. In Lemma \ref{lem:eig_ETTpcs} , we obtain the complete spectral decomposition of $\ps \tncone \ps$. It is noted that it is self-adjoint. Hence, $\ETTps$ is self-adjoint. Also it tells us that $\wtr$ is an eigenspace of $\ps \tncone \ps$, with eigenvalue $\frac{1}{4}$, and that the only other non-zero eigenvalues are $\frac{1}{12}\left( 1 \pm \frac{2}{N} \right)$, and that all eigenspaces are orthogonal to each other. From Lemma \ref{lem:eig_ETTpw_ETTpws} we obtain the complete spectral decomposition of $\pw \tncone \pw$ and $\pws \tncone \pws$, both of which are also seen to be self-adjoint. Also, it is seen that $\wtr$ lies in the kernel of both, and that the only remaining non-zero eigenvalues are $\frac{1}{12}\left( 1 \pm \frac{1}{N} \right)$. This proves the lemma. 
\end{proof}

 \begin{lem}
 \label{lem:ETT_ortho_decomposition}
 Proof that equation \eqref{eq:ETTp_decomposition} is an orthogonal decomposition of $\ETT_+$.
 \end{lem}
\begin{proof}
In equation \eqref{eq:S3_id_decomposition}, we multiply on the right (or left) by $\ETT_+$.
\begin{equation}
    \label{vc_eq:S3_id_decomposition_ETTp}
     \ETT_+ \ = \ \left( \ps + \pw + \pws \ \right) \ \ \ETT_+.
\end{equation} Next, note that $\ETT_+$ commutes with $\ps$, $\pw$ and $\pws$ since $\ETT_+$ commutes with all $\Wa{\pi}$'s as proved in Appendix \ref{sec:appendix_tncpm_rela}. Since $\ps$, $\pw$ and $\pws$ square to themselves, one gets
\begin{align}
    \label{eq:ETT_ps_1}
\ps \ETT_+ \  = \  \ps \  \ETT_+ \ \ps, \\
\label{eq:ETT_pw_1}
\pw \ETT_+ \  = \  \pw \  \ETT_+ \ \pw, \\
\label{eq:ETT_pws_1}
\pw \ETT_+ \  = \  \ps \  \ETT_+ \ \ps.
\end{align}
Noting that $\supp{\ps}$, $\supp{\pw}$, $\supp{\pws}$ are orthogonal, the orthogonal decomposition is proved. \end{proof}

\begin{lem}
\label{lem:ETTm_eigenvalues}
 $\ETT_{-}$ (defined in equation \eqref{eq:ETTm}) is lower or equal to $\frac{3}{2N}$. $\supp{\ETT_-}$ lies in the orthogonal complement of $\wtr$.
\end{lem}
\begin{proof}
Structure of proof: in Lemma \ref{lem:tnconem_upperbound} we prove that the singular values of $\tnconem$ are equal to $\frac{1}{2N}$. Since $\tnctwom=\Wa{[(123)]}\tnconem\Wa{[(132)]}$ and $\tncthreem=\Wa{[(132)]}\tnconem\Wa{[(123)]}$ (proved in Section \ref{sec:appendix_tncpm_rela} of the Appendix), $\tnctwom$ and $\tncthreem$ have the same singular values as $\tnconem$. Thus
\begin{align}
    \label{eq:ETTm_tnconem_rela}
    \op{ \ {\ETT}_- } \ = & \  \op{\tnconem + \tnctwom + \tncthreem}    \notag \\  \le \  & \   \op{\tnconem} + \op{\tnctwom} + \op{\tncthreem} \notag \\ & \ \le \ \dfrac{3}{2N}.
\end{align}
To prove that $\supp{\ETT_-}$ lies in the orthogonal complement of $\wtr$, note that $\ETT_-$ is spanned by the union of $\supp{\tnconem}$, $\supp{\tnctwom}$ and $\supp{\tncthreem}$. Note that $\supp{\tnconem}\subseteq \supp{\left(\id - \Wa{[(23)]} \right)}$,
$\supp{\tnctwom}\subseteq \supp{\left(\id - \Wa{[(13)]} \right)}$ and $\supp{\tncthreem}\subseteq \supp{\left(\id - \Wa{[(12)]} \right)}$. In the proof of Lemma \ref{lem:eig_ETTpcs}, we establish that $\wtr \in \subset \supp{\ps}$. Also $\wtr \subset  \supp{\left(\id + \Wa{{(23)}}\right)}$ using the form of $\ket{w_{\mathrm{tr}}}$ in equation \eqref{eq:wtr2}. Thus, in fact $\pcs \ \ket{w_\mathrm{tr}} = \ket{w_\mathrm{tr}}$, using the form of $\pcs$ in equation \eqref{eq:pcs_ps}. From equation \eqref{eq:pcs_ps}, we also see that $\supp{\left(\id - \Wa{\tau} \right)}$ is orthogonal to $\supp{\pcs}$ for all transpositions $\tau \in \Sr$.
\end{proof}

\begin{lem}
\label{lem:tnconem_upperbound}
The largest singular value of $\tnconem$ is $\dfrac{1}{2N}$.
\end{lem}
\begin{proof}
For $\s{a}{b}=1$, define
\begin{equation}
    \label{eq:g}
    \g{a}{b} \ := \ \dfrac{\sqrt{2}}{N} \ \sum_{\substack{h:\\ \s{a}{h}=1}} \ \left( -1 \right)^{\s{b}{h}} \ \ket{a,h,a+h}.
\end{equation}
Firstly, note that $\g{a}{b} \in \supp{\tnconem}$, since  $\bk{a,h,a+h}{\Bar{a};{b}} =  \bk{a,a+h,h}{\Bar{a};{b}} $ for all $h \in \Fs$. Next, note that for any $a,b$ such that $\s{a}{b}=1$, \  $\g{a}{b} = - \g{a}{a+b}$. Hence at most $\frac{\left(N^2-1\right)\left(N^2\right)}{2}$ of the $\g{a}{b}$ can be linearly dependent. Next, for a given $a,b$ and $c,d$ such that $\s{a}{b}=\s{c}{d}=1$, we see that 
\begin{equation}
    \label{eq:inner_prod}
    \bk{\Bar{a};b}{\Bar{c};d} \ = \ \delta_{a,c} \ \left( \ \delta_{b,c} - \delta_{b+c,a} \right),
\end{equation}
such that one can obtain an orthonormal basis of $\frac{\left(N^2-1\right)\left(N^2\right)}{2}$ vectors from the $\g{a}{b}$, which then spans $\supp{\left( \id - \Wa{[(23)]}\right)}$. We show how to construct such a basis: note that the subset $\left\{ 0, a \right\}$ is an abelian subgroup of $\F^{2m}$, and has $\frac{\left(N^2-1\right)\left(N^2\right)}{4}$ cosets in $\F^{2m}$. Label the coset from $1$ to $\frac{\left(N^2-1\right)\left(N^2\right)}{4}$, and choose coset representatives $b_j$. Then the subset $\left\{ \g{a}{b_j} \right\}$ is an orthonormal basis for $\supp{\left(\id - \Wa{[(23)]}\right)}$. Next, one sees that 
\begin{equation}
    \label{eq:tnconem_SVD}
    \tnconem  \ \dfrac{1}{\sqrt{2}} \ \left( \ \ket{a,b_j,a+b_j} \ - \ket{a,a+b_j,b_j}  \right) \ = \ - \dfrac{1}{2N} \ \g{a}{b_j}. 
\end{equation}
 Thus $\tnconem$ takes us from one orthonormal basis for $\supp{\left(\id - \Wa{[(23)]}\right)}$ to another orthonormal basis for $\supp{\left(\id - \Wa{[(23)]}\right)}$, and equation \eqref{eq:tnconem_SVD} is essentially the singular-value decomposition of $\tnconem$, with the singular values $\frac{1}{2N}$.
\end{proof}

\begin{lem}
\label{lem:tnconep_spectral_decomposition}
The spectral decomposition of $\tnconep$ is 
\begin{equation}
\label{eq:tnconep_spectral_decomposition}
\tnconep  \ = \ \dfrac{1}{4} \  \sum_{a\in\Fs} \ \kb{\hat{a}}{\hat{a}},
\end{equation}
where for any $a,b \in \Fs$,
\begin{align}
    \label{eq:def_z}
&\q{a} \ := \ \dfrac{\sqrt{2}}{N} \ \sum_{\substack{h: \\ \s{a}{h}=1}} \  \ \ket{a,h,a+h}, \\ 
& \ \bk{\hat{a}}{\hat{b}} \ = \ \delta_{a,b}.
\end{align}
\end{lem}
\begin{proof}
From equation \eqref{eq:def_tnconep}, we note that $\supp{\tnconep} \subset \supp{\left(\id+\Wa{[(23)]}\right)}$. Define the following vectors in $\supp{\left(\id+\Wa{[(23)]}\right)}$ for any $a\in \Fs$ and $b \in \F^{2m}$, and such that $\s{a}{b}=0$.
\begin{align}
    \label{eq:def_w}
    \w{a}{b} \ := \ \dfrac{\sqrt{2}}{N} \ \sum_{\substack{h: \\ \s{a}{h}=1}} \ (-1)^{\s{b}{h}} \ \ket{a,h,a+h}.
\end{align}
Note that $\q{a}$ really is just $\w{a}{0}$, if we allowed the second register to take the value $0$ (which we temporarily allow throughout the body of this proof). This differentiation between both types of vectors will prove to be convenient. That $\w{a}{b} \in \supp{\left(\id+\Wa{[(23)]}\right)}$ is proved next:
\begin{align}
\label{eq:w_23_symmetry}
    \Wa{[(23)]} \ \w{a}{b} \ \propto  & \  \sum_{\substack{h: \\ \s{a}{h}=1}} \ (-1)^{\s{b}{h}} \ \ket{a,a+h,h}. \notag \\
    = &  \ \sum_{\substack{h: \\ \s{a}{a+h}=1}} \ (-1)^{\s{b}{a+h}} \ \ket{a,a+h,h} \notag \\ 
    = &  \ \sum_{\substack{h: \\ \s{a}{h}=1}} \ (-1)^{\s{b}{h}} \ \ket{a,h,a+h}.
\end{align}
The second line uses  $\s{a}{b}=0$, and to go from the second line to the third, define $h'=a+h$, and note that $\s{a}{h}=1 \Leftrightarrow \s{a}{h'}=1$. Since $h'$ is a dummy variable, the prime symbol is then removed. Invariance is established by noting that the RHS of equations \eqref{eq:w_23_symmetry} and \eqref{eq:def_w} are proportional to one another. Next, for $a,c \in \Fs$ and $b,d \in \F^{2m}$ such that $\s{a}{b}=\s{c}{d}=0$, we get
\begin{align}
    \label{eq:w_inner_product}
    & \ \bk{\hat{a};b}{\hat{c};{d}} \notag \\ 
  = & \ \dfrac{2}{N^2} \ \delta_{a,c} \ \sum_{\substack{h:\\ \s{a}{h}=1}} \ \left(-1\right)^{\s{h}{b+d}}   \notag \\
  = & \  \delta_{a,c} \ \left( \left( \delta_{b,d} - \delta_{b+d,a}\right) + \frac{2}{N^2}\ \left(1-\delta_{b,d} + \delta_{b+d,a}  \right) \left( \  \mathrm{S}_1 - \ \mathrm{S}_2  \  \right)  \right) \notag \\
 = & \ \delta_{a,c} \ \left( \delta_{b,d} \ - \delta_{b+d,a} \right), 
\end{align}
where \begin{equation} \label{eq:def_sum1_sum2} \mathrm{S}_1 \  := \   \sum_{\substack{h:\\ \s{a}{h}=1 \\ \s{b+d}{h}=0}}   \ 1, \ \ \ \ \mathrm{S}_2 \ := \ \sum_{\substack{h:\\ \s{a}{h}=1 \\ \s{b+d}{h}=1}} \ 1. \end{equation} One can use Lemma \ref{lem:1} to show that when $a$ and $b+d$ are linearly independent $\mathrm{S}_1 \ - \mathrm{S}_2 = 0$. 

\noindent Equation \eqref{eq:w_inner_product} informs us that 
\begin{equation}
\label{eq:wab_waa+b_rela}
\w{a}{b} \ = \ \w{a}{a+b},
\end{equation}
which is also seen by directly plugging in $a+b$ instead of $b$ in equation \eqref{eq:def_w}. Thus the vectors $\w{a}{b}$, upto the identification between $\w{a}{b}$ and $\w{a}{a+b}$, form an orthonormal basis for $\supp{\left( \id + \Wa{[(23)]} \right)}$. Next, we show that they're also the eigenvectors of $\tnconep$
Note that $\tnconep \ \w{a}{b} = \tncone \ \w{a}{b}$.
\begin{align}
    \label{eq:tnconep_eigenvalue_equation_1} 
  \ & \ \frac{N}{\sqrt{2}} \times  \ \tnconep \   \w{a}{b}  \notag \\
= & \   \left( \ \sum_{\substack{h:\\ \s{a}{h}=1}}  \ (-1)^{\s{b}{h}} \ \tnconep \ket{a,h,a+h} \ \right). \notag \\
= & \   \left( \ \sum_{\substack{h:\\ \s{a}{h}=1}}  \ (-1)^{\s{b}{h}} \ \left( \ \dfrac{1}{N^2} \ \sum_{\substack{k:\\ \s{a}{k}=1 \\ \s{h}{k}=1 } } \ \ket{a,k,a+k} \  \right)  \ \right). \notag \\
= & \   \dfrac{1}{N^2} \  \ \sum_{\substack{k:\\ \s{a}{k}=1}}  \left( \mathrm{S}_1 \ - \mathrm{S}_2 \right)  \  \ket{a,k,a+k},
\end{align}
where we define \begin{equation} \label{eq:def_S1_S2}  \ \mathrm{S}_1 \ := \  \sum_{\substack{h:  \\ \s{h}{a}=1 \\ \s{h}{k}=1 \\ \s{h}{b}=0 }}  \ 1, \; \; \; \; \mathrm{S}_2 \ := \ \sum_{\substack{h:\\ \s{h}{a}=1 \\  \s{h}{k}=1 \\ \s{h}{b}=1 }} \ 1.  \end{equation} First, let's note that $a,k$ are always linearly independent since $\s{a}{k}=1$. When $b=0$, then $\mathrm{S}_2=0$, and Lemma \ref{lem:1} tell us that $S_1=\frac{N^2}{4}$.When $b=a$, $\mathrm{S}_1=0$ and $S_2=\frac{N^2}{4}$. For all other values of $b$, Lemma \ref{lem:1} tells us that both sums equal and cancel out. When $b=0$, Lemma \ref{lem:1} tell us that the first sum if $\frac{N^2}{4}$, whereas the second sum is $0$, while for $b=a$ the reverse is true. Hence we get.  
\begin{align}
    \label{eq:tnconep_eigenvalue_equation_2}
    & \ \tnconep \  \w{a}{b} \ = \ \dfrac{1}{4} \ \left( \delta_{b,0} \ + \delta_{b,a} \right) \w{a}{b}.
\end{align}
The equation proves that $\q{a}$ are eigenvectors of $\tnconep$ with eigenvalue $\frac{1}{4}$, and that it's orthogonal complement (which is spanned by $\w{a}{b}$'s for $b\neq 0,a$) is the eigenspace corresponding to the $0$ eigenvalue. 
\end{proof}


\begin{lem}
\label{lem:eig_ETTpcs}
$\ps \ \tncone \ \ps$ has three eigenvalues, which are as follows: $\frac{3}{4}$, $\frac{1}{4} \left( 1 \pm \frac{2}{N} \right)$. All its eigenspaces are orthogonal. The eigenspace corresponding to $\frac{3}{4}$ is $\wtr$.
\end{lem}
\begin{proof}
We obtain the complete spectral decomposition of $\tncone$ in Lemma \ref{lem:tnconep_spectral_decomposition}. Note that the kernel of $\tncone$ is contained in the kernel of $\ps \tncone \ps$. Hence the support of $\ps \tncone \ps$ is spanned by $\ps \q{a}$, for all $a \in \Fs$ ($\q{a}$'s are defined in equation \eqref{eq:def_z}). Hence, using the spectral decomposition of $\tncone$, which is obtained in Lemma \ref{lem:tnconep_spectral_decomposition}, we get
\begin{equation}
\label{eq:ps_tnconep}
\ \ps \ \tnconep \ \ps \ =  \ \dfrac{1}{4} \  \sum_{a\in\Fs} \ \ps \kb{\hat{a}}{\hat{a}} \ps.
\end{equation} From equation \eqref{eq:ps_tnconep} one clearly sees that $\ps \tncone \ps$ is self-adjoint. Hence all its eigenspaces are mutually orthogonal. The inner product among the vectors $\ps \q{a}$ is detailed in Section \ref{sec:appendix_inner} in the Appendix. 
\begin{align}
\label{eq:inner_a_ps_b}
 & \ \bra{\hat{a}} \ps \q{b} \ = \ \frac{1}{3} \ \left( \delta_{a,b} \ + \ \dfrac{4}{N^2} \ \delta_{\s{a}{b},1} \right). 
 \end{align} Hence $\ps \q{a}$ aren't orthonormal.
 
 \noindent We decompose $\supp{\left(\id + \Wa{[(23)]} \right)}$ as follows.
\begin{align}
    \label{eq:wtr2}
    &  \ \wtr \ = \ \mathrm{span} \  \left\{ \  \ket{w_\mathrm{tr}} \ = \frac{N}{\sqrt{2}} \sum_{a \in \Fs} \ \q{a}  \right\},   \\ 
    \label{eq:wpp}
    & \ \wpp \ := \ \mathrm{span} \ \left\{ \ \ket{w_+} \ = \  \sum_{a \in \Fs}\  \lambda_a \ \q{a}, \ \sum_{a \in \Fs} \ \lambda_a \ = 0, \ \sum_{\substack{h: \\ \s{h}{a}}=1} \ \lambda_h \ = \ \frac{N}{2} \ \lambda_a \right\}, \\ 
    \label{eq:wm}
    & \ \wm \ := \ \mathrm{span} \ \left\{ \ \ket{w_-} \ = \ \sum_{a \in \Fs}\  \lambda_a \ \q{a}, \ \sum_{a \in \Fs} \ \lambda_a \ = 0, \ \sum_{\substack{h: \\ \s{h}{a}}=1} \ \lambda_h \ = \ - \frac{N}{2} \ \lambda_a \right\}.
\end{align}
We note that $\ket{w_\mathrm{tr}}$ is the same as was defined in equation \eqref{eq:def_wtr}. Note that the equations in the definition of $\wpp$ and $\wm$ are consistent and have solutions. We know this because the same equations are used to decompose the irreducible representations in the diagonal sector in Subsubsections \ref{subsubsec:ET_aan} and \ref{subsubsec:ET_ana_naa}, and they were also introduced in the diagonal sector in Helsen et al \cite{Helsen2016}. We also know that $\dim \wpp = \frac{N(N-1)}{2}-1$ and $\dim\wm= \frac{N(N+1)}{2}-1$. And that the three spaces $\wtr$, $\wpp$ and $\wm$ are orthogonal to each other, collectively spanning the entire $N^2-1$ dimensional support of $\id + \Wa{[(23)]}$. Using $\ket{w}_\mathrm{tr}$ defined in equation \eqref{eq:def_wtr}, 
one can directly verify that $\ps \ket{w_\mathrm{tr}} = \ket{w_\mathrm{tr}}$\footnote{$\wtr$ is the sum of all basis vectors $\ket{a,b,a+b}$ with $\s{a}{b}=1$, and hence has to invariant under $\ps$.}. Also we have that 
\begin{equation}
\label{eq:wtr_eigenvalue}
\ps \  \tnconep \ \ps  \ket{w_\mathrm{tr}} \ = \ {\small \text{ $\frac{N}{\sqrt{2}} \times \ $}} \ps \ \left(  \sum_{a \in \Fs} \ \tnconep  \ \q{a} \right) \ = \frac{1}{4} \ket{w_\mathrm{tr}},
\end{equation}
where we used Lemma \eqref{lem:tnconep_spectral_decomposition}. Next, for $\ps \ \ket{w_\pm} \in \mathcal{W}_\pm$
\begin{align}
    \label{eq:wpp_eigenvalue}
    & \ \ps \ \tnconep \ps \ \ \ps \ket{w_\pm} \ \notag \\ 
    = & \ \ps \ \times \  \frac{1}{4} \ \sum_{a,b\in\Fs} \ \lambda_b \ \bra{\hat{a}}\ps \q{b} \ \ \q{a} \notag \\ 
    = & \ \frac{1}{4} \ \sum_{a,b \in \Fs} \  \lambda_b \ \left( \frac{1}{3}\left(  \delta_{a,b} + \ \dfrac{4}{N^2} \ \delta_{\s{a}{b},1} \right)\right) \  \q{a}, \, \, \left( \mathrm{using} \ \mathrm{equation} \ \mathrm{\eqref{eq:inner_a_ps_b}}\right) \notag \\
    = & \ \frac{1}{12} \left(1 \pm \dfrac{2}{N} \right) \ \ps \  \ket{w_{\pm}},
\end{align}
Thus the remaining eigenvalues of $\ps \tncone \ps$ are $\frac{1}{12} \left( 1 \pm \frac{2}{N} \right)$, and the corresponding eigenspaces are given by $\ps \mathcal{W}_{\pm}$. That $\ket{w_{tr}}$ is orthogonal to $\ps \ket{w_{\pm}}$ seen from the fact that $\bra{w_\mathrm{tr}}\ps\ket{w_{\pm}} = \bk{w_\mathrm{tr}}{w_{\pm}}=0$.
\end{proof}
\begin{rem}
\label{rem:irrep_tr}
The eigenspace $\wtr$ corresponds to the uniform sum of all basis vector $\ket{a,b,a+b}$ (where $\s{a}{b}=1$). The same eigenspace is also identified in \cite{Tan20} (Section 5.3). 
\end{rem}
\begin{lem}
\label{lem:eig_ETTpw_ETTpws}
The eigenvalues of $\pw \tncone \pw$ and $\pws \tncone \pws $ are $\frac{1}{12}\left(1 \pm \frac{1}{N} \right)$ and $0$. 
\end{lem}
\begin{proof}
We prove this following the same steps as in Lemma \ref{lem:eig_ETTpcs}. We explicitly prove the lemma for $\ETTpw$, and the proof for $\ETTpws$ follows with identical steps. In fact, by taking the complex conjugate of all the steps in the proof body below, one obtains equations, which together constitute the proof of $\ETTpws$.
\noindent $\supp{\left( \id + \Wa{[(23)]} \right)}$  is decomposed into an orthogonal direct sum of $\wtr$, $\wpp$ and $\wm$ as in equation \eqref{eq:wtr2}, \eqref{eq:wpp} and \eqref{eq:wm}. In the body of the proof of Lemma \ref{lem:eig_ETTpcs}, we noted that $\wtr \subset \supp{\ps}$, and since $\ps$ and $\pw$ are orthogonal projectors, $\pw \tncone \pw \ \ket{w_\mathrm{tr}} = 0$. Next, let $\ket{w_\pm} \in \pw \ \mathcal{W}_{\pm}$, then we have 
\begin{align}
    \label{eq:tncone_pw_wp_eigenvalue}
    & \ \pw \ \tnconep \pw \ \ \pw \ \ket{w_\pm} \ \notag \\ 
    = & \ \pw \times  \frac{1}{4} \ \sum_{a,b\in\Fs} \ \lambda_b \ \bra{\hat{a}}\pw \q{b} \ \q{a} \notag \\ 
    = & \ \pw \times \ \frac{1}{4} \ \sum_{a,b \in \Fs} \  \lambda_b \ \left( \frac{1}{3}\left(  \delta_{a,b} - \ \delta_{\s{a}{b},1} \right)\right) \  \q{a}, \, \, \left( \mathrm{using} \ \mathrm{equation} \ \mathrm{\eqref{eq:inner_a_pw_or_pws_b}}\right) \notag \\
    = & \ \frac{1}{12} \left(1 \mp \dfrac{1}{N} \right) \ \pw \ \ket{w_{\pm}},
\end{align}
where $\bra{\hat{a}}\pw \q{b}$ is evaluated in Section \ref{sec:appendix_inner}.
\noindent Thus the remaining eigenvalues are $\frac{1}{12} \left( 1 \pm \frac{1}{N} \right)$.
\end{proof}

\subsubsection{Action of \texorpdfstring{$\ET$}{random transvections} on \texorpdfstring{$\vc$}{commuting sector} } 
\label{subsec:eigenvalues_ET_vc}
\begin{lem}
\label{vc_lem:eig_ET}
On $\vc$, $\ET$'s highest eigenvalue is $+1$ with a one-dimensional eigenspace. It is given by
\begin{equation}
 \label{vc_eq:def_wtr}
     \ \wtr \ := \ \mathrm{span} \  \left\{ \  \ket{w_\mathrm{tr}} \ :=  \sum_{\substack{a,b \\ 
    \s{a}{b}=0 \\ a,b,a+b \neq 0}} \ \ket{a,b, a+b}  \right\}. \end{equation}  
The second highest eigenvalue of $\ET$ on $\vc$ is upper bounded by $\frac{1}{2}  \left( 1  + \frac{4}{N} + \frac{2}{N(N-2)}\right)$.
\end{lem}
\begin{proof}
Structure of proof: we first prove that $\ET$ can be written as a sum of the following three linear operators.
\begin{equation}
    \label{vc_eq:ET_decomposition}
    \ET \ = \ \frac{1}{4} \ \id + \ \ETTd + \ETTCS,
\end{equation}
with the properties
\begin{enumerate}
\item $\ETTd$ is a self-adjoint linear operator. It's highest eigenvalue is $\frac{3}{4}$, and $\wtr$ is the corresponding eigenspace. The remaining non-zero eigenvalues of $\ETTd$ are $\frac{1}{4} \left( 1 + \frac{2}{N-2} \right)$, $\frac{1}{4} \left( 1 + \frac{2}{N+2} \right)$, $\frac{1}{4} \left( 1 - \frac{2}{N-2} \right)$ and $\frac{1}{4} \left( 1 - \frac{2}{N+2} \right)$.
\item $\ETTCS$ is also self-adjoint. It's largest eigenvalue is lower or equal to $\frac{3}{2N}$, i.e., 
\begin{equation}
    \label{vc_eq:op_bound_ETTm}
    \op{\ETTCS } \ \le \dfrac{3}{2N}.
\end{equation}
$\supp{\ETTCS}$ lies in the orthogonal complement of $\wtr$.
\end{enumerate}
We construct the target decomposition of equation \eqref{vc_eq:ET_decomposition} in Lemma \ref{vc_lem:ET_decomposition}. We obtain the eigenvalues of $\ETTd$ in Lemma \ref{vc_lem:ETTp_eigenvalues} along with their eigenspaces. In Lemma \ref{vc_lem:ET_decomposition} we prove the inequality  \eqref{vc_eq:op_bound_ETTm}.

\end{proof}

\begin{lem}
\label{vc_lem:ET_decomposition}
In this lemma, we construct the decomposition given in equation \eqref{vc_eq:ET_decomposition}.
\end{lem}
\begin{proof}
Using equation \eqref{eq:ET_action}, one sees that the action of $\ET$ on $\ket{a,b,a+b}$ when $a,b$ are such that $\s{a}{b}=0$ is 
\begin{align}
    \label{vc_eq:ET_action_commuting_sector}
    \ET \  \ket{a,b,a+b}  \  = \ & \ \ \ \ \ \ \ \dfrac{1}{N^2} \ \sum_{\substack{h: \\ 
    \s{a}{h}=0 \\ \s{b}{h}=0 }} \ \ket{a,b,a+b}  \\ 
    \label{vc_eq:T}
    & \  + \ \dfrac{1}{N^2} \  \ \sum_{\substack{h: \\
    \s{a}{h}=0 \\ \s{b}{h}=1 }} \ \kt{a}{b+h}{a+b+h}  \\
    \label{vc_eq:S}
    &\  +  \ \dfrac{1}{N^2} \ \sum_{\substack{h: \\ 
    \s{a}{h}=1 \\ \s{b}{h}=0 }} \ \kt{a+h}{b}{a+b+h} \\ 
    \label{vc_eq:R}
    & \  +  \ \dfrac{1}{N^2} \ \sum_{\substack{h: \\ 
    \s{a}{h}=1 \\ \s{b}{h}=1 }} \ \kt{a+h}{b+h}{a+b}.
\end{align}
\noindent In Section \ref{vc_sec:appendix:ET_vnc}, we prove from equation \eqref{vc_eq:ET_action_commuting_sector} (and the expressions \eqref{vc_eq:T}, \eqref{vc_eq:S} and \eqref{vc_eq:R}) that the action of $\ET$, when restricted to $\vc$ can be re-written in the following form 
\begin{equation} 
\label{vc_eq:ET_vnc_restriction}
\ET \  = \ \dfrac{1}{4} \id \ + \ \tcone \ + \ \tctwo \ + \ \tcthree,
\end{equation} where by $\id$ we mean the identity restricted to $\vc$, and we define the linear operators $\tcone$, $\tctwo$ and $\tcthree$ on $\vc$ by providing their action on the basis vectors $ \ \ket{a,b,a+b}$, for $\s{a}{b}=0$, and $a,b,a+b \neq 0$.
\begin{align}
        \label{vc_eq:def_tncone}
        \tcone \ \ket{a,b,a+b} \ := \ \dfrac{1}{N^2} \ \sum_{\substack{h: \\
    \s{a}{h}=0 \\ \s{b}{h}=1 }} \ \ket{a,h,a+h},
\\
        \label{vc_eq:def_tnctwo}
                \tctwo \ \ket{b,a,a+b} \ := \ \dfrac{1}{N^2} \ \sum_{\substack{h: \\
    \s{a}{h}=0 \\ \s{b}{h}=1 }} \ \ket{h,a,a+h}, 
    \\
        \label{vc_eq:def_tncthree}
                \tcthree \ \ket{a+b,b,a} \ := \ \dfrac{1}{N^2} \ \sum_{\substack{h: \\
    \s{a}{h}=0 \\ \s{b}{h}=1 }} \ \ket{a+h,h,a}. 
    \end{align}
    
We define $\ETT$ as 
\begin{equation}
    \label{vc_eq:defn_ETT}
    \ETT \ := \ \ET - \dfrac{1}{4} \ \id. 
\end{equation}
Then from equation \eqref{vc_eq:ET_vnc_restriction} we see that 
\begin{equation}
    \label{vc_eq:ETT}
    \ETT \ = \ \tcone \ + \ \tctwo \ + \ \tcthree.
\end{equation}

\noindent Below we give some important properties of $\tcone$, $\tctwo$ and $\tcthree$, and also some relations among them. 
\begin{enumerate}
    \item From equation \eqref{vc_eq:def_tncone} it is seen that 
    \begin{equation}
        \label{vc_eq:tncone_adjoint}
        \bra{a,b,a+b}  \ T  \ \ket{a,c,a+c} \ = \ \frac{2}{N^2} \ \delta_{\s{b}{c},1} \ = \ \bra{a,c,a+c}  \ T  \ \ket{a,b,a+b}.
    \end{equation}
    Thus, $\tcone$ is self-adjoint. Similarly, $\tctwo$ and $\tcthree$ are also self-adjoint, and all of them admit a spectral decomposition into orthogonal eigenspaces.
    \item In Lemma \ref{vc_lem:tnc_rela}, we prove that $\tcone$, $\tctwo$ and $\tcthree$ commute with $\Wa{[(23)]}$, $\Wa{[(13)]}$ and $\Wa{[(12)]}$, respectively. Thus $\tcone$ can be diagonalized along the $+1$ and $-1$ eigenspaces of $\Wa{[(23)]}$, and similarly for $\tctwo$ and $\tcthree$.
    \item Lemma \ref{vc_lem:tnc_rela} also proves that $\tctwo = \Wa{[(123)]} \tcone \Wa{[(132)]} = \Wa{[(12)]} \tcone \Wa{[(12)]}$, and $\tcthree = \Wa{[(132)]} \tcone \Wa{[(123)]} =  \Wa{[(13)]} \tcone \Wa{[(13)]}$. Hence, obtaining the spectral decomposition of $\tcone$ also gives us the spectral decomposition of $\tctwo$ and $\tcthree$.
    \item In Lemma \ref{vc_lem:tncone_spectral_decomposition} we obtain the spectral decomposition for $\tcone$. We see the following.
    \begin{enumerate}
        \item \begin{align}
        \label{vc_eq:ker_tncone}
       & \frac{1}{2}\left( \id - \Wa{[(23)}\right) \ \tcone \ = 0.  \\
        \label{vc_eq:supp_tncone}
    &    \frac{1}{2}\left( \id + \Wa{[(23)}\right) \ \tcone \ = \tcone.  
        \end{align}
        \item Define 
        \begin{equation}
            \label{vc_eq:def_vcone}
            \vcone \ := \ \supp{\left( \id + \Wa{[(23)]} \right)}.
        \end{equation} Then, in other words, the $\supp(\tcone) = \vcone$. Note that $\vcone$ is $\frac{\left(N^2-1\right)\left(N^2-4 \right)}{4}$ dimensional.
        \item $\vcone$ has an orthogonal decomposition along two subspaces: $\vdcone$ which is $N^2-1$ dimensional and its orthogonal complement, $\vdcone^\perp$, which is $\frac{\left(N^2-1 \right) \left( N^2-8 \right)}{4}$ dimensional, with the following properties: let $\pdcone$ and $\pdcone^\perp$ be the projectors on $\vdcone$ and $\vdcone^\perp$ respectively. Then the following is true. Then,
        \begin{align}
            \label{vc_eq:tncone_spectraldecomposition_1}
           \pdcone\ \tcone \ = \ & \frac{1}{4} \ \pdcone,  \\ 
            \label{vc_eq:tncone_spectraldecomposition_2} 
            \pdcone^\perp \ \left(\tcone\right)^2 \ = \ & \frac{1}{4N^2} \ \pdcone^\perp,
            \end{align}
            Hence the eigenvalues of $\tcone$ are $\frac{1}{4}$ and $\pm \frac{1}{2N}$. We will use the following orthogonal decomposition of $\tcone$.
            \begin{equation}
                \label{vc_eq:tncone_decomposition}
                \tcone \ = \ \tconep \ + \tconem,
            \end{equation}
            where 
            \begin{align}
                \label{vc_eq:def_tconep}
                & \tconep \ := \ \pdcone \ \tcone \\
                \label{vc_eq:def_tconem}
                & \tconem \ := \ \pdcone^\perp \ \tcone.
            \end{align}
    \end{enumerate}
    \item Since $\tctwo$ and $\tcthree$ are unitary transformations of $\tcone$, their eigendecompositions are identical to $\tcone$'s. Hence we have the following.
    \begin{align}
        \label{vc_eq:tctwo_decomposition}
        & \tctwo \ = \ \tctwop \ + \ \tctwom, \\
        \label{vc_eq:tcthree_decomposition}
        & \tcthree \ = \ \tcthreep \ + \ \tcthreem,
    \end{align}
    where, for $\pdctwo = \Wa{[(123)]} \pdcone \Wa{[(132)]} = \Wa{[(12)]} \pdcone \Wa{[(12)]}$ , and $\pdcthree = \Wa{[(132)]} \pdcone \Wa{[(123)]} = \Wa{[(13)]} \pdcone \Wa{[(13)]} $, we get \begin{align}
                \label{vc_eq:def_tctwop}
                & \tctwop \ := \ \pdctwo \ \tctwo \ = \ \Wa{[(123)]} \ \tconep \ \Wa{[(132)]} \ = \ \Wa{[(12)]} \tconep \Wa{[(12)]} \\ 
                \label{vc_eq:def_tctwom}
                & \tctwom \ := \ \pdctwo^\perp \ \tctwo \ = \ \Wa{[(123)]} \ \tconem \ \Wa{[(132)]}  \ = \ \Wa{[(12)]} \tconem \Wa{[(12)]} \\
                \label{vc_eq:def_tcthreep}
                & \tcthreep \ := \ \pdcthree \ \tcthree \  = \ \Wa{[(132)]} \ \tconep \ \Wa{[(123)]}  \ = \ \Wa{[(13)]} \tconep \Wa{[(13)]} ,\\
                \label{vc_eq:def_tcthreem}
                & \tcthreem \ := \ \pdcthree^\perp \ \tcthree  \  = \ \Wa{[(132)]} \ \tconem \ \Wa{[(123)]}  \ = \ \Wa{[(13)]} \tconem \Wa{[(13)]}.
            \end{align}
            \end{enumerate} 
Define
 \begin{align}
     \label{vc_eq:ETTp}
     \ETTd \ & \:= \ \tconep \  + \  \tctwop \  + \  \tcthreep, \\
     \label{vc_eq:ETTm}
     \ETTCS \ & \ := \ \tconem \  + \  \tctwom \  + \  \tcthreem, 
 \end{align} from which we see that
 \begin{align}
\label{vc_eq:ETT_decomposition}
 \ETT \ 
 = \  \ \ETTd \ + \ \ETTCS.
 \end{align}
 Since $\op{\tconem}$, $\op{\tctwom}$ and $\op{\tcthreem}$ $\le $ $\frac{1}{2N}$, from equation \eqref{vc_eq:ETTm} we have that 
 \begin{equation}
 \label{vc_eq:ETTm_op}
 \op{\ETTCS} \ \le \ \dfrac{3}{2N}.
 \end{equation}
 Finally, Lemma \ref{vc_lem:ETTp_eigenvalues} gives the full eigendecomposition of $\ETTd$, which is as was specified in Lemma \ref{vc_lem:ET_decomposition}. Hence the construction of the decomposition is as was promised in the proof of Lemma \ref{vc_lem:ET_decomposition}.
 \end{proof}
 
\begin{lem}
    \label{vc_lem:tnc_rela}
   $\tcone$, $\tctwo$ and $\tcthree$ have the following properties.
    \begin{itemize}
    \item $\tcone$, $\tctwo$ and $\tcthree$ commute with $\Wa{[(23)]}$, $\Wa{[(13)]}$, and $\Wa{[(12)]}$ respectively. \begin{align}
    \label{vc_eq:tncone_axis_of_symmetry}
    \Wa{[(23)]} \ \tcone \ \Wa{[(23)]} \ = \ \tcone,\\
    \label{vc_eq:tnctwo_axis_of_symmetry}
    \Wa{[(13)]} \ \tctwo \ \Wa{[(13)]} \ = \ \tctwo, \\
    \label{vc_eq:tncthree_axis_of_symmetry}
    \Wa{[(12)]} \ \tcthree \ \Wa{[(12)]} \ = \ \tcthree. 
\end{align}
\item $\tcone$, $\tctwo$ and $\tcthree$ are permutations of each other:  \begin{align}
\label{vc_eq:tnc_rela_3_1}
\tcthree \ = & \    \Wa{[(13)]} \  \tcone \ \Wa{[(13)]} \\ 
\label{vc_eq:tnc_rela_3_2}
= & \  \  \Wa{[(132)]} \ \tcone \Wa{[(123)]}  \\
= & \  \Wa{[(123)]} \ \tctwo \Wa{[(132)]} \notag \\
= & \   \Wa{[(23)]} \  \tctwo \ \Wa{[(23)]}, \notag  \\
\label{vc_eq:tnc_rela_2}
\ \tctwo \ = \ & \ \Wa{[(123)]} \ \tcone \Wa{[(132)]} \\ 
= & \  \Wa{[(12)]} \  \tcone \ \Wa{[(12)]}  \notag \\ 
=  & \  \Wa{[(132)]} \ \tcthree \Wa{[(123)]} \notag \\
= & \ \Wa{[(23)]} \  \tcthree \ \Wa{[(23)]} \notag \\ 
\label{vc_eq:tnc_rela_1}
\tcone \ = & \ \Wa{[(123)]} \ \tcthree \Wa{[(132)]} \\
= & \   \Wa{[(13)]} \  \tcthree \ \Wa{[(13)]} \notag  \\
=  &  \  \Wa{[(132)]} \ \tctwo \Wa{[(123)]} \notag \\
= \ &   \Wa{[(12)]} \  \tctwo \ \Wa{[(12)]}, \notag 
\end{align}
\end{itemize}
\end{lem}
\begin{proof}
Structure of proof:first we prove equation \eqref{vc_eq:tncone_axis_of_symmetry}. Then we prove the first (top-most) relation in equation \eqref{vc_eq:tnc_rela_3_1}. Equation \eqref{vc_eq:tnc_rela_3_2} is proved by conjugating both sides of equation \eqref{vc_eq:tncone_axis_of_symmetry} by $\Wa{[(13)]}$, and noting that $\Wa{[(123)} = \Wa{[(13)]} \Wa{[(23)]}$. The rest of the relations can be proved using identical steps.    
\noindent Proof of equation \eqref{vc_eq:tncone_axis_of_symmetry}: in equation \eqref{vc_eq:def_tncone} we see that 
\begin{align}
 \ \tcone \ \ket{a,a+b,b} & \ = \   \dfrac{1}{N^2} \ \, \, \sum_{\substack{k: \\
    \s{a}{k}=0 \\ \s{a+b}{k}=1 }} \ \ket{a,k,a+k} \notag \\ 
    & \ = \   \dfrac{1}{N^2} \ \, \, \sum_{\substack{k: \\
    \s{a}{k}=0 \\ \s{b}{k}=1 }} \ \ket{a,k,a+k}  \notag  \\ 
    \label{vc_eq:tncone_2}
    & \ = \ \dfrac{1}{N^2} \ \sum_{\substack{h: \\ \s{a}{h}=0 \\ \s{b}{h}=1}} \ \ket{a,a+h,h},
\end{align}
where we go from the first line to the second using
\begin{equation}
    \label{vc_eq:tncone_2_1_2}
    \left\{ k \in \Fs \  \big| \ 
    \s{a}{k}=0, \  \s{a+b}{k}=1 \right\} \ = \ \left\{ k \in \Fs \  \big| \ \s{a}{k}=0, \  \s{b}{k}=1 \right\}.
\end{equation} To go from the second line to the third, we change variables from $k$ to $h=a+k$, by noting that 
\begin{equation}
    \label{vc_eq:tncone_2_2_3}
     \left\{ \ h \ \big| \  \s{a}{h}=0, \  \s{b}{h}=1 \right\} \ = \ a \ + \  \left\{ k \in \Fs \  \big| \ \s{a}{k}=0, \  \s{b}{k}=1 \right\}, \, \, \, \left(\mathrm{using} \ \s{a}{b}=0\right).
\end{equation} Noting that the LHS of equation \eqref{vc_eq:tncone_2} can be re-written as  $ \tcone  \Wa{[(23)]} \  \ket{a,b,a+b}$, and the RHS in \eqref{vc_eq:tncone_2} can be re-written as the RHS of \eqref{vc_eq:tncone_3}, we get the following.
\begin{align}
    \label{vc_eq:tncone_3}
 \ \tcone \ \  \Bigg( \ \Wa{[(23)]} \  \ket{a,b,a+b} \Bigg) & \ = \  \ \Wa{[(23)]}  \ \left(  \ \dfrac{1}{N^2} \ \, \, \sum_{\substack{h: \\
    \s{a}{h}=0 \\ \s{b}{h}=1 }} \ \ket{a,h,a+h} \right).
\end{align} Comparing equation \eqref{vc_eq:tncone_3} with equation \eqref{vc_eq:def_tncone} we see that the action of $\tcone$ on $\ket{a,b,a+b}$ commutes with the action of $\Wa{[(23)]}$ on $\ket{a,b,a+b}$. Since this is true for all basis vectors $\ket{a,b,a+b}$ in $\vc$, we obtain equation \eqref{vc_eq:tncone_axis_of_symmetry}. 

\noindent Proof of equation \eqref{vc_eq:tnc_rela_3_1}: apply $\Wa{[(13)]}$ to both sides of equation \eqref{vc_eq:def_tncone}
\begin{align}
    \label{vc_eq:tncone_tncthree_step1}
  \  \Wa{[(13)]} \ \tcone \ \Wa{[(13)]} \  \ket{a+b ,b,a}  \ = &  \   \dfrac{1}{N^2} \ \, \, \sum_{\substack{h: \\
    \s{a}{h}=0 \\ \s{b}{h}=1 }} \ \ket{a+h,h,a} \notag \\
    = &  \ \tcthree \ \ket{a+b,b,a}, 
\end{align}
and equation \eqref{vc_eq:tnc_rela_3_1} is proved by noting that equation \eqref{vc_eq:tncone_tncthree_step1} holds for all basis vectors $\ket{a+b,b,a}$.
\end{proof}

\begin{lem}
\label{vc_lem:tncone_spectral_decomposition}
The spectral decomposition of $\tcone$ (defined in equation \eqref{vc_eq:def_tncone}) is 
\begin{equation}
\label{vc_eq:tncone_spectral_decomposition_1}
\tconep  \ = \ \dfrac{1}{4} \  \pdcone \  + \ \dfrac{1}{2N} \pqcone \ - \dfrac{1}{2N} \prcone, 
\end{equation}
where $\pdcone$, $\pqcone$ and $\prcone$ are orthogonal projectors defined on the following orthogonal spaces. 
\begin{align} 
    \label{vc_eq:def_vdcone}
& \vdcone \ := \ \mathrm{span} \ \left\{ \ \q{a} \ \forall \ a \in \Fs \ \big| \  \q{a} \ := \ \dfrac{\sqrt{2}}{ \sqrt{ N^2-4 } }  \ \sum_{\substack{h: \\ \s{a}{h}=0\\ h \neq a,0}} \  \ \ket{a,h,a+h} \right\}, \\ 
\label{vc_eq:def_vqcone}
& \vqcone \ := \ \mathrm{span} \ \left\{ \ \ket{a_+} \ = \sum_{\substack{b: \\ \s{a}{b}=0 \\ b \neq 0,a}} \lambda_{b} \ \ket{a,b,a+b}, \ \forall \  a\in \Fs \ \big| \ \lambda_b \text{ satisfy equations in \eqref{vc_eq:set_plus}.} \right\},\\ 
\label{vc_eq:def_vrcone}
&\vrcone \ := \ \mathrm{span} \ \left\{ \ \ket{a_-} \ = \sum_{\substack{b: \\ \s{a}{b}=0 \\ b \neq 0,a}} \lambda_{b} \ \ket{a,b,a+b}, \ \forall \  a\in \Fs \ \big| \ \lambda_b \text{ satisfy equations in \eqref{vc_eq:set_minus}.} \right\},
\end{align} 
where 
\begin{align}
    \label{vc_eq:set_plus}
    & \sum_{\substack{b: \\ \s{a}{b}=0 \\ b \neq 0,a}} \lambda_b = 0,   \ \ \ \ \ \sum_{\substack{h: \\ \s{h}{a}=0 \\ \s{h}{b}=1}} \lambda_h \ = \ \dfrac{1}{2N}  \lambda_b ,
\end{align}
and 
\begin{align}
    \label{vc_eq:set_minus}
    & \sum_{\substack{b: \\ \s{a}{b}=0 \\ b \neq 0,a}} \lambda_b = 0,  \ \ \ \ \ \sum_{\substack{h: \\ \s{h}{a}=0 \\ \s{h}{b}=1}} \lambda_h \ = \ -\dfrac{1}{2N}  \lambda_b.
\end{align}
$\dim \vdcone = N^2-1$, $\dim \vqcone = \frac{\left(N^2-1 \right) \left( N^2 - 2N -8 \right)}{8}$, and $\dim \vrcone = \frac{\left(N^2-1 \right) \left( N^2 + 2N -8 \right)}{8}$
\end{lem}
\begin{proof} 
From equation \eqref{vc_eq:tncone_adjoint} we know that $\tcone$ is self-adjoint and has a spectral decomposition with orthogonal eigenspaces.

\noindent From the definition of $\tcone$ in equation \eqref{vc_eq:def_tncone}, we see that 
\begin{equation}
    \label{vc_eq:tncone_symm1}
    \tcone \ \ket{a,b,a+b} \ = \tcone \ \ket{a,a+b,b},
\end{equation}
since if some $h$ is contained in the set $\left\{ k \in \F^{2m} \ \big| \ \s{a}{k}=0, \s{b}{k}=1 \right\} $, then so is $a+h$. This further implies 
\begin{align}
\label{vc_eq:tncone_symm2}
   & \tcone \ \ket{a,b,a+b} \ = \  \tcone \ \Wa{[(23)]} \ \ket{a,b,a+b} \notag \\  \ & \; \;\;\;\;\;\;\;\;\;\;\;\;\;\;\;\;\;\;\;\;\;\;\;  =  \ \Wa{[(23)]} \ \tcone \ \ket{a,b,a+b} \notag \\
 \Longrightarrow & \ \left( \id - \Wa{[(23)]} \right) \ \tcone \ \ket{a,b,a+b} \ = \  0.
\end{align}
Since equation \eqref{vc_eq:tncone_symm2} is true for all $\ket{a,b,a+b}$, we get 
\begin{equation}
    \label{vc_eq:tncone_symm3}
    \left( \id - \Wa{[(23)]} \right) \ \tcone \ = \ 0,
\end{equation}
which informs us that $\supp{ \tcone} \subseteq \vcone$.

\noindent Next, we prove that for $\q{a}$, as defined in equation \eqref{vc_eq:def_vdcone}, is an eigenvector of $\tcone$, with eigenvalue $\frac{1}{4}$.
\begin{align}
    \label{vc_eq:q_eigenvector}
    \tcone \  \ket{a} \ = \ & \ \tcone \  \frac{\sqrt{2}}{\sqrt{N^2-4}} \ \sum_{\substack{b: \\
    \s{a}{b}=0 \\ b\neq 0,a}} \ \ket{a,b,a+b} \notag \\ 
     = \ & \frac{\sqrt{2}}{N^2 \ \sqrt{N^2-4}} \ \sum_{\substack{h: \\ \s{a}{h}=0}} \  \mathrm{S}_{h} \    \ \ket{a,h,a+h} \notag \\ 
    = \ & \ \dfrac{1}{4} \ \q{a},
\end{align}
where in the second line $\mathrm{S}_h =  \left| \ \left\{ b \in \F^{2m} \ \big| \ \s{a}{b}=0, \s{h}{b}=1, \ b \neq 0,a \right\} \  \right|$. Note that the condition $\s{h}{b}=1$ ensures that $b \neq 0,a$. $\mathrm{S}_h =  \left| \ \left\{ b \in \F^{2m} \ \big| \ \s{a}{b}=0, \s{h}{b}=1  \right\} \  \right|$. Also, had $h=0,a$, then $\s{b}{h}\neq 1$, so $h \neq 0,a$. Since $a,h$ are line arly independent, Lemma \ref{lem:1} tells us that this size is $\dfrac{N^2}{4}$. This proves equation \eqref{vc_eq:def_vdcone}. 
    
\noindent Let $b,c \in \Fs \setminus \{a\}$ such that $\s{a}{b}=\s{a}{c}=0$ and $c,a+c \neq b$.

Then
\begin{align}
    \label{vc_eq:tnconesquared_action0}
    & \ {\tcone}^2 \ \left( \left( \ket{a,b,a+b} + \ket{a,a+b,b } \right) -  \left( \ket{a,c,a+c} + \ket{a,a+c,c} \right)  \right) \notag \\ 
    = & \   \left( \dfrac{2}{N^2} \sum_{\substack{h: \\ \s{a}{h}=0 \\ \s{b}{h}=1}} \ T \ \ket{a,h,a+h}  \right) -  \left( \dfrac{2}{N^2} \sum_{\substack{h: \\ \s{a}{h}=0 \\ \s{c}{h}=1}} \ T \ \ket{a,h,a+h}  \right)   \notag \\ 
    = & \ \dfrac{2}{N^4} \ \sum_{\substack{k:\\ \s{a}{k}=0}} \ \left( \mathrm{S}_{b,k}  \ - \mathrm{S}_{c,k} \ \right) \ \ket{a,k,a+k},
\end{align}
where 
\begin{align}
    \label{vc_eq:Sbk}
    \mathrm{S}_{b,k} = \left| \left\{ h \in \F^{2m} \ \big| \ \s{a}{h}=0, \ \s{b}{h}=\s{k}{h}=1 \ \right\} \right| \\
    \label{vc_eq:Sck}
    \mathrm{S}_{c,k} = \left| \left\{ h \in \F^{2m} \ \big| \ \s{a}{h}=0, \ \s{c}{h}=\s{k}{h}=1 \ \right\} \right| 
\end{align}
Note that if $k=0,a$ then $\mathrm{S}_{b,k}, \mathrm{S}_{c,k} = 0$ since $\s{h}{k}\neq 1$ in that case. Using Lemma \ref{lem:1}, we see that 
\begin{align}
    \label{vc_eq:S_bk}
    \mathrm{S}_{b,k} \ = \  \dfrac{N^2}{4} \  \left( \delta_{k,b} + \delta_{k,a+b} \right) - \dfrac{N^2}{8} \ \left( 1 - \left(\delta_{k,b} + \delta_{k,a+b} \right) \right) \\
    \label{vc_eq:S_ck}
    \mathrm{S}_{c,k} \ = \  \dfrac{N^2}{4} \  \left( \delta_{k,c} + \delta_{k,a+c} \right) - \dfrac{N^2}{8} \ \left( 1 - \left(\delta_{k,c} + \delta_{k,a+c} \right) \right).
\end{align}
Using equations \eqref{vc_eq:S_bk} and \eqref{vc_eq:S_ck} in equation \eqref{vc_eq:tnconesquared_action0}, we finally get
\begin{align}
    \label{vc_eq:tnconesquared_action}
    &  {\tcone}^2 \ \left( \left( \ket{a,b,a+b} + \ket{a,a+b,b } \right) -  \left( \ket{a,c,a+c} + \ket{a,a+c,c} \right) \right) \notag \\  =  & \ \dfrac{1}{4 N^2} \ \left( \left( \ket{a,b,a+b} + \ket{a,a+b,b } \right) -  \left( \ket{a,c,a+c} + \ket{a,a+c,c} \right) \right).
\end{align}
Note that $\vdcone^\perp$ (orthogonal complement of $\vdcone$ in $\vcone$) is spanned by vectors of the form $\ket{a,b,a+b} + \ket{a,a+b,b }  -  \left( \ket{a,c,a+c} + \ket{a,a+c,c}\right)$. Hence the eigenvalues of $\tcone$ on $\vdcone^\perp$ (which is the orthogonal complement of $\vdcone$ in $\vcone$) can be $\frac{1}{2N}$ or $-\frac{1}{2N}$. If $\ket{a_\pm}$ is an eigenvector with eigenvalue $\pm\frac{1}{2N}$, and it has the basis expansion 
\begin{equation}
    \label{vc_eq:a_pm}
    \ket{a_\pm} \  = \sum_{\substack{b: \\ \s{a}{b}=0 \\ b \neq 0,a}} \ \lambda_b \ \ket{a,b,a+b},
\end{equation}
then $\lambda_b$ should satisfy the following two conditions.
\begin{enumerate}
    \item So that $\ket{a_\pm}$ are orthogonal to $\q{a}$, the $\lambda_b$'s should satisfy
    \begin{equation}
    \label{vc_eq:a_pm_orthogonality_q}
    \sum_{\substack{b: \\ \s{a}{b}=0 \\ b\neq 0,a}} \ \lambda_b \ = \ 0. 
    \end{equation}
    \item So that it has eigenvalue $\pm \dfrac{1}{2N}$, the $\lambda_b$'s must satisfy
    \begin{equation}
        \label{vc_eq:a_pm_eigenvalue_condition}
        \sum_{\substack{h: \\ \s{a}{h}=0 \\ \s{b}{h}=1}} \ \lambda_h \ = \ \pm  \dfrac{2}{N} \ \lambda_b.
    \end{equation}
\end{enumerate}
Hence we have obtained the eigenspaces in equations \eqref{vc_eq:def_vqcone} and \eqref{vc_eq:def_vrcone}. 

\noindent Using equation \eqref{vc_eq:tncone_adjoint}, we see that $\tr \tcone = 0$. If $n_+$ is the degeneracy of $\frac{1}{2N}$, then the degeneracy of $-\frac{1}{2N}$ is $\mathrm{rank} \  \tcone - \dim \vdcone - n_+$. We get the $\tr \tcone =0$ equation in terms of the eigenvalues and degeneracies.
\begin{equation}
    \label{vc_eq:tr_tcone}
    \dfrac{1}{2N} \ n_+ - \dfrac{1}{2N} \left( \mathrm{rank} \  \tcone - \dim \vdcone - n_+ \right) + \dfrac{1}{4} \dim \vdcone,
\end{equation}
which can be solved using the fact that $\mathrm{rank} \tcone = \dim \vcone = \frac{\left(N^2-1\right)\left(N^2-4 \right)}{4}$, and that $\dim \vdcone=N^2-1$. Using this we see that the multiplicity of the $\frac{1}{2N}$ eigenvalue is $ = \frac{\left(N^2-1 \right) \left( N^2 - 2N -8 \right)}{8}$, whereas the multiplicity of the $-\frac{1}{2N}$ eigenvalue is $ = \frac{\left(N^2-1 \right) \left( N^2 + 2N -8 \right)}{8}$.
\end{proof}

\begin{lem}
\label{vc_lem:ETTp_eigenvalues} 
The largest eigenvalue of $\ETTd$ (defined in equation \eqref{vc_eq:ETTp}) is $\frac{3}{4}$, and $\wtr$ (defined in \eqref{vc_eq:def_wtr}) is the corresponding eigenspace. All remaining eigenvalues are less than or equal to $\frac{1}{4} \left( 1 + \frac{2}{N-2} \right)$, and their corresponding eigenspaces lie in the orthogonal complement of $\wtr$.
\end{lem}
\begin{proof}
It is convenient to adopt the notation defined in equation \eqref{eq:def_pj}, with the identification that $\pj{0}=\ps$, $\pj{1}=\pw$ and $\pj{2}=\pws$. 

\noindent First we obtain an orthogonal decomposition of $\ETTd$ using the orthogonal projectors $\pj{k}$: note that due to the permutation relations among $\tconep$, $\tctwop$ and $\tcthreep$ (see Lemma \ref{vc_lem:tnc_rela} and equations \eqref{vc_eq:def_tctwop} and \eqref{vc_eq:def_tcthreep}), $\ETTd$ commutes with all $\Wa{\pi}$'s, and hence consequently commutes with $\pj{k}$'s. Noting the completeness relation in equation \eqref{eq:S3_id_decomposition}, we multiply on the right (or left) by $\ETTd$.
\begin{equation}
    \label{vc_eq:S3_id_decomposition_ETTp_commuting}
     \ETTd \ = \ \sum_{k=0}^2 \ \pj{k} \ \ETTd  \ = \ \sum_{k=0}^2 \ \pj{k} \ \ETTd \ \pj{k},
\end{equation} where we used the fact that $\pj{k}^2 = \pj{k}$ and and that $\pj{k}$ commutes with $\ETTd$. Define
 \begin{align}
   \label{vc_eq:def_ETTpjk}
    & \ \ETTpjk \ := \ \pj{k} \ \ETTd \ \pj{k}.
 \end{align} Then we get the following orthogonal decomposition of $\ETTd$
 \begin{align}
     \label{vc_eq:ETTp_decomposition}
     \ETTd \  = & \ \ETTps \  +  \ETTpw \ + \ \ETTpws.
 \end{align}
 Hence the eigendecomposition of $\ETTd$ is given by the eigendecomposition of $\ETTps$, $\ETTpw$ and $\ETTpws$. To obtain these eigendecompositions, note the following. \begin{align}
    \label{vc_eq:ETCCpcs_ps_tconep_rela}
    & \ETTpjk \ = \ 3 \ \pj{k} \tconep \pj{k},
\end{align} 
which we prove using two things: the permutation relations: that $\tctwo = \Wa{[(123)]} \tcone \Wa{[(132)]}$ and $\tcthree = \Wa{[(132)]} \tcone \Wa{[(123)]}$ (obtained in Lemma \ref{vc_lem:tnc_rela}), and also that
\begin{equation}
\label{vc_eq:pjktconeppjk_invariance}
\pj{k} \Wa{\sigma} \ \tconep \Wa{\sigma^-1} \ \pj{k} \ = \ \pj{k} \tconep \pj{k}.
\end{equation}

\noindent Thus one needs to obtain the eigendecomposition of $\pj{k} \tconep \pj{k}$. From Lemma \ref{vc_lem:tncone_spectral_decomposition}, we see that $\tconep$ has the following eigendecomposition
\begin{equation}
    \label{vc_eq:tncone_spectral_decomposition}
    \tconep \ = \ \dfrac{1}{4} \ \kb{\hat{a}}{\hat{a}},
\end{equation} where $\q{a}$ was defined in equation \eqref{vc_eq:def_vdcone}. We decompose $\vdcone$ (also defined in equation \ref{vc_eq:def_vdcone}) into three orthogonal subspaces. 
\begin{align}
    \label{vc_eq:wtr2}
    &  \ \wtr \ = \ \mathrm{span} \  \left\{ \  \ket{w_\mathrm{tr}} \ = \sqrt{\frac{N^2-4}{2}} \sum_{a \in \Fs} \ \q{a}  \right\},   \\ 
    \label{vc_eq:wpp}
    & \ \wpp \ := \ \mathrm{span} \ \left\{ \ \ket{w_+} \ = \  \sum_{a \in \Fs}\  \lambda_a \ \q{a}, \ \sum_{a \in \Fs} \ \lambda_a \ = 0, \ \sum_{\substack{h: \\ \s{h}{a}}=1} \ \lambda_h \ = \ \frac{N}{2} \ \lambda_a \right\}, \\ 
    \label{vc_eq:wm}
    & \ \wm \ := \ \mathrm{span} \ \left\{ \ \ket{w_-} \ = \ \sum_{a \in \Fs}\  \lambda_a \ \q{a}, \ \sum_{a \in \Fs} \ \lambda_a \ = 0, \ \sum_{\substack{h: \\ \s{h}{a}}=1} \ \lambda_h \ = \ - \frac{N}{2} \ \lambda_a \right\},
\end{align}
We note that $\ket{w_\mathrm{tr}}$ is the same as was defined in equation \eqref{vc_eq:def_wtr}, which can be seen as follows:
\begin{align}
    \label{vc_eq:wtr_expansion}
    \sum_{a \in \Fs 
     } \ \q{a} \ \propto \sum_{a\in\Fs} \left( \sum_{\substack{b \\ \s{a}{b}=0 \\ b\neq 0,a}} \ \ket{a,b,a+b} \right) \ =  \ \sum_{\substack{a,b \\ 
    \s{a}{b}=0 \\ a,b,a+b \neq 0}} \ket{a,b, a+b}  \ = \ket{w_\mathrm{tr}}.
\end{align} Note that the equations in the definition of $\wpp$ and $\wm$ are consistent and have solutions. We know this because these are the same equations are used to decompose the irreducible representations in the diagonal sector in Subsubsections \ref{subsubsec:ET_aan} and \ref{subsubsec:ET_ana_naa}, and they were also introduced in the diagonal sector in Helsen et al \cite{Helsen2016} (see Lemma 4 therein). We also know that $\dim \wpp = \frac{N(N-1)}{2}-1$ and $\dim\wm= \frac{N(N+1)}{2}-1$. And that the three spaces $\wtr$, $\wpp$ and $\wm$ are orthogonal to each other, collectively spanning the entire $N^2-1$ dimensional support of $\id + \Wa{[(23)]}$.

\noindent We want to further decompose $\wtr$, $\wpp$ and $\wm$ using $\pj{k}$. We start with $\wtr$.

\noindent Note that $\Wa{\sigma} \ket{w_\mathrm{tr}} = \ket{w_\mathrm{tr}}$ for\footnote{That $\Wa{\pi} \ket{w_\mathrm{tr}} = \ket{w_\mathrm{tr}}$ is actually true for all $\pi \in \Sr$.} $\sigma = \id,  [(123)], [(132)]$. We demonstrate this for $\sigma = [(123)]$.
\begin{align}
    \label{vc_eq:pi_wtr}
    \Wa{[(123)]} \ \ket{w_\mathrm{tr}} \ & = \ \sum_{\substack{a,b \\ \s{a}{b}=0 \\ a,b,a+b \neq 0}} \ \Wa{[(123)]} \ \ket{a,b,a+b} \notag \\
    & = \  \sum_{\substack{a,b \\ \s{a}{b}=0 \\ a,b,a+b \neq 0}} \ \ket{a+b,a,b} \notag \\ 
    & = \  \sum_{\substack{a',b' \\ \s{a'}{b'}=0 \\ a',b',a'+b' \neq 0}} \ \ket{a',b',a'+b'}, \notag \\ 
    & = \ \ket{w_\mathrm{tr}},
\end{align}
where the dummy variables were changed in the last line as follows: $b'=a$, and $a'=a+b$ (so that $a'$ varies with $b$ for fixed $b'=a$). This tells us the following.
\begin{equation}
    \label{vc_eq:ps_pw_pws_wtr}
    \pj{0} \ket{w_\mathrm{tr}} \ = \  \ket{w_\mathrm{tr}}, \; \; \pj{1} \ket{w_\mathrm{tr}} \ = \  \pj{2} \ \ket{w_\mathrm{tr}} \  =  \ 0.
\end{equation} Finally, we have the following orthogonal decomposition of $\vdcone$.
\begin{align}
    \label{vc_eq:vdcone_decomposition}
    \vdcone \ = \ \underbrace{ \wtr \oplus \vzp \oplus \vzm}_{\supp{\ETTps}} \oplus \underbrace{\vop  \oplus \vom}_{\supp{\ETTpw}} \oplus \underbrace{\vtp  \oplus \vtm}_{\supp{\ETTpws}},
\end{align}
where \begin{align}
    \label{vc_eq:def_vjpm}
    \vjpm \ := \ \supp{\pj{k}} \cap \wpm.
\end{align}
 Since $\wpm$ and $\supp{\pj{k}}$ are both invariant under the action of $\tconep$, $\tconep$ can be decomposed as follows.
\begin{align}
    \label{vc_eq:ETTp_decomposition_two}
    \ETTd \ =   \ \  \ptr \ \   \ETTd &   \ \  + \ \   \pzp\  \ \ETTd \ \  +\  \ \pzm \  \  \ETTd \notag \\ & \  \  +  \  \ \pop \  \  \ETTd  \ \  + \  \ \  \pom  \ \  \ETTd  \  \ \notag \\
& \  \ +  \  \ \ptp  \  \ \ETTd \   \ + \   \ \ptm \  \  \ETTd,
\end{align}
where \begin{equation}
    \label{vc_eq:def_pjpm} 
    \pjpm \ := \ \vjpm.
\end{equation}


\noindent On $\wtr$ we see 
\begin{align}
\label{vc_eq:ps_tcone_ps_wtr}
  \ \ps \ \tconep \ps \ \left( \ps \ket{w_\mathrm{tr}}\right) \ = \ \dfrac{1}{4} \ \ps \ket{w_\mathrm{tr}} \ \ \ 
   \Longrightarrow \ \  \ \ETTps \ \left( \ps \ket{w_\mathrm{tr}} \right) \ = \ \dfrac{3}{4} \ \ps \ket{w_\mathrm{tr}},
\end{align}
which proves the first part of Lemma \ref{vc_lem:ETTp_eigenvalues}.

\noindent To obtain the remaining eigenvalues, note that $\pj{k}\ket{w_\pm}$ span $\vjpm$, where $\ket{w_\pm} \in \wpm$. Using the expansions of $\ket{w_\pm}$ in equations \eqref{vc_eq:wpp} and \eqref{vc_eq:wm}, we obtain the action of $\pj{k} \tcone \pj{k}$ on $\pj{k}\ket{w_\pm}$.
\begin{align}
    \label{vc_eq:pjk_tcone_pjk_wpm}
    & \ \pj{k} \tcone \pj{k} \ \left( \pj{k} \ \ket{w_\pm} \right) \notag  \\ 
    = &  \ \dfrac{1}{4} \ \sum_{a,b \in \Fs} \ \lambda_b \ \left(\frac{1}{3} \ \bra{\hat{a}} \pj{k} \q{b} \right) \ \q{a}, \notag \\
    = & \ \begin{cases} \ \ \ \ 
 \dfrac{1}{12} \left( 1 - \dfrac{2}{N+2} \right) \ \pj{0} \ket{w_+} \\  \ \ \ \ 
 \dfrac{1}{12} \left( 1 + \dfrac{2}{N-2} \right) \ \pj{0} \ket{w_-} \\ \ \ \ \ 
 \dfrac{1}{12} \left( 1 + \dfrac{1}{N+2} \right) \ \pj{1} \ket{w_+} \\ \ \ \ \ 
 \dfrac{1}{12} \left( 1 - \dfrac{1}{N-2} \right) \ \pj{1} \ket{w_-} \\  \ \ \ \ 
 \dfrac{1}{12} \left( 1 + \dfrac{1}{N+2} \right) \ \pj{2} \ket{w_+} \\ \ \ \ \ 
 \dfrac{1}{12} \left( 1 - \dfrac{1}{N-2} \right) \ \pj{2} \ket{w_-}\end{cases},
\end{align} where we used the inner product between $\pj{k}\q{a}$ and $\pj{k}\q{b}$, which is
\begin{align}
    \label{vc_eq:vnconep_innerproducts}
& \    \bra{\hat{a}} \ \pj{k} \ \ket{\hat{b}} \ 
=  \ \begin{cases}
   \frac{1}{3} \ \left( \delta_{a,b} \ + \ \frac{4}{N^2-4}  \ \left(1-\delta_{a,b}\right) \delta_{\s{a}{b},0}  \right), \; k=0, \\ 
   \frac{1}{3} \ \left( \delta_{a,b} \ - \ \frac{2}{N^2-4} \  \left(1-\delta_{a,b}\right) \delta_{\s{a}{b},0}  \right), \; k=1,2,
\end{cases}
\end{align} and also used the relation 
\begin{equation}
\label{vc_eq:sum_lambda}
\sum_{\substack{b: \\ \s{a}{b}=0 \\ b \neq a }} \ \lambda_b \ = -  \sum_{\substack{b: \\ \s{a}{b}=1}} \ \lambda_b     \ \ - \lambda_a \ = \ - \left(  \dfrac{\pm N  \  + \ 2}{2}\right) \lambda,
\end{equation}
which we obtain from equations \eqref{vc_eq:wpp} and  \eqref{vc_eq:wm}.

\noindent Hence the full eigen decomposition of $\ETTd$ is given by 
\begin{align}
    \label{vc_eq:ETTd_complete_spectral_decomposition}
    \ETTd \ = & \ \  \frac{3}{4} \ptr \ + \ \frac{1}{4} \left( 1 + \frac{2}{N-2} \right) \ \pzm \ + \ \frac{1}{4} \left( 1 + \frac{1}{N+2} \right) \ \left( \pop + \ptp \right) \ \notag \\ 
    & + \ \frac{1}{4} \left( 1 - \frac{1}{N-2} \right) \ \left( \pom + \ptm \right) \ + \ \frac{1}{4} \left( 1 - \frac{1}{N+2} \right) \ \pzp.
\end{align}
\end{proof}

\section{Proofs of equation \texorpdfstring{\eqref{eq:ET_vnc_restriction}}{}}
\label{sec:appendix:ET_vnc}
On the RHS of equation \eqref{eq:ET_action_non_commuting_sector}, we see that the first term is given by
\begin{equation}
    \label{eq:ET_non_commuting_first_term}
    \dfrac{1}{N^2} \ \sum_{\substack{h: \\
    \s{a}{h}=0 \\ \s{b}{h}=0 }} \ \ket{a,b,a+b} \ = \ \dfrac{1}{4} \ \ket{a,b,a+b},
\end{equation}  where we used Lemma \ref{lem:1} to arrive at the result. We simplify the term in the expression \eqref{eq:T} by observing the following: for each $h$ which satisfies $\s{a}{h}=0$, $\s{b}{h}=1$, $b+h$ will satisfy $\s{a}{b+h}=\s{b}{b+h}=1$ (using $\s{a}{b}=1$). Similarly, for each $h'$ that satisfies $\s{a}{h'}=\s{b}{h'}=1$, $a+h'$ will satisfy $\s{a}{b+h'}=0$, and $\s{b}{b+h'}=1$. We replace $b+h$ in the term $\kt{a}{b+h}{a+b+h}$ with $h'=h+b$, thus obtaining $\ket{a,h',b+h'}$ and sum over $h'$ which satisfy $\s{a}{h'}=\s{b}{h'}=1$. Removing the primed symbol from $h'$ (since it is a dummy variable), we see that the expression in equation \eqref{eq:T} is
    \begin{equation}
    \label{eq:T2}
 \ \dfrac{1}{N^2} \  \ \sum_{\substack{h: \\
    \s{a}{h}=0 \\ \s{b}{h}=1 }} \ \kt{a}{b+h}{a+b+h} \  =  \ \dfrac{1}{N^2} \ \sum_{\substack{h: \\
    \s{a}{h}=1 \\ \s{b}{h}=1 }} \ \ket{a,h,a+h}.
    \end{equation} \noindent The expression in equation \eqref{eq:S} is simplified in a similar way to give us  
    \begin{equation}
    \label{eq:S2}
    \dfrac{1}{N^2} \ \sum_{\substack{h: \\ 
    \s{a}{h}=1 \\ \s{b}{h}=0 }} \ \kt{a+h}{b}{a+b+h} \ = \
    \dfrac{1}{N^2} \  \ \sum_{\substack{h: \\
    \s{a}{h}=1 \\ \s{b}{h}=1 }} \ \ket{h,b,b+h}.
    \end{equation} \noindent Finally, the expression in equation \eqref{eq:R} simplifies to give us
    \begin{equation}
        \label{eq:R2}
\dfrac{1}{N^2} \ \sum_{\substack{h: \\ 
    \s{a}{h}=1 \\ \s{b}{h}=1 }} \ \kt{a+h}{b+h}{a+b} \ = \
    \dfrac{1}{N^2} \  \ \sum_{\substack{h: \\
    \s{b}{h}=1 \\ \s{a+b}{h}=1 }} \ \ket{a+b+h,h,a+b}.
    \end{equation}

    \section{Permutation relations among  \texorpdfstring{$\tnconepm$}{}, \texorpdfstring{$\tnctwopm$}{} and \texorpdfstring{$\tncthreepm$}{} }
    \label{sec:appendix_tncpm_rela}
    
    We prove the following 
    \begin{itemize}
    \item $\tnconepm$, $\tnctwopm$ and $\tncthreepm$ commute with $\Wa{[(23)]}$, $\Wa{[(13)]}$, and $\Wa{[(12)]}$ respectively. \begin{align}
    \label{eq:tnconepm_axis_of_symmetry}
    \Wa{[(23)]} \ \tnconepm \ \Wa{[(23)]} \ = \ \tnconepm,\\
    \label{eq:tnctwopm_axis_of_symmetry}
    \Wa{[(13)]} \ \tnctwopm \ \Wa{[(13)]} \ = \ \tnctwopm, \\
    \label{eq:tncthreepm_axis_of_symmetry}
    \Wa{[(12)]} \ \tncthreepm \ \Wa{[(12)]} \ = \ \tncthreepm. 
\end{align}
\item $\tnconepm$, $\tnctwopm$ and $\tncthreepm$ are permutations of each other:  \begin{align}
\label{eq:tncpm_rela_3_1}
\tncthreepm \ = & \    \Wa{[(13)]} \  \tnconepm \ \Wa{[(13)]} \\ 
\label{eq:tncpm_rela_3_2}
= & \  \  \Wa{[(132)]} \ \tnconepm \Wa{[(123)]}  \\
= & \  \Wa{[(123)]} \ \tnctwopm \Wa{[(132)]} \notag \\
= & \   \Wa{[(23)]} \  \tnctwopm \ \Wa{[(23)]}, \notag  \\
\label{eq:tncpm_rela_2}
\ \tnctwopm \ = \ & \ \Wa{[(123)]} \ \tnconepm \Wa{[(132)]} \\ 
= & \  \Wa{[(12)]} \  \tnconepm \ \Wa{[(12)]}  \notag \\ 
=  & \  \Wa{[(132)]} \ \tncthreepm \Wa{[(123)]} \notag \\
= & \ \Wa{[(23)]} \  \tncthreepm \ \Wa{[(23)]} \notag \\ 
\label{eq:tncpm_rela_1}
\tnconepm \ = & \ \Wa{[(123)]} \ \tncthreepm \Wa{[(132)]} \\
= & \   \Wa{[(13)]} \  \tncthreepm \ \Wa{[(13)]} \notag  \\
=  &  \  \Wa{[(132)]} \ \tnctwopm \Wa{[(123)]} \notag \\
= \ &   \Wa{[(12)]} \  \tnctwopm \ \Wa{[(12)]}, \notag 
\end{align}
\item $\ETT_\pm$ commutes with $\Wa{\pi}$ for all $\pi \in \Sr$.
\end{itemize}

\begin{proof}
Equation \eqref{eq:tnconepm_axis_of_symmetry} is proved by conjugating both sides of equations \eqref{eq:def_tnconem} and \eqref{eq:def_tnconep} by $\Wa{[(23)]}$, and noting that the RHS of both equations remain invariant. Equations \eqref{eq:tnctwopm_axis_of_symmetry} and \eqref{eq:tncthreepm_axis_of_symmetry} can be proved similarly.

\noindent Equation \eqref{eq:tncpm_rela_3_1} is proved by conjugating equations \eqref{eq:def_tnconem} and \eqref{eq:def_tnconep} by $\Wa{[(13)]}$, and by noting that $[(13)][(23)][(13)]=[(12)]$, and using equation \eqref{eq:tnc_rela_3_1}, one then notes that the RHS are the same as equations \eqref{eq:def_tncthreem} and \eqref{eq:def_tncthreep} respectively. To prove relation \eqref{eq:tncpm_rela_3_2}, conjugate equation \eqref{eq:tnconepm_axis_of_symmetry} by $\Wa{[(13)]}$, and use equation \eqref{eq:tncpm_rela_3_1} and the relation:  $[(132)]=[(13)][(23)]$. All other equations can be proved using identical steps. 

\noindent That $\ETT_\pm$ commutes with $\Wa{\pi}$ for all $\pi \in \Sr$ now follows since $\ETT_\pm = \tnconepm+\tnctwopm+\tncthreepm$.
\end{proof}
    
\section{Inner elements of the orthogonal complement of r products between vectors in \texorpdfstring{$\vncones$}{ }}
\label{sec:appendix_inner}
We prove 
\begin{align}
\label{eq:inner_a_ps_b_appendix}
 & \ \bra{\hat{a}} \ps \q{b} \ = \ \frac{1}{3} \ \left( \delta_{a,b} \ + \ \dfrac{4}{N^2} \ \delta_{\s{a}{b},1} \right). \\
 \label{eq:inner_a_pw_or_pws_b}
 & \ \bra{\hat{a}} \   \pw \ \q{b} \ = \  \bra{\hat{a}} \   \pws \ \q{b}  \ = \  \frac{1}{3} \ \left( \delta_{a,b} \ - \ \dfrac{2}{N^2} \ \delta_{\s{a}{b},1} \right),
 \end{align}
 by proving it for general $\pj(k)$ defined in equation \eqref{eq:def_pj}. 
\begin{align}
    \label{eq:vnconep_innerproducts}
& \    \bra{\hat{a}} \ \pj{k} \ \ket{\hat{b}} \notag \\ 
= &  \ \left( \frac{\sqrt{2}}{N}  \sum_{\substack{h:\\
\s{a}{h}=1}}  \bra{a,h,a+h} \ \right)  \ \frac{1}{3} \left( \id + \omega^k \ \Wa{[(123)]} \ + \omega^{2k} \ \Wa{[(132)]} \right) \  \left(  \frac{\sqrt{2}}{N}  \sum_{\substack{k:\\
\s{b}{k}=1}}  \ket{b,k,a+k}     \right)  \notag \\ 
= &  \  \frac{2}{N^2}  \sum_{\substack{h,k:\\
\s{a}{h}=1 \\ \s{b}{k}=1 }} \frac{1}{3} \left( \delta_{a,b}  \ \delta_{h,k} + \omega^k \ \delta_{h,b}\ \delta_{k,a+b} \ + \omega^{2k} \ \delta_{k,a} \ \delta_{h,a+b} \right)   \notag \\
= &  \  \frac{1}{3} \ \left( \  \delta_{a,b} \ + \  \left( \omega^k + \omega^{2k}  \right) \ \frac{2}{N^2} \  \delta_{\s{a}{b},1} \ \right). \notag \\ 
= & \ \begin{cases}
   \frac{1}{3} \ \left( \delta_{a,b} \ + \ \frac{4}{N^2} \  \delta_{\s{a}{b},1}  \right), \; k=0, \\ 
   \frac{1}{3} \ \left( \delta_{a,b} \ - \ \frac{2}{N^2} \  \delta_{\s{a}{b},1}  \right), \; k=1,2,
\end{cases}
\end{align}
which proves equations \eqref{eq:inner_a_ps_b} and \eqref{eq:inner_a_pw_or_pws_b}.

\section{Proofs of equation \texorpdfstring{\eqref{vc_eq:ET_vnc_restriction}}{}}
\label{vc_sec:appendix:ET_vnc}
On the RHS of equation \eqref{vc_eq:ET_action_commuting_sector}, we see that the first term is given by
\begin{equation}
    \label{vc_eq:ET_non_commuting_first_term}
    \dfrac{1}{N^2} \ \sum_{\substack{h: \\
    \s{a}{h}=0 \\ \s{b}{h}=0 }} \ \ket{a,b,a+b} \ = \ \dfrac{1}{4} \ \ket{a,b,a+b},
\end{equation}  where we used Lemma \ref{lem:1} to arrive at the result. We simplify the term in the expression \eqref{vc_eq:T} by observing the following: for each $h$ which satisfies $\s{a}{h}=0$, $\s{b}{h}=1$, $b+h$ will satisfy $\s{a}{b+h}=0$, and $\s{b}{b+h}=1$ (using $\s{a}{b}=0$). Similarly, for each $h'$ that satisfies $\s{a}{h'}=0$, and $\s{b}{h'}=1$, $a+h'$ will satisfy $\s{a}{b+h'}=0$, and $\s{b}{b+h'}=1$. We replace $b+h$ in the term $\kt{a}{b+h}{a+b+h}$ with $h'=h+b$, thus obtaining $\ket{a,h',b+h'}$ and sum over $h'$ which satisfy $\s{a}{h'}=0$, $\s{b}{h'}=1$. Removing the primed symbol from $h'$ (since it is a dummy variable), we see that the expression in equation \eqref{vc_eq:T} is
    \begin{equation}
    \label{vc_eq:T2}
 \ \dfrac{1}{N^2} \  \ \sum_{\substack{h: \\
    \s{a}{h}=0 \\ \s{b}{h}=1 }} \ \kt{a}{b+h}{a+b+h} \  =  \ \dfrac{1}{N^2} \ \sum_{\substack{h: \\
    \s{a}{h}=0 \\ \s{b}{h}=1 }} \ \ket{a,h,a+h}.
    \end{equation} \noindent The expression in equation \eqref{vc_eq:S} is simplified in a similar way to give us  
    \begin{equation}
    \label{vc_eq:S2}
    \dfrac{1}{N^2} \ \sum_{\substack{h: \\ 
    \s{a}{h}=1 \\ \s{b}{h}=0 }} \ \kt{a+h}{b}{a+b+h} \ = \
    \dfrac{1}{N^2} \  \ \sum_{\substack{h: \\
    \s{a}{h}=1 \\ \s{b}{h}=0 }} \ \ket{h,b,b+h}.
    \end{equation} \noindent Finally, the expression in equation \eqref{vc_eq:R} simplifies to give us
    \begin{equation}
        \label{vc_eq:R2}
\dfrac{1}{N^2} \ \sum_{\substack{h: \\ 
    \s{a}{h}=1 \\ \s{b}{h}=1 }} \ \kt{a+h}{b+h}{a+b} \ = \
    \dfrac{1}{N^2} \  \ \sum_{\substack{h: \\
    \s{b}{h}=1 \\ \s{a+b}{h}=0 }} \ \ket{a+b+h,h,a+b}.
    \end{equation} 
    \section{Proof of Theorem \ref{thm:nc_inv_subsp}}
    \label{sec:appendix_nc_0}
    We begin by proving equation \eqref{eq:nc_d_intertwiner}. For that, first note that $\dim \mathrm{span} \left\{ \ket{a,a} \ | \ a \in \Fs \right\} = N^2-1$. Also, the vectors $\ket{a,a}$ form an ONB for $\mathrm{span} \left\{ \ket{a,a} \ | \ a \in \Fs \right\}$. From equation \eqref{eq:inner_11} in Lemma \ref{lem:A_subspace_inner}, see that $\dim \mathrm{span} \left\{ \ket{\hat{a}} \ | \ a \in \Fs \right\} = N^2-1$. Also, the vectors $\ket{\hat{a}}$ are an ONB for $\mathrm{span} \left\{ \ket{\hat{a}} \ | \ a \in \Fs \right\}$. Hence the dimensions of both spaces are the same. Next note that
    \begin{equation}
        \label{lem:transvection_d}
        \Ua{h} \ \ket{a,a} \ = \ket{a+\s{a}{h}h,a+\s{a}{h}h}, \ \forall \ h \in \Fm, 
    \end{equation}
    which can be directly verified by substituting from equation \eqref{eq:Th_matrix_form} into equation \eqref{eq:Clifford_action_special_basis} for $t=2$. Next from Lemma \ref{lem:Uh_A_action}, note that by putting $b=0$ in equation \eqref{eq:Uh_ab_action}, we get
    \begin{equation}
        \label{eq:Uh_hat_a}
        \Ua{h} \ \ket{\hat{a}} \ = \ \ket{\hat{a} +  \s{a}{h}\hat{h} }.
    \end{equation}
Thus it is seen that $L^{(d)\rightarrow (\mathrm{nc})}$ satisfies the equation $ \Ua{h} \  L^{(d)\rightarrow (\mathrm{nc})} \ \ket{a,a} \ = \ L^{(d)\rightarrow (\mathrm{nc})} \ \Ua{h} \ \ket{a,a}$ for all $a \in \Fs$. Note that $\Cl$ is  generated by all Clifford transvection and Paulis. The action of the diagonal sector subspace (Lemma 4 in \cite{Helsen2016}) is invariant under the adjoint action of Paulis, just like all vectors of the form $\ket{a,b,a+b}$ are (see Lemma \ref{lem:Pauli_basis_choice}). This immediately tells us that  $L^{(d)\rightarrow (\mathrm{nc})}$ is an intertwiner between both representations of the Clifford group. This proves equation \eqref{eq:nc_c_dd}. ~ \newline ~ \newline 
Next, we prove that the subrepresentations acting on the subspaces $V_1 \oplus V_2$ in Lemma 4 of \cite{Helsen2016} are equivalent to the subrepresentations acting on  $V^{(\mathrm{nc})}_{1} \oplus V^{(\mathrm{nc})}_{2}$ in equation \eqref{eq:nc_V_b}. Lemma 4 in \cite{Helsen2016} and Lemma \ref{lem:dimA} tell us that $\dim \left( V_1 \oplus V_2 \right) = \dim \left(  V^{(\mathrm{nc})}_1 \oplus V^{(\mathrm{nc})}_2 \right) = N^2-2$. Lemma \ref{lem:Uh_A_action} informs us that for any transvection Clifford $\Ua{h}$, and non-zero vector $\ket{x} \in V^{(\mathrm{nc})}_1 \oplus V^{(\mathrm{nc})}_2$, with the expansion $\ket{x}=\sum_{a\in \Fs} \lambda_a \ket{A(a)}$, we must have that $\sum_{a\in\Fs} \lambda_a=0$ (Lemma \ref{lem:dimA}), we get that $\Ua{h}\ket{x}=\sum_{a\in \Fs} \lambda_a \ket{A(a+\s{h}{a}h)}$ (using equation \eqref{eq:Uh_ab_action}). Similarly for the same choice of $\lambda_a$, the vector $\ket{v} \in V_1 \oplus V_2$ in the diagonal sector subspace, $\ket{v} = \ \sum_{a \in \Fs} \lambda_a \ket{a,a}$, the action of the same transvection Clifford is $\Ua{h} \ket{v}  \ = \ \sum_{a \in \Fs} \lambda_a \ket{a+\s{a}{h}h,a+\s{a}{h}h}$. Define $K:V^{(\mathrm{nc})}\oplus V^{(\mathrm{nc})} \rightarrow V_1 \oplus V_2$ as 
\begin{equation}
\label{eq:K_def}K \left( \sum_{a\in \Fs} \lambda_a \ket{A(a)} \right) \ = \   \sum_{a \in \Fs} \lambda_a \ket{a,a},\end{equation} with $\lambda_a$ such that $\sum_{a \in \Fs} \lambda_a =0$. It is easily verified that $K$ is linear. Thus $K$ is an intertwiner between both subrepresentations of the Clifford group.  \newline \noindent First note that $\mathrm{P}_S$ and $\mathcal{P}_1$ are both projectors, such that $\mathcal{P}_s \mathcal{P}_1 = 0$, i.e., they have orthogonal supports. Next, $\ET$ commutes with $W_\sigma$ for all $\sigma \in \Sr$. Thus, the action of $\ET$ on $V^{(\mathrm{nc})}_{\mathrm{null,1}}$ and on $ V^{(\mathrm{nc})}_{\mathrm{null,S}}$ is completely determined by the action of $\ET$ on the subspace spanned by $\ket{a;b} + \ket{b;a+b} + \ket{a+b;a}$ for $a,b \in \Fs$ such that $\s{a}{b}=0$ and $a \neq b$. Lemma \ref{lem:ET_action_Vnull} tells us that $\left(\ET - \ \frac{1}{4} \mathrm{id} \right) \left( \ket{\hat{a};b} + \ket{\hat{b};a+b} + \ket{\hat{a+b};a} \right) =0$. Thus vectors of the form $\ket{\hat{a};b} + \ket{\hat{b};a+b} + \ket{\hat{a+b};a}$ lie in the $1/4$ eigenspace of $\ET$. Since $\ET$ commutes with $\Ua{h}$ for each $h \in \Fm$, and since $\Ua{h}$ generate the action of the Clifford group on $\mathrm{span} \left\{ \ket{a,b,a+b} \ | \ a,b\in \Fs, \ a \neq b \right\}$, the result follows from a corollary of Schur's lemma. 
\begin{lem}
    \label{lem:A_subspace_inner}
    For any $a,b \in \Fs$, we have that 
    \begin{equation}
        \label{eq:A_subspace_inner1}
        \bk{A(a)}{A(b)} \ = \ \left( \dfrac{N^2}{2}-1 \right) \delta_{a,b} \ - \ \delta_{\s{a}{b},0}.
    \end{equation}
    \end{lem}
    \begin{proof}
    From equation \eqref{eq:nc_a_b_0}, we get that for any $h,k \in \Fs$ and $a,b \in \Fm$,
    \begin{align}
        \label{eq:inner_11}
    &     \bk{\hat{h};a}{\hat{k};b} \notag \\ 
        \ = \ & \frac{2}{N^2} \left( \sum_{s:\s{s}{h}=1} \ \left( -1 \right)^{\s{s}{a}} \ \bra{h,s,h+s}  \right) \ \left( \sum_{t:\s{t}{k}=1} \ \left( -1 \right)^{\s{t}{b}} \ \ket{k,t,k+t}  \right)  \notag  \\
        \ = \ & \frac{2}{N^2} \  \sum_{s:\s{s}{h}=1} \ \sum_{t:\s{t}{k}=1} \  \ \left( -1 \right)^{\s{s}{a}+\s{t}{b}} \   \bk{h,s,h+s}{k,t,k+t}  \notag  \\ 
\ = \ & \frac{2}{N^2} \ \delta_{h,k}  \  \sum_{s:\s{s}{h}=1}  \ \left( -1 \right)^{\s{s}{a+b}},  \notag  \\ 
\ = \ & \delta_{h,k} \left( \delta_{a,b} \ - \ \delta_{h,a+b} \right).
    \end{align} The last line can be explained as follows. If $a=b$, then from Lemma \ref{lem:1} there are $N^2/2$ $s \in \F$ such that $\s{s}{h}=1$, and we get that $\bk{\hat{h};a}{\hat{k};b}=\delta_{h,k}\delta_{a,b}$. If $h=a+b$, then, using the same arguments as above we see that $\bk{\hat{h};a}{\hat{k};b}=-\delta_{h,k}\delta_{h,a+b}$. When $h$ and $a+b$ are linearly independent, then Lemma \ref{lem:1} informs us that there are $N^2/4$ solutions for $s$ for the simultaneous equations $\s{s}{h}=1$, $\s{s}{a+b}=0$, and again $N^2/4$ solutions for the simultaneous equations $\s{s}{h}=1$, $\s{s}{a+b}=1$. Thus, the summation gives $0$. \newline \noindent Next, from equation \eqref{eq:Aab}, 
    \begin{align}
        \label{eq:inner_product_Aab}
          & \bk{A(a)}{A(b)} \notag \\
          \  = \ &  \left( \ \sum_{h: \substack{\s{h}{a}=0 \\ h \neq a,0}} \ \bra{\hat{h};a} \right) \left( \ \sum_{k: \substack{\s{k}{b}=0 \\ k \neq b,0}} \ \ket{\hat{k};b} \right) \notag \\ 
          \ = \ & \ \sum_{h: \substack{\s{h}{a}=0 \\ h \neq a,0}} \ \sum_{k: \substack{\s{k}{b}=0 \\ k \neq b,0}}  \ \delta_{h,k} \left( \delta_{a,b} \ - \ \delta_{h,a+b} \right) \notag \\ 
          \ = \ & \  \delta_{a,b} \left( \sum_{h: \substack{\s{h}{a}=0 \\ h \neq a,0}} \ 1  \right) \ -\left(  \delta_{\s{a}{b},0}  -  \delta_{a,b}\right)  \notag \\ 
          \ = \ & \  \delta_{a,b} \ \left(  \frac{N^2}{2} - 1 \right) \  - \  \delta_{\s{a}{b},0}.
    \end{align}
    The second last line can be explained as follows: if $a=b$, we only need to count the number of $h \in \F$ such that $\s{h}{a}=0$, which is $N^2/2$ from Lemma \ref{lem:1}. Excluding the cases $h=0,a$, we get that this sum of $N^2/2-2$. Note that when $a=b$, $\s{a}{b}=0$. This verifies the final line. If, instead, $h=a+b$, we need to count over all $h$, which are non-zero, such that $\s{h}{a}=\s{h}{b}=0$, and $h=a+b$. If $\s{a}{b}=1$, then these conditions won't be satisfied. If, $\s{a}{b}=0$, this condition is satisfies only for one choice of $h$, i.e., $h=a+b$. Note that $h\neq0$. This implies that $a \neq b$. Thus the sum, in this case simply becomes $- \delta_{\s{a}{b},0}$. Thus we see that the final line is correct in this case. Hence equation \eqref{eq:Aab} is proved. 
    \end{proof}

    \begin{lem}
    \label{lem:dimA}
   The dimension of the subspace spanned by vector $\ket{A(a)}$, which are defined in equation \eqref{eq:Aab} is
        \begin{equation}
            \label{eq:dim_Aab}
\dim \mathrm{span}\left\{ \ket{A(a)} \ | \ a \in \Fs \right\}=N^2-2.
        \end{equation}
We also have that $\sum_{a \in \Fs} \ket{A(a)}=0$, and any non-zero vector $\ket{x} \in  \mathrm{span}\left\{ \ket{A(a)} \ | \ a \in \Fs \right\}$ of the form $\ket{x} \ = \ \sum_{a\in\Fs} \lambda_a \ket{A(a)}$, has the property $\sum_{a \in \Fs} \lambda_a =0$, and for any choice of $\lambda_a$ such that $\sum_{a\in\Fs} |\lambda_a|^2 >0$, and $\sum_{a \in \Fs} \lambda_a =0$, $\ket{x}$ is non-zero.
\end{lem}
\begin{proof}
    Consider an $\left(N^2-1\right) \times \left(N^2-1\right)$ matrix $H$, whose matrix elements are indexed by vectors in $\Fs$, and 
    \begin{equation}
        \label{eq:MatrixH}
        H_{a,b} \ = \ \delta_{\s{a}{b},0}, \ \forall \ a,b \in \Fs.    \end{equation}
Using Lemma \ref{lem:1}, we see that $\sum_{b \in \Fs} H_{a,b} \ = \ \frac{N^2}{2}-1$, for all $a \in \Fs$. Thus $\frac{N^2}{2}-1$ is an eigenvalue of $H$. We compute the matrix elements of $H^2$ as follows.
\begin{align}
    \label{eq:Hsquared}
    \left(H^2 \right)_{a,b} \ = \  \left|  \left\{  \ c \in \Fs \ | \ \s{a}{c}=\s{b}{c}=0 \right\}\right| \ = \ \dfrac{N^2}{4} \delta_{a,b} \ + \ \dfrac{N^2}{4}-1,
\end{align} which we obtain as follows: when $a=b$, Lemma \ref{lem:1} tell us that $\left|  \left\{  \ c \in \Fs \ | \ \s{a}{c}=0 \right\}\right|=\frac{N^2}{2}-1$, whereas when $a \neq b$, Lemma \ref{lem:1} tell us that $\left|  \left\{  \ c \in \Fs \ | \ \s{a}{c}=\s{b}{c}=0 \right\}\right|=\frac{N^2}{4}-1$. This agrees with the final line in equation \eqref{eq:Hsquared}. Let $J$ be an $\left(N^2-1 \right) \times \left( N^2-1\right)$ matrix, whose all matrix elements are $1$. Then $H$ satisfies the equation:
\begin{equation}
    \label{eq:equationH}
    H^2 \ = \ \frac{N^2}{4} \unity_{N^2-1} \ + \ \left( \frac{N^2}{4}-1 \right) \ J.
\end{equation}
$\mathrm{Rank} J =1$, and the corresponding eigenspace of $H^2$ is given by $\sum_{b \in \Fs} \left( H^2 \right)_{a,b} \  = \ \frac{N^2-4}{4}$, as expected, since it should be the square of the eigenvalue of $H$, for the same eigenvector of $H$. Since the eigenspaces of the other eigenvalues of $H^2$ lie in $\mathrm{ker}J$, the other eigenvalues have to be $\pm\frac{N}{2}$. The multiplicity of the eigenvalue $N^2/2-1$ is simply $1$. 
\newline \noindent Let $G$ be the $\left(N^2-1\right) \times \left(N^2-1\right)$ gram matrix of the vectors $\ket{A(a)}$ defined in equation \eqref{eq:Aab}. Then $G_{a,b} \ = \ \bk{A(a)}{A(b)}$. From Lemma \ref{lem:A_subspace_inner}, we get that $G= \left( \frac{N^2}{2}-1 \right) \unity_{N^2-1} \ - \ H$, thus it has eigenvalues $0$, $\frac{N^2 \pm N -2}{2}$, where the multiplicity of the eigenvalue $0$ is $1$. This proves that $\mathrm{rank}G= \dim \mathrm{span}\left\{ \ket{A(a)} \ | \ a \in \Fs \right\}= N^2-1$. This also proves that $\sum_{a \in \Fs} \ket{A(a)}=0$, since $||\sum_{a \in \Fs} \ket{A(a)} ||^2=0$, which is inferred by observing that $G J = 0$. Since $\dim  \mathrm{span}\left\{ \ket{A(a)} \ | \ a \in \Fs \right\} = N^2-2$, this tells us that any vector $\ket{x} \in  \mathrm{span}\left\{ \ket{A(a)} \ | \ a \in \Fs \right\}$ of the form $\sum_{a\in\Fs} \lambda_a \ket{A(a)}$, where $\sum_{a \in \Fs} \lambda_a =0$, is non-zero (because its norm, which may be computed from $G$ will be non-zero). 
\end{proof}

\begin{lem}
\label{lem:Uh_A_action}
For any $a,h \in \Fs$, and $b\in \Fm$, we have that 
\begin{align}
    \label{eq:Uh_ab_action}
  &  \Ua{h} \ \ket{\hat{a};b} \ = \ \ket{\hat{a}+\s{a}{h}\hat{h};b+\s{b}{h}h}, \\ 
    \label{eq:Uh_A_action}
  &  \Ua{h} \ \ket{A(a)} \ = \ \ket{A \left( a + \s{a}{h}h \right) }.
\end{align}
\end{lem}
\begin{proof}
We verify equation \eqref{eq:Uh_ab_action} first. From equation \eqref{eq:nc_a_b_0}, and equation \eqref{eq:Uh_action_vectorized} we get that 
\begin{align}
    \label{eq:Uh_ab_action1}
     \Ua{h} \ \ket{a;b} \ 
 \ \propto \ & \sum_{k:\s{k}{a}=1} \ \left( -1 \right)^{\s{k}{b}} \ \ket{a+\s{h}{a}h,k+\s{k}{h}k,a+k+\s{h}{a+k}h}  \  \notag \\ 
 \ = \ & \sum_{k':\s{k}{a}=1} \ \left( -1 \right)^{\s{k'}{b+\s{h}{b}}} \ \ket{a+\s{h}{a}h,k',a+k'+\s{h}{a}h},
\end{align} where the last line is observed by noting that if for any $k$, we define $k'=k+\s{k}{h}h$, then $\s{k'}{a+\s{a}{h}h}=\s{k}{a}$. Replacing $a$ with $b$ we get, similarly, that $\s{k'}{b+\s{h}{b}h}=\s{k}{b}$. Summing over the new dummy variables $k'$ instead of the old variable $k$, we get the last line. \newline \noindent Next, from equation \eqref{eq:Aab} and equation \eqref{eq:Uh_ab_action}, we get 
\begin{align}
    \label{eq:Uh_A_action1}
    \Ua{h} \ \ket{A(a)} 
    = \ & \sum_{k: \substack{\s{k}{a}=0 \\ k \neq a,0}} \   \Ua{h} \ \ket{\hat{k};a} \notag \\ 
     = \ &  \sum_{k: \substack{\s{k}{a}=0 \\ h \neq a,0}} \   \Ua{h} \ \ket{\hat{k+\s{h}{k}h};a+\s{a}{h}h} \notag \\ 
     = \ & \sum_{k':\substack{\s{k'}{a+\s{h}{a}h}=0 \\ k' \neq 0, a+\s{h}{a}h}} \ \ket{\hat{k'};a+\s{h}{a}h},
\end{align} where we explain the last line as follows. For any $k$, denote $k'\coloneqq k+ \s{h}{k}h$. Then we see that $\s{k'}{a}=\s{k}{a}$. Also,  $k=\s{h}{k}h$, then $k=\s{h}{\s{h}{k}h}h=0$. We rule out $k=0$ in the summation over $k$. Similarly, if $k+a=\s{k+a}{h}h$, then $k+a=0$, which we rule out in the summation. Thus $k'\neq0,a+\s{h}{a}h$. Replacing the dummy variable $k$ with the newer dummy variable $k'$ thus gives us the last line.
\end{proof}
\begin{lem}
\label{lem:ET_action_Vnull}
We have that 
\begin{equation}
    \label{eq:ET_action_Vnull}
\left(\ET - \ \frac{1}{4} \mathrm{id} \right) \left( \ket{\hat{a};b} + \ket{\hat{b};a+b} + \ket{\hat{a+b};a} \right) =0.
\end{equation}
\end{lem}
\begin{proof}
Note that 
\begin{align}
    \label{eq:ET_action_Vnull1}
    & \ET \left( \ket{a;b} + \ket{b;a+b} + \ket{a+b;a} \right) \notag \\
    \ \propto \ & \sum_{h \in \Fm} \  \left(  \begin{aligned}   \ \ket{a+\s{a}{h}h;b+\s{b}{h}h}   \\ + \ket{b+\s{b}{h}h;a+b+\s{a+b}{h}h} \\
    + \ket{a+b+ \s{a+b}{h}h;a+\s{a}{h}h} \end{aligned}  \right)  \notag \\
    \ = \ & \sum_{\substack{h \in \Fm \\ \s{a}{h}=0 \\ \s{b}{h}=0}} \left( \ket{a;b} + \ket{b;a+b} + \ket{a+b;a} \right)  \notag \\ 
    + & \sum_{\substack{h \in \Fm \\ \s{a}{h}=1 \\ \s{b}{h}=0}} \left( \ket{a+h;b} + \ket{b;a+b+h} + \ket{a+b+h;a+h} \right)  \notag \\
    + & \sum_{\substack{h \in \Fm \\ \s{a}{h}=0 \\ \s{b}{h}=1}} \left( \ket{a;b+h} + \ket{b+h;a+b+h} + \ket{a+b+h;a} \right)  \notag \\
    + & \sum_{\substack{h \in \Fm \\ \s{a}{h}=1 \\ \s{b}{h}=1}} \left( \ket{a+h;b+h} + \ket{b+h;a+b} + \ket{a+b;a+h} \right).
\end{align} In the second line we re-express the summation \begin{equation}
    \label{eq:ET_action_Vnull1_term1}
\sum_{\substack{h\in\Fm\\ \s{a}{h}=1 \\ \s{b}{h}=0}} \ket{a+h;b} = \sum_{\substack{h' \in \Fm \\ \s{a}{h'}=1 \\ \s{b}{h'}=0}} \ \ket{h';b}
\end{equation} by defining a new dummy variables $h'\coloneqq a+h$ using the old dummy variable $h$ and noting that it satisfies the conditions: $\s{a}{h'}=\s{a}{h}$ and $\s{b}{h'}=\s{b}{h}$. Similarly one may express \begin{equation}
    \label{eq:ET_action_Vnull1_term2}
\sum_{\substack{h\in\Fm\\ \s{a}{h}=1 \\ \s{b}{h}=0}} \ket{a+b+h;a+h} = \sum_{\substack{h' \in \Fm \\ \s{a}{h'}=1 \\ \s{b}{h'}=0}} \ \ket{h';h'+b},
\end{equation} by defining a new dummy variable $h'\coloneqq a+b+h$ using the old dummy variable and noting that $\s{h'}{a}=\s{h}{a}$ and $\s{h'}{b}=\s{h}{b}$. Using equation \eqref{eq:nc_a_b_0}, we see that the sum of both term is $0$. The term 
\begin{equation}
    \label{eq:ET_action_Vnull1_term3}
\sum_{\substack{h\in\Fm\\ \s{a}{h}=1 \\ \s{b}{h}=0}} \ket{b;a+b+h} =0
\end{equation} because for any $h \in \left\{ k \in \Fm \ | \ \s{a}{k}=1, \ \s{b}{k}=0 \right\}$, $b+h$ will also belong to the same set. Thus the vectors $\ket{b;a+b+h}$ and $\ket{b;a+h}$ both occur in the summation and from equation \eqref{eq:nc_a_b_0} we see that their sum is $0$. All terms will cancel out this way. \newline
Above we explained in detail while the sum in the second line in equation \eqref{eq:ET_action_Vnull1} is $0$. The third and fourth lines are also $0$, and this can be proved using the same nature of arguments. This proves the lemma. 
\end{proof}
    \section{Proof of Theorem \ref{thm:nc_1_decomposition}}
    \label{sec:appendix:nc_W}
    Note that all the eigenspaces listed in Theorem \ref{thm:nc_inv_subsp}
 span $\mathrm{supp}  \dfrac{1}{2}\left( \mathrm{id}+ W_{[(23)]}\right)$. Also note that $W $ anti-commutes with $\dfrac{1}{2}\left( \mathrm{id}+ W_{[(23)]}\right)$, while commuting with the action of the Clifford group (due to Schur-Weyl duality). This proves that if an irreducible representation within $V^{(\mathrm{nc})}_x$ isn't annihilated by $W$, then $W$'s action on takes it to an orthogonal subspace, which carries the same irreducible representation. Hence proved. 


\begin{thebibliography}{9}
\bibitem{RB1} Emerson, Joseph; Alicki, Robert; Zyczkowski, Karol (2005). "Scalable noise estimation with random unitary operators". Journal of Optics B: Quantum and Semiclassical Optics. 7 (10): S347. 
\bibitem{RB2} Dankert, Christoph; Cleve, Richard; Emerson, Joseph; Livine, Etera (2009). "Exact and Approximate Unitary 2-Designs: Constructions and Applications". Physical Review A. 80: 012304.
\bibitem{RB3} Knill, E; Leibfried, D; Reichle, R; Britton, J; Blakestad, R; Jost, J; Langer, C; Ozeri, R; Seidelin, S; Wineland, D.J. (2008). "Randomized benchmarking of quantum gates". Physical Review A. 77 (1): 012307.
\bibitem{RB4} Magesan, Easwar; Gambetta, Jay M.; Emerson, Joseph (2011). "Scalable and Robust Randomized Benchmarking of Quantum Processes". Physical Review Letters. 106 (31–9007): 180504.
\bibitem{RB5} Magesan, Easwar; Gambetta, Jay M.; Emerson, Joseph (2012). "Characterizing quantum gates via randomized benchmarking". Physical Review A. 85 (1050–2947): 042311. 
\bibitem{RB6} Dugas, A; Wallman, J; Emerson, J (2015). "Characterizing universal gate sets via dihedral benchmarking". Physical Review A. 92 (6): 060302.
\bibitem{RB7} Dugas, Arnaud; Boone, Kristine; Wallman, Joel; Emerson, Joseph (2018). "From randomized benchmarking experiments to gate-set circuit fidelity: how to interpret randomized benchmarking decay parameters". New Journal of Physics. 20 (9): 092001.
\bibitem{RB8} Boone, Kristine; Dugas, Arnaud; Wallman, Joel; Emerson, Joseph (2019). "Randomized benchmarking under different gate sets". Physical Review A. 99 (3): 032329.
\bibitem{RB9} Wallman, Joel; Granade, Chris; Harper, Robin; Flammia, Steven (2015). "Estimating the coherence of noise". New Journal of Physics. 17 (11): 113020.
\bibitem{RB10} Gambetta, Jay M.; Corcoles, A.D.; Merkel, Seth T.; Johnson, Blake R.; Smolin, John A.; Chow, Jerry M.; Ryan, Colm A.; Rigetti, Chad; Poletto, Stefano; Ohki, Thomas A.; Ketchen, Mark B.; Steffen, Matthias (2012). "Characterization of Addressability by Simultaneous Randomized Benchmarking". Physical Review Letters. 109 (31–9007): 240504.
\bibitem{RB11} Quantum Circuits for Exact Unitary t-Designs and Applications to Higher-Order Randomized Benchmarking, Yoshifumi Nakata, Da Zhao, Takayuki Okuda, Eiichi Bannai, Yasunari Suzuki, Shiro Tamiya, Kentaro Heya, Zhiguang Yan, Kun Zuo, Shuhei Tamate, Yutaka Tabuchi, and Yasunobu Nakamura, PRX Quantum 2, 030339 (2021)
\bibitem{DiVincenzo2002} D. P. DiVincenzo, D. W. Leung, and B. M. Terhal. Quantum data hiding. IEEE Trans. Inf.
Theory, 48(3):580–598, 2002.
\bibitem{DP1}O. Szehr, F. Dupuis, M. Tomamichel, and R. Renner. Decoupling with unitary approximate
two-designs. New Journal of Physics, 15(5):053022, 2013.
\bibitem{DP2} Decoupling with random diagonal unitaries, Yoshifumi Nakata, Christoph Hirche, Ciara Morgan, and Andreas Winter, Quantum 1, 18 (2017).
\bibitem{BH1}  P. Hayden and J. Preskill. Black holes as mirrors: quantum information in random subsystems. J. High
Energy Phys., 2007(09):120, 200
 \bibitem{BH2} Y. Sekino and L. Susskind. Fast scramblers. J. High Energy Phys., 2008(10):065, 2008.
\bibitem{BH3} N. Lashkari, D. Stanford, M. Hastings, T. Osborne, and P. Hayden. Towards the fast scrambling
conjecture. J. High Energy Phys., 2013(4), 2013
\bibitem{TH1}  S. Popescu, A. J. Short, and A. Winter. Entanglement and the foundations of statistical mechanics. Nat. Phys., 2(11):754–758, 2006.
\bibitem{TH2} S. Goldstein, J. L. Lebowitz, R. Tumulka, and N. Zangh´ı. Canonical Typicality. Phys. Rev. Lett., 96(5):050403, 2006.
\bibitem{TH3} P. Reimann. Foundation of Statistical Mechanics under Experimentally Realistic Conditions. Phys. Rev. Lett., 101(19):190403, 2008
\bibitem{TR1} O. T. O’Meara, Symplectic groups, ser. Mathematical Surveys. American Mathematical Society, Providence, R.I., 1978, vol. 16.
\bibitem{TR2}  D. Callan, “The generation of Sp(F2) by transvections,” J. Algebra, vol. 42, no. 2, pp. 378–390, 1976.
\bibitem{Tan20} X. Tan, N. Rengaswamy and R. Calderbank. Approximate Unitary 3-Designs from Transvection Markov Chains. quant-ph/2011.00128
\bibitem{Koenig14} R. Koenig and J. A. Smolin, "How to efficiently select an arbitrary Clifford group element", J. Math. Phys., vol. 55, no. 12, Dec. 2014.
\bibitem{Pllaha21} Tefjol Pllaha, Kalle Volanto and Olav Tirkkonen. Decomposition of Clifford Gates. quant-ph/2102.11380
\bibitem{misc_transvections1} T. Can, N. Rengaswamy, R. Calderbank, and H. D. Pfister, “Kerdock Codes Determine
Unitary 2-Designs,” IEEE Trans. Inform. Theory, vol. 66, no. 10, pp. 6104–6120, 2020,
\bibitem{misc_transvections2} N. Rengaswamy, R. Calderbank, S. Kadhe, and H. D. Pfister, “Logical Clifford synthesis
for stabilizer codes,” IEEE Trans. Quantum Engg., vol. 1, 2020.
\bibitem{misc_transvections3} N. Rengaswamy, R. Calderbank, S. Kadhe, and H. D. Pfister, “Synthesis of Logical Clifford Operators via Symplectic Geometry,” in Proc. IEEE Int. Symp. Inform. Theory,
Jun 2018, pp. 791–795.
\bibitem{ion_trap1} N. Rengaswamy, R. Calderbank, S. Kadhe and H. D. Pfister, "Logical Clifford Synthesis for Stabilizer Codes," in IEEE Transactions on Quantum Engineering, vol. 1, pp. 1-17, 2020, Art no. 2501217, doi: 10.1109/TQE.2020.3023419.
\bibitem{ion_trap2} ] Sørensen A and Mølmer K 1999 Quantum computation with ions in thermal motion Phys. Rev. Lett. 82 1971
\bibitem{ion_trap3} Debnath S, Linke N M, Figgatt C, Landsman K A, Wright K and Monroe C 2016 Demonstration of a small programmable
quantum computer with atomic qubits Nature 536 63–6
\bibitem{ion_trap4} Figgatt C, Ostrander A, Linke N M, Landsman K A, Zhu D, Maslov D and Monroe C 2019 Parallel entangling operations on a
universal ion-trap quantum computer Nature 572 368–72
\bibitem{ion_trap5}  Grzesiak N et. al. Efficient arbitrary simultaneously entangling gates on a trapped-ion quantum computer Nat. Commun. 11
2963 (2020)
\bibitem{ion_trap6} Martinez E A, Monz T, Nigg D, Schindler P and Blatt R 2016 Compiling quantum algorithms for architectures with multi-qubit
gates New J. Phys. 18 063029
\bibitem{Diaconis2003} Diaconis, P. (2003). Random walks on groups: Characters and geometry. In C. Campbell, E. Robertson, \& G. Smith (Eds.), Groups St Andrews 2001 in Oxford (London Mathematical Society Lecture Note Series, pp. 120-142). Cambridge: Cambridge University Press. doi:10.1017/CBO9780511542770.017
\bibitem{Aaronson2008} Scott Aaronson and Daniel Gottesman
Phys. Rev. A 70, 052328 (2004). 
\bibitem{Bravyi21} Sergey Bravyi, Ruslan Shaydulin, Shaohan Hu, Dmitri Maslov, Clifford Circuit Optimization with Templates and Symbolic Pauli Gates,	Quantum 5, 580 (2021)
\bibitem{Gidney2021} C. Gidney, Stim: a fast stabilizer circuit simulator, Quantum 5, 497 (2021)
\bibitem{Bravyi_Maslov_21} S. Bravyi and D. Maslov, "Hadamard-Free Circuits Expose the Structure of the Clifford Group," IEEE Transactions on Information Theory, vol. 67, no. 7, pp. 4546-4563, July 2021, doi: 10.1109/TIT.2021.3081415.

\bibitem{Helsen2016} "Representations of the multi-qubit Clifford group".
Journal of Mathematical Physics 59, 072201 (2018); https://doi.org/10.1063/1.4997688
 Jonas Helsen,  Joel J. Wallman, and Stephanie Wehner.
\bibitem{RBCL1} Multiqubit randomized benchmarking using few samples
Jonas Helsen, Joel J. Wallman, Steven T. Flammia, and Stephanie Wehner
Phys. Rev. A 100, 032304 (2019)
\bibitem{RBCL2} Efficient unitarity randomized benchmarking of few-qubit Clifford gates, Bas Dirkse, Jonas Helsen, and Stephanie Wehner
Phys. Rev. A 99, 012315 – Published 10 January 2019
\bibitem{Zhu16} The Clifford group fails gracefully to be a unitary 4-design
Huangjun Zhu, Richard Kueng, Markus Grassl, David Gross, https://arxiv.org/abs/1609.08172
\bibitem{Hahn2019} Hahn, H., Zarantonello, G., Schulte, M. et al. Integrated 9Be+ multi-qubit gate device for the ion-trap quantum computer. npj Quantum Inf 5, 70 (2019). https://doi.org/10.1038/s41534-019-0184-5
\bibitem{Shapira2020} Theory of robust multiqubit nonadiabatic gates for trapped ions
Yotam Shapira, Ravid Shaniv, Tom Manovitz, Nitzan Akerman, Lee Peleg, Lior Gazit, Roee Ozeri, and Ady Stern, Phys. Rev. A 101, 032330 – Published 20 March 2020
\bibitem{Rasmussen2020} Single-step implementation of high-fidelity n-bit Toffoli gates. S. E. Rasmussen, K. Groenland, R. Gerritsma, K. Schoutens, and N. T. Zinner. Phys. Rev. A 101, 022308 – Published 7 February 2020
\bibitem{Goel2021} Native multiqubit Toffoli gates on ion trap quantum computers, Nilesh Goel, J. K. Freericks, arXiv:2103.00593
\bibitem{Gross17} David Gross, Sepehr Nezami, and Michael Walter, Schur-Weyl duality for the Clifford group with applications: Property testing, a robust Hudson theorem, and de Finetti representations, arXiv preprint arXiv:1712.08628 (2017).
\bibitem{Appleby2011} D.M. Appleby, Ingemar Bengtsson, Stephen Brierley, Markus Grassl, David Gross, Jan-Ake Larsson, The monomial representations of the Clifford group, Quantum Information and Computation, Vol. 12 No. 5 \& 6, 0404-0431 (2012)
\bibitem{Gross21} Gross, D., Nezami, S. \& Walter, M. Schur–Weyl Duality for the Clifford Group with Applications: Property Testing, a Robust Hudson Theorem, and de Finetti Representations. Commun. Math. Phys. 385, 1325–1393 (2021). https://doi.org/10.1007/s00220-021-04118-7
 \bibitem{Zhu15} H. Zhu.  Multiqubit Clifford groups are unitary 3-designs.
Phys. Rev. A 96, 062336, 2017
\bibitem{Webb15} Z. Webb, The Clifford group forms a unitary 3-design,
Quantum Inf. Comput. 16, 1379–1400 (2016)
\bibitem{Chau05} H. F. Chau. Unconditionally secure key distribution in higher dimensions by depolarization.
IEEE Trans. Inf. Theory, 51(4):1451–1468, 2005. arXiv:quant-ph/0405016.
\bibitem{Cleve16} R. Cleve, D. Leung, L. Liu, and C. Wang. Near-linear constructions of exact unitary 2-designs. Quant.
Info. \& Comp., 16(9 \& 10):0721–0756, 2016 
\bibitem{Schur}  Issai Schur (1905), "Neue Begründung der Theorie der Gruppencharaktere" (New foundation for the theory of group characters), Sitzungsberichte der Königlich Preußischen Akademie der Wissenschaften zu Berlin, pages 406-432
\bibitem{Gross2006} D. Gross, K. Audenaert, and J. Eisert , "Evenly distributed unitaries: On the structure of unitary designs", \href{https://doi.org/10.1063/1.2716992}{Journal of Mathematical Physics 48, 052104 (2007)},  \href{https://arxiv.org/abs/quant-ph/0611002}{Arxiv link: quant-ph/0611002}
\bibitem{Gross2005} D. Gross. Diploma Thesis. University of Potsdam (2005).
Available online at http://gross.qipc.org
\bibitem{Dankert2009} Dankert, Christoph; Cleve, Richard; Emerson, Joseph; Livine, Etera (2009). "Exact and Approximate Unitary 2-Designs: Constructions and Applications". Physical Review A. 80: 012304.
\bibitem{wiki_group_ring} Wikipedia page on 'Group Ring'. Available at \href{https://en.wikipedia.org/wiki/Group_ring}{this} link.
\bibitem{Haferkamp2020} Jonas Haferkamp, Felipe Montealegre-Mora, Markus Heinrich, Jens Eisert, David Gross, Ingo Roth, Quantum homeopathy works: Efficient unitary designs with a system-size independent number of non-Clifford gates, arXiv:2002.09524

\bibitem{Pitt_Coste_unpublished} Pittet Ch., Coste L. S., A survey on the relationships between volume growth, isoperimetry, and the behavior of simple random walk on Cayley graphs (unpublished). Pdf is available for download from \href{https://pi.math.cornell.edu/~lsc/avpub.html}{this} webpage.

\end{thebibliography}
\end{document}